%% file: match.tex
\documentclass[UKenglish]{lipics-v2021}
\PassOptionsToPackage{dvipsnames}{xcolor}

\newif\ifarxiv
\arxivtrue

\usepackage{sl}
\usepackage{cancel}
\usepackage[utf8]{inputenc}
\usepackage{enumerate}
\usepackage[longend,boxed]{algorithm2e}
\usepackage{dashrule}
\usepackage{bussproofs}
\usepackage{tikz}
\usetikzlibrary{shapes.geometric, arrows, patterns, positioning}
\definecolor{redblue}{rgb}{1,0,1}
\usepackage{scalerel}
\usepackage{enumitem}
\usepackage{pifont}
\usepackage{multirow}
\usepackage{tabularx}
\usepackage{todonotes}
\usepackage{booktabs}

\usepackage{graphicx}
\usepackage[framemethod=TikZ]{mdframed}
\usetikzlibrary{shapes.geometric, arrows}
\usetikzlibrary{decorations.pathmorphing, decorations.pathreplacing, decorations.shapes, arrows.meta}

\newcommand\Ccancel[2][black]{\renewcommand\CancelColor{\color{#1}}\cancel{#2}}

\usepackage{macros}
\usepackage{c_macros}
\usepackage{js_macros}

\usepackage{xspace}
\usepackage{textcomp}

\usepackage{listings}
\usepackage{lstautogobble}
\definecolor{bluekeywords}{rgb}{0.13, 0.13, 1}
\definecolor{greencomments}{rgb}{0, 0.5, 0}
\definecolor{redstrings}{rgb}{0.9, 0, 0}
\definecolor{graynumbers}{rgb}{0.5, 0.5, 0.5}
\lstdefinelanguage{wisl}{
   morekeywords={predicate, +, lemma, proof, apply, statement,
                 forall, if, assert, bind, unfold, function, return, ret},
   sensitive=false,
   comment=[l]{//}
}

\newtheorem{property}{Property}

\nolinenumbers

\title{Compositional Symbolic Execution for Correctness and Incorrectness Reasoning \ifarxiv (Extended Version) \else \fi} 
\titlerunning{Compositional Symbolic Execution for Correctness and Incorrectness Reasoning} 

\author{Andreas Lööw}{Imperial College London, United Kingdom}{}{}{}
\author{Daniele Nantes-Sobrinho}{Imperial College London, United Kingdom}{}{}{}
\author{Sacha-Élie Ayoun}{Imperial College London, United Kingdom}{}{}{}
\author{Caroline Cronjäger}{Ruhr-Universität Bochum, Germany}{}{}{}
\author{Petar Maksimović}{Imperial College London, United Kingdom \and Runtime Verification Inc., United States of America}{}{}{}
\author{Philippa Gardner}{Imperial College London, United Kingdom}{}{}{}

\authorrunning{A. Lööw et al.} 

\Copyright{Andreas Lööw, Daniele Nantes-Sobrinho, Sacha-Élie Ayoun, Caroline Cronjäger, Petar Maksimović, and Philippa~Gardner} 



\supplement{}
\supplementdetails[subcategory={}, cite={}, swhid={}]{Software}{https://doi.org/10.4230/DARTS.10.2.13}

\funding{This work was supported by the EPSRC Fellowship ``VetSpec: Verified Trustworthy Software Specification'' (EP/R034567/1).}

\acknowledgements{We would like to thank Nat Karmios for help with preparing the artefact for this paper. We would also like to thank the anonymous reviewers for their comments.}

\EventEditors{Jonathan Aldrich and Guido Salvaneschi}
\EventNoEds{2}
\EventLongTitle{38th European Conference on Object-Oriented Programming (ECOOP 2024)}
\EventShortTitle{ECOOP 2024}
\EventAcronym{ECOOP}
\EventYear{2024}
\EventDate{September 16--20, 2024}
\EventLocation{Vienna, Austria}
\EventLogo{}
\SeriesVolume{313}
\ArticleNo{25}

\begin{document}

\maketitle




\begin{abstract}
The introduction of separation logic has led to the development of symbolic execution techniques and tools that are (functionally) compositional with function specifications that can be used in broader calling contexts. Many of the compositional symbolic execution tools developed in academia and industry have been grounded on a formal foundation, but either the function specifications are not validated with respect to the underlying separation logic of the theory, or there is a large gulf between the theory and the implementation of the tool.

We introduce a formal compositional symbolic execution engine which
creates and uses function specifications from an underlying separation
logic and provides a sound theoretical foundation for, and indeed was
partially inspired by, the Gillian symbolic execution platform. This
is achieved by providing an axiomatic interface which describes the
properties of the consume and produce operations used in the engine to
update compositionally  the symbolic state, for example when calling
function specifications. This consume-produce technique is used by
VeriFast, Viper, and Gillian,  but has not been previously characterised independently of the tool. As part of our result, we give consume and produce operations inspired by the Gillian implementation that satisfy the properties described by our axiomatic interface. A surprising property is that our engine semantics provides a common foundation for both correctness and incorrectness reasoning, with the difference in the underlying engine only amounting to the choice to use satisfiability or validity. We use this property to extend the Gillian platform, which previously only supported correctness reasoning, with incorrectness reasoning and automatic true bug-finding using incorrectness bi-abduction. We evaluate our new Gillian platform by using the Gillian instantiation to C. This instantiation is the first tool grounded on a common formal compositional symbolic execution~engine~to~support~both~correctness~and~incorrectness~reasoning.
\end{abstract}

\input{sections/introduction}
\input{sections/overview}

\input{sections/language}

\input{sections/cse}
\input{sections/func-interface}

\input{sections/func-implementation}
\input{sections/bi-abduction}
\input{sections/applications}

\input{sections/evaluation}
\input{sections/related-work}
\input{sections/conclusion}

\bibliographystyle{plainurl}
\bibliography{match.bib}



\ifarxiv
\appendix
\counterwithin{theorem}{section}

\newpage
\input{sections/app-concrete}

\newpage
\input{sections/app-symbolic}
\newpage
\input{sections/app-symbolic-soundness}
\newpage
\input{sections/app-symbolic-soundness-func}
\newpage
\input{sections/app-macproduce}
\newpage
\input{sections/app-biabduction}
\newpage
\input{sections/app-applications-proofs}
\newpage
\input{sections/app-gillian}
\else
\fi
\end{document}

%% file: sections/introduction.tex

\section{Introduction}


One of the main challenges that  modern program analysis tools based on static symbolic execution~\cite{symb:exec:survey} must face  is {\em scalability}, that is, the ability to tractably analyse large, dynamically changing codebases. 
Such scalability requires symbolic techniques and tools that are {\em functionally compositional} (or simply compositional) where the analysis works on functions in isolation, at any point in the codebase, and then records the results in simple function specifications that can be used in broader calling contexts.
However, the traditional symbolic execution tools and frameworks based on first-order logic, such as CBMC~\cite{cbmc} and Rosette~\cite{torlak:pldi:2014}, can only be compositional for functions that manipulate the variable store, but not for functions that manipulate the heap, limiting their usability.

A key insight is that, to obtain compositionality,  the analysis should work with function specifications that are  {\em local}, in that they should describe the function behaviour only on the {\em partial} states or resources that the function accesses or manipulates, and a mechanism for using such specifications when the function is called by code working on a larger state. 
This insight was first introduced in separation logic (SL)~\cite{seplogic,reyseplogic}, a modern over-approximating (OX) program logic for compositional verification of correctness properties, which features local function specifications that can be called on larger state with the help of the {\em frame rule}. Recently, these ideas have been adapted to under-approximate (UX) reasoning in the context of incorrectness separation logic (ISL)~\cite{isl} for compositional {true bug-finding}.


Ideas from  SL and ISL have led to the development of efficient compositional symbolic execution tools in academia and industry, such as the tool VeriFast~\cite{verifast} and the multi-language platforms Viper~\cite{Muller16} and Gillian~\cite{gilliancav} for semi-automatic OX verification based on SL, and Meta's multi-language platform Infer-Pulse~\cite{Le22} for automatic UX true bug-finding based on ISL. However,  there are issues with the associated formalisms of the tools: 
either the function specifications created and used by the tool are not validated with respect to the underlying separation logic; 
or there is a large gulf between the formalism and the implementation of the tool.
VeriFast, Viper, and Gillian, on the one hand,  all come with a sound compositional symbolic operational semantics that closely model the tools.
These tools handle the creation and use of function specifications (and the folding and unfolding of predicates) using two operations, called {\em consume} and {\em produce}, which, respectively, removes from and adds to a given symbolic state the symbolic state corresponding  to a given assertion. The formalisms  accompanying the  tools, however, do not properly connect their function specifications to the underlying separation logics. Thus, function specifications created by the tools cannot soundly be used by other tools, and vice versa. On the other hand, the formalism of  Infer-Pulse describes  compositional symbolic execution 
as proof search in ISL,  and similarly with its SL-predecessor Infer~\cite{calcagno:nasa:2011}, 
thereby making the connection
to its separation logic direct. However, the gap between the formalism and the tool is considerable.




In this paper, inspired by the Gillian platform~\cite{gillianpldi, gilliancav}, we formally define a compositional symbolic execution (CSE) engine that provides a sound theoretical foundation for building compositional OX and UX analysis 
tools. Our engine is  described by  a compositional symbolic operational semantics using 
the consume and produce operations to interface with
 function specifications  valid in either SL or ISL. 
A surprising property of our semantics is that it is simple to switch between OX and UX reasoning. In more detail, our CSE engine features the following theoretical contributions:
\begin{enumerate}
\item {\em specification interoperability}, in the sense that it can both create and use function specifications validated in an underlying separation logic, allowing for a mix-and-match of specifications validated in various ways from various sources;
\item {\em an axiomatic approach to compositionality}, in that we provide an axiomatic description of  properties  for the consume and produce operations, which have not been previously characterised axiomatically;
\item {\em a general soundness result}, which states that, assuming the axiomatic properties of the consume and produce operations and the validity of function specifications that are used with respect to the underlying separation logic, the CSE engine is sound and creates function specifications that are valid with respect to the said  logic;
\item {\em a unified semantics}, which captures the essence of {\em both} the OX reasoning of VeriFast, Viper, and Gillian, and the UX reasoning in Infer-Pulse, with the difference amounting {\em only} to the choice of using satisfiability or validity, and different axiomatic properties of the consume and produce operations. 
\end{enumerate}

We instantiate our general soundness result by giving example implementations of the consume and produce operations,  inspired by those found in  Gillian, which we prove satisfy the properties laid out by the axiomatic interface. 
For clarity of presentation, although both our CSE engine and our consume and produce operations are inspired by Gillian, we have opted to work with a fixed, linear memory model and a simple while language instead of the parametric memory model approach of Gillian and its goto-based intermediate language~GIL. The move from a fixed to a parametric memory model is straightforward and planned~future~work.

In addition, this paper brings the following practical contributions:
\begin{enumerate}
\item a demonstrator Haskell implementation of our CSE engine with  example implementations of the consume and produce operations;
\item an extension of the Gillian platform with automatic compositional UX true bug-finding using UX  bi-abduction in the style of Infer-Pulse, making Gillian the first unified tool for OX and UX compositional reasoning about real-world code.
\end{enumerate}


The Gillian platform already supported whole-program symbolic testing as found in, for example, CBMC, and semi-automatic OX compositional verification underpinned by SL  as in, for example, VeriFast.
Because our CSE engine  has pinpointed the small differences required for the switch between  OX and UX reasoning, we were able to simply extend Gillian with automatic compositional UX true-bug finding without affecting its other analyses.
Interestingly, UX true bug-finding has not been implemented in a consume-produce engine before. We demonstrate the additional UX reasoning  by extending the CSE engine with UX bi-abductive reasoning~\cite{calcagno:popl:2009,calcagno:jacm:2011,isl,Le22}, an automatic technique which has enabled compositional reasoning to scale to industry-grade codebases, and which works by generating function specifications from their implementations by incrementally constructing the~calling~context. We implement this technique following the approach pioneered by OX tool JaVerT~2.0~\cite{javert2}, where missing-resource errors are used to generate fixes that drive the specification construction. We evaluate this extension of Gillian using its Gillian-C instantiation, on a real-world Collections-C data-structure library~\cite{collections}, obtaining promising initial results and performance.

%% file: sections/overview.tex

\section{Overview: Compositional Symbolic Execution}%
\label{sec:overview}

We give an overview of our CSE engine, together with example  analyses that we show can be hosted on top of this engine.
Our CSE engine
consists of three engines built on top of each other, labelled by
different reasoning modes, OX and  UX, that  are appropriate for different types of analyses.
%
In short, the reasoning modes can be characterised as follows:
\begin{table}[h!]
\centering
\begin{tabular}{llll}
\toprule
\textbf{Mode} & \textbf{Guarantee} & \textbf{Consequence rule}
  & \textbf{Analysis}\\ \midrule
  OX & Full path coverage & Forward logical consequence & Full verification in SL \\
UX & Path reachability & Backward
                                                    logical  consequence &
                                                                   True
                                                                   bug-finding in ISL \\
\end{tabular}
\end{table}
%

The \emph{core engine} (\S\ref{sec:language-symbolic}) is a simple
symbolic execution engine for our  demonstrator programming language
(\S\ref{sec:language-concrete}). It  does not use assertions to
update symbolic state and is hence compatible with both OX and UX reasoning.
It is sufficient to capture whole-program symbolic
testing found in, e.g., CMBC and Gillian. 
The \emph{compositional engine}
(\S\ref{sec:function-interface},~\S\ref{sec:function-implementation})
is built on top of this core engine. It 
can switch between either the OX or UX mode of reasoning, providing 
support for 
the use of SL and ISL function specifications 
by extending the core engine with consume and
produce operations for updating the symbolic state. In OX mode, it captures the  full verification found in e.g. VeriFast
and Gillian, soundly underpinned by~SL. For the first time, in UX mode
it also captures 
ISL analysis, not previously found in a symbolic execution tool. We
demonstrate this by building 
the  UX \emph{bi-abductive engine}~(\S\ref{sec:bi-abduction})  on 
top of the UX  compositional engine to 
automatically  fix 
missing-resource errors (e.g.,~a~missing~heap~cell) using the UX
bi-abductive technique from Infer-Pulse, 
to capture automatic 
true bug-finding underpinned by ISL.




\subsection{Core Engine}\label{ki:core}

The core symbolic execution engine provides a foundation on top of which the other components are built. It is essentially a  standard symbolic execution engine that is slightly adapted to handle both usual language errors and the missing-resource errors, which can occur now that we are working with partial state. 

Our engine operates over partial  symbolic states $\sst = (\ssto, \smem, \spc)$ comprising: a symbolic variable store~$\ssto$, holding symbolic values for the  program variables; a partial symbolic heap~$\smem$, representing the memory on which programs operate; and a symbolic path condition~$\spc$, holding constraints accumulated during symbolic execution. We work with a simple demonstrator programming language (cf.~\S\ref{sec:language-concrete}) and linear heaps:  that is, partial-finite maps from natural numbers to values. The core engine is both OX- and UX-sound, also referred to as exact~(EX)~\cite{esl}, as established by~Thm.~\ref{thm:ux-ox-sound}.

\subparagraph*{Example.} In Fig.~\ref{fig:f-exec} (left), we define a simple function $\mathtt{f}$. In Fig.~\ref{fig:f-exec} (middle), we illustrate its symbolic execution, which starts from
 $\sst = (\{ \pvar c \mapsto \svar c, \pvar x \mapsto \svar x, \pvar r \mapsto \nil \}, \emptyset, \true)$, 
 with the function parameters ($\pvar c$ and $\pvar x$) initialised with some symbolic variables ($\svar c$ and $\svar x$),
the local function variables ($\pvar r$) initialised to $\nil$, the heap 
set to empty and the path condition set to $\true$.
 Next, executing the if-statement with condition $\pvar c \ge 42$ yields
  three branches: one in which $\pvar c$ is not a natural number, in which 
  the execution fails with an evaluation error; one in which $\pvar c \geq 42$, in which the execution 
  continues; and one in which $\pvar c < 42$, in which the program throws a user-defined error. 
Next, executing the lookup $\pderef{\pvar{r}}{\pvar x}$ results in two more branches: one in which $\pvar x$ is not a heap address (natural number), yielding a type error 
and one in which $\pvar x$ is a heap address. In that branch, the lookup fails with a missing-resource error as the heap is empty.

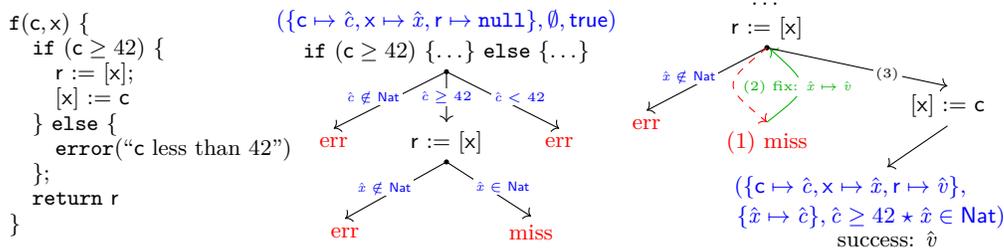
\begin{figure}[!t]
\begin{minipage}{0.25\textwidth}
\footnotesize
\(
\begin{array}{l}
\mathtt{f}(\pvar c, \pvar x)~\{ \\
\tab \texttt{if}~( \pvar{c} \ge 42 )~\{ \\ 
\tab\tab \pderef{\pvar{r}}{\pvar x}; \\
\tab\tab \pmutate{\pvar{x}}{\pvar c} \\
\tab \}~\mathtt{else}~\{ \\ 
\tab\tab \perror(\strlit{$\pvar c$ less than 42}) \\
\tab \}; \\ 
\tab\preturn{\pvar{r}} \\
\}
\end{array}
\)
\end{minipage}
\begin{minipage}{0.33\textwidth}
\small
  \begin{tikzpicture}[scale=.75]
 \node (init) at (3, 4)  {\footnotesize \texttt{if} $(\pvar{c} \geq 42)$ $\{ \dots\!\: \}$ \texttt{else} $\{ \dots\!\: \}$}; \filldraw (init.south) circle (1pt);
 \node [blue,yshift=4] (spec) at (init.north) {\footnotesize $(\{ \pvar c \mapsto \svar c, \pvar x \mapsto \svar x, \pvar r \mapsto \nil \},  \emptyset, \true)$};
  \node (firstlevelb) at (3, 2.4)  {\footnotesize $\pderef{\pvar{r}}{\pvar{x}}$}; \filldraw (firstlevelb.south) circle (1pt); 
  \node[red] (firstlevela) at (1, 2.4) {\footnotesize  err};
  \node[red] (firstlevelc) at (5, 2.4)  {\footnotesize  err}; 
  \draw[->] (init.south) to (firstlevelb.north);
  \draw[->] (init.south) to (firstlevela.north);
  \draw[->] (init.south) to (firstlevelc.north);
  \node[blue,fill=white, inner sep=1, outer sep=0] (name)  at (1.7, 3.2) {\tiny{$\svar{c}\notin \Nat$}}; 
  \node[blue,fill=white, inner sep=1, outer sep=0] (name1)  at (4.3, 3.2) {\tiny{$\svar{c}<42$}}; 
  \node[blue,fill=white, inner sep=1, outer sep=0] (name2)  at (3, 3.2) {\tiny{$\svar{c}\geq 42$}}; 
  \node[red] (outcome2) at (1.2, 0.8) {\footnotesize  err};
  \draw[->] (firstlevelb.south) to (outcome2.north);
  \node[blue,fill=white, inner sep=1, outer sep=0] (name3)  at (1.9, 1.6) {\tiny{$\svar{x}\notin \Nat$}};
  \node[red] (missing) at (4.5,0.8) {\footnotesize miss};
  \draw[->] (firstlevelb.south) -- (missing.north);
  \node[blue,fill=white, inner sep=1, outer sep=0] (name4) at (4,1.6){\tiny{$\svar{x}\in\Nat$}};
\end{tikzpicture}
\end{minipage}
\begin{minipage}{0.39\textwidth}
\small
  \begin{tikzpicture}[scale=.8]
  \node at (3, 2.8)  {\footnotesize \dots};
  \node (firstlevelb) at (3, 2.4)  {\footnotesize $\pderef{\pvar{r}}{\pvar{x}}$}; \filldraw (firstlevelb.south) circle (1pt); 
%
  \node[red] (outcome2) at (1, 0.8){\footnotesize err};
  \draw[->] (firstlevelb.south) to (outcome2.north);
  \node[blue,fill=white, inner sep=1, outer sep=0] (name3)  at (1.7, 1.6) {\tiny{$\svar{x}\notin \Nat$}};

  \node[red] (missing) at (3,.5) {\footnotesize (1) miss};
  \draw[->, dashed, red] (firstlevelb.south) .. controls (2.25,1.45) .. (missing.north);

     \draw[->, color= green!70!black]
             (missing.north) .. controls (3.75, 1.42) ..  (3.05,2.05);
      \node[color= green!70!black,fill=white, inner sep=1, outer sep=0] (biabduce) at (3.5,1.42) {\tiny(2) fix: $\svar x \mapsto \svar v$};

      \node (command) at (6,1.1) {\footnotesize $\pmutate{\pvar{x}}{\pvar{c}}$};
      \draw[->] (firstlevelb.south) to (command.north);
      \node[fill=white, inner sep=1, outer sep=0] at (5,1.63) {\tiny (3)};
      \node [blue,text width=4cm] (post) at (5,-0.5) {\footnotesize $(\{ \pvar c \mapsto \svar c, \pvar x \mapsto \svar x, \pvar r \mapsto \svar v \}, $ \\ $\{ \svar x \mapsto \svar c \},  \svar c \geq 42 \lstar \svar x \in \Nat)$};      
     \node[yshift=-14] (outcome3) at (post) {\footnotesize success: $\svar v$};
      \draw[->] (command.south) to (post.north);
\end{tikzpicture}
\end{minipage}
\caption{Definition and symbolic execution of function \texttt{f}}%
\label{fig:f-exec}
\end{figure}

\subparagraph*{Analysis: EX Whole-program Symbolic Testing.} 
The core engine can be used to perform whole-program symbolic testing in the style of CBMC~\cite{cbmc} and Gillian~\cite{gillianpldi}, in which the user creates symbolic variables, imposes some initial constraints on them, runs the symbolic execution to completion, and asserts that some final constraints hold.

\subsection{Compositional Engine}\label{ki:cse}


Our compositional engine extends the core engine to support calling, in its respective OX and UX mode, SL and ISL function specifications, denoted {\small $\quadruple{\vec{\pvar{x}} = \vec{x} \lstar P}{{\sf f}(\vec{\pvar{x}})}{\Qok}{\Qerr}$} and {\small $\islquadruple{\vec{\pvar{x}} = \vec{x} \lstar P}{{\sf f}(\vec{\pvar{x}})}{\Qok}{\Qerr}$},\footnote{UX quadruples can be split into two triples, but OX quadruples cannot. To unify our presentation, we consider both types of specifications in quadruple~form.} respectively, where $P$ is a pre-condition and $\Qok$ and $\Qerr$ are success and error post-conditions. A SL specification gives an OX description of the function behaviour whereas an ISL specification~gives~a~UX~description:
\begin{description}
\item[(SL)] All terminating executions of the function ${\sf f}$ starting in a state satisfying $P$ either end successfully in  a state that satisfies $\Qok$ or fault in a state that satisfies $\Qerr$.
\item[(ISL)] Any state satisfying either the success $\Qok$ or error post-conditions $\Qerr$ is reachable from some state satisfying the pre-condition $P$ by executing the function ${\sf f}$.
\end{description}

The engine also adds support, in both OX and UX mode, for folding and unfolding of user-defined predicates, describing inductive data-structures such as linked lists.


In both cases, the call to function specifications and the folding and unfolding of predicates are implemented following the consume-produce engine style, underpinned by $\mathsf{consume}$ and $\mathsf{produce}$ operations, which, in essence, remove (consume) and add (produce) the symbolic state corresponding to a given assertion from and to the current symbolic state. 
For example, in Fig.~\ref{fig:ki-interface}, the symbolic execution is in a symbolic state $\sst$ and calls a function ${\sf f}(\vec{\pvar{x}})$ by its specification in ISL mode.  The (successful) function call is implemented by first consuming the symbolic state $\sst_{\mkern-3muP}$ corresponding to the pre-condition $P$, leaving the symbolic frame $\sst_{\mkern-3muf}$, and then producing into $\sst_{\mkern-3muf}$ the symbolic state $\sst_{\Qok}$ corresponding to the post-condition $\Qok$.

\begin{figure}[t]
\centering
\begin{tikzpicture}[scale=0.6]
\fill[gray!10!, dashed] (4.2,2) rectangle (7.4,8);
\draw (0.5,6.6) rectangle (2,7.4);
\draw (2,6.6) rectangle (3.6,7.4);
\node at (1.5,8) {{\tiny \(\sst=\sst_{\mkern-3muf}\cdot \sst_{\mkern-3muP}\)}};
\node (state1) at (1.3,7){\tiny\(\sst_{\mkern-3muf}\)};
\node (statep) at (2.7,7){\tiny\(\sst_{\mkern-3muP}\)};
\node at (2.5,5){{\tiny\(y:={\tt f}(\vec{\pvar x})\)}};
\draw (0.5,2.6) rectangle (2,3.4);
\draw (2,2.6) rectangle (3.6,3.4);
\node (statef) at (1.3,3){\tiny\(\sst_{\mkern-3muf}\)};
\node (stateq) at (2.7,3){\tiny\(\sst_{\Qok}\)};
\node at (5.8,2.3){\tiny Properties~\ref{prop:wf}-\ref{prop:prod-compl} (\S\ref{sec:function-interface})};
\node[right] (uxspec) at (7.2,7) {{\tiny\(\islquadruple{\vec{\pvar{x}} = \vec{x} \lstar \colorbox{teal!20!}{$P$}}{\mathtt{f}(\vec{\pvar{x}})}{\colorbox{magenta!10!}{$\Qok$}}{\Qerr}\)}};
\node[draw,text width=2.4cm,anchor=south,text centered,xshift=0.4cm] (interface) at (5,8.2) {\scriptsize Axiomatic Interface};
\draw[->, shorten >= 1pt, shorten <= 2pt]
             (9.7, 6.6) .. controls (4,5) and (3,6) .. (3.1,6.7);
\node at (8.2,5.4) {\tiny $\mac(\macUX,P,\ssubst,\sst)\rightsquigarrow (\ssubst',\sst_{\mkern-3muf})$}; 
\node at (12,3.5) {\tiny \(\produce(\Qok,\ssubst',\sst_{\mkern-3muf})\rightsquigarrow \sst_{\mkern-3muf}\cdot \sst_{\Qok}\)};
\draw[->, shorten >= 1pt,shorten <= 2pt]
             (12.5, 6.6) .. controls (12,2) and (3,5) .. (2.6,3.4);
\draw[thick,->] (state1) to (statef);       
\end{tikzpicture}
\caption{UX function-call rule: successful case}\label{fig:ki-interface}
\end{figure}
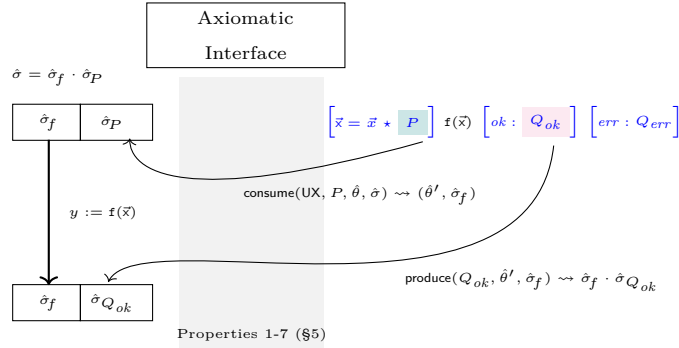

Our approach is novel in two ways: (1) we provide an axiomatic interface that captures the sufficient properties of the consume and produce operations for the engine to be sound; and (2) we provide example implementations (in the same style as the rest of the engine, that is, using inference rules) for the consume and produce operations that we prove satisfy the~axiomatic~interface. Moreover, our consume and produce operations switch their behaviour between the mode of reasoning (OX/UX), as described next.

\subparagraph*{Axiomatic Interface.}
We have identified sufficient properties for the consume and produce operations to be OX or UX sound (cf. Thm.~\ref{thm:ux-ox-sound-func}). 
Here we will describe a general idea of the \mac operation, the more complex of the two operations, the signature of which is: 
\[
  \mathsf{consume}(m, P, \ssubst, \sst) \rightsquigarrow  (\ssubst', \sst_{\mkern-3muf}) \mid
  \oxabort(\sval)
\]
The consume operation takes a mode $m$ (OX or UX), an assertion~$P$,  a substitution $\ssubst$, and a symbolic state $\sst$, where the substitution $\ssubst$ links known logical variables in $P$ to symbolic values in $\sst$. The operation finds which part of~$\sst$ could match $P$, resulting in potentially multiple successful or unsuccessful matches, and then, per match, returns either the pair $(\ssubst', \sst_{\mkern-3muf})$, which comprises a substitution $\ssubst'$  and a resulting symbolic state $\sst_{\mkern-3muf}$, or it aborts with error information $\hat v$. 
Some  properties the interface of  $\mathsf{consume}$ mandates are the following:
\begin{enumerate}
\item In successful consumption, there exists a symbolic state $\sst_{\mkern-3muP}$ such that $\sst = \sst_{\mkern-3muf} \cdot \sst_{\mkern-3muP}$ (where~$\cdot$~denotes state composition, which composes the corresponding components of the two states together) and that every concrete state described by $\sst_{\mkern-3muP}$ satisfies $P$. This tells us that the matched part of $\sst$ does correspond to $P$, that the effect of $\mathsf{consume}$ is its removal from $\sst$, and that consumption can be viewed as the frame inference problem~\cite{berdine:aplas:2005}, with the resulting state $\sst_{\mkern-3muf}$ constituting the frame;
\item In OX mode, $\mathsf{consume}$ does not drop paths;
in UX mode, it does not drop information.
\end{enumerate}
The interface allows, e.g., tool developers to design an OX-consume that (soundly) drops certain information deemed unneeded or a UX-consume that (soundly) drops paths according to a tool-specific criteria (e.g., as in UX bi-abduction in Infer-Pulse).

\subparagraph*{Example  Implementation: Consume.} 
We provide example implementations for the consume and produce operations  (\S\ref{sec:function-implementation})
that explore the similarities between OX and UX reasoning,  and allow us to maintain unified implementations across both reasoning modes. Our consume operation has a mode switch $m$, allowing for OX- or UX-specific behaviour, which we use to control the {\em only difference} in our implementation between the two modes: the consumption of pure (non-spatial) information (cf. Fig. \ref{fig:mac-pure}, left). For soundness, our implementation of the consume operation has to be compatible with the SL and ISL guarantees: in  OX mode, consume requires full path coverage, and in UX, it requires path reachability.

We illustrate our \mac implementation by example. 
Consider the symbolic state $\sst \defeq (\emptyset, \smem,  \spc)$, where $\smem \defeq \{ 1 \mapsto \svar v, 2 \mapsto 10, 3 \mapsto 100 \}$ and $\spc \defeq \svar x > 0 \land \svar v > 5$, and let us consume the assertion $P \defeq x \mapsto y \lstar y \geq 10$ from $\sst$ knowing that $\ssubst = \{ \svar x /x \}$, meaning that the logical variable $x$ of~$P$ is mapped to the symbolic variable $\svar x$ of $\sst$. This consumption is presented in the diagram below:

\begin{figure}[!h]
\small
\begin{tikzpicture}[scale=.9]
\node [anchor=east] (consume) at (0, 0) {$\mathsf{consume}(m, P, \ssubst, \sst)$};

\node [anchor=west] (out1) at (3, 1.2) {$\oxabort([\strlit{\sf consPure}, \svar v \geq 10, \spc \land {\color{blue} \svar x = 1}])$};

\node [anchor=west] (out2) at (3, 0.6) {$(\ssubst\uplus\{{\color{orange}  \svar v/y}\}, (\emptyset, \{ \Ccancel[red]{1 \mapsto \svar v}, 2 \mapsto 10, 3 \mapsto 100 \}, \spc \land {\color{blue} \svar x = 1} \land {\color{magenta}\svar v \ge 10}))$};

\node [anchor=west] (out3) at (3, 0) {$(\ssubst\uplus\{{\color{orange} 10/y }\}, (\emptyset, \{ 1 \mapsto \svar v, \Ccancel[red]{2 \mapsto 10}, 3 \mapsto 100 \},  \spc \land {\color{blue} \svar x = 2}))$};

\node [anchor=west] (out4) at (3, -0.6) {$(\ssubst\uplus \{{\color{orange}  100/y}\}, (\emptyset, \{ 1 \mapsto \svar v, 2 \mapsto 10, \Ccancel[red]{3 \mapsto 100} \}, \spc \land {\color{blue} \svar x = 3}))$};

\node [anchor=west] (out5) at (3, -1.2) {$\oxabort([\strlit{\sf MissingCell}, \svar x, \spc \land {\color{blue} \svar x \notin \dom(\smem)}])$};

\draw[->] (consume.east) to[bend left=18] (out1.west);
\draw[->] (consume.east) to[bend left=10] (out2.west);
\draw[->] (consume.east) to (out3.west);
\draw[->] (consume.east) to[bend right=10] (out4.west);
\draw[->] (consume.east) to[bend right=18] (out5.west);

\node [fill=white, inner sep=1, outer sep=0] at (1.7, 1.0) {\scriptsize 0: OX};
\node [fill=white, inner sep=1, outer sep=0] at (1.7, 0.5) {\scriptsize 1: UX};
\node [fill=white, inner sep=1, outer sep=0] at (1.7, 0.0) {\scriptsize 2: OX/UX};
\node [fill=white, inner sep=1, outer sep=0] at (1.7, -0.5) {\scriptsize 3: OX/UX};
\node [fill=white, inner sep=1, outer sep=0] at (1.7, -1.0) {\scriptsize 4: OX/UX};

\end{tikzpicture}
\end{figure}
\vspace{-2mm}
\begin{itemize}[leftmargin=*]
\item the arrows are labelled with the mode  $m$ of operation of  \mac, 
  being either only OX, or only UX, or OX/UX when the consumption has the same behaviour in both modes;
\item both our OX and UX consumption branch on all possible matches: in this case, the cell assertion $x \mapsto y$ can be matched to any of the three cells in the heap (branches 0--3), but it could also refer to a cell that is outside of the heap (branch 4);
\item when branching occurs, then the branching condition is added to the path condition of the resulting state (the constraints in \textcolor{blue}{blue}), ensuring information is not dropped;
\item the heap cell corresponding to $x \mapsto y$ is removed when matched successfully (branches 1, 2, 3), and in those cases we learn the value corresponding to $y$ (substitution extension in \textcolor{orange}{orange} where $\uplus$ denotes disjoint union);
\item for the heap cell $\{ 1 \mapsto \svar v \}$, our OX consumption (branch 0) must abort since the $\spc$ does not imply $y \geq 10$ when $y = \svar v$, whereas UX consumption (branch 1) can proceed by restricting the path condition (constraint in \textcolor{magenta}{magenta}), since dropping paths is sound in UX; this allows our UX consumption to successfully consume more assertions; OX consumption cannot do the same since that would render e.g. the function-call rule, which is implemented in terms of consume, unsound in OX mode;
\item our UX consumption could alternatively drop the missing-cell abort outcome (branch~4), however, some analyses, such as bi-abduction, have use for this error information so we~do~not~drop~it.
\end{itemize}
\ifarxiv
 The detailed computation of this example using our consume implementation is in Ex.~\ref{ex:consP_modes}.
\else
\fi

\subparagraph*{Analysis: Verification.}
We use our compositional engine to provide semi-automatic OX verification: that is, given a function implementation and an OX function specification, the engine checks if the implementation satisfies the specification. This analysis is semi-automatic in that the user may have to provide loop invariants as well as ghost commands for, e.g., predicate manipulation and lemma application. It is implemented in the standard way for consume-produce tools.

\subsection{Bi-abductive Engine}\label{ki:biabduction}

Bi-abduction is a technique that enables automatic compositional analysis by allowing incremental discovery of resources needed to execute a given piece of code. It was introduced in the OX verification setting~\cite{calcagno:popl:2009,calcagno:jacm:2011}, later forming the basis of the bug-finding tool Infer~\cite{calcagno:nasa:2011}, and was recently ported to the UX setting of true bug-finding, underpinning Infer-Pulse~\cite{Le22}.

Our UX bi-abduction advances the state of the art in two ways. Firstly, UX bi-abduction has thus far been intertwined with the proof search of the symbolic execution tool it has formulated for~\cite{isl,Le22}. Inspired by an approach developed in the OX tool JaVerT~2.0~\cite{javert2}, we 
add UX bi-abduction as a layer on top of CSE by generating fixes from missing-resource errors. This covers both missing-resource errors from the execution of the commands of the language (e.g., in heap lookup, the looked-up cell might not be in the heap) as well as invocations of $\mathsf{consume}$ (e.g., if the resource required by a function pre-condition is not in the heap). In more detail, when a missing-resource error occurs, a fix is generated and applied to the current symbolic state, allowing the execution to continue.
Secondly, our UX bi-abduction is able to reason about predicates, allowing us to synthesise and soundly use a broader range of function specifications from other formalisms and tools, in particular specifications that capture unbounded behaviour rather than bounded or single-path behaviour.

\subparagraph*{Analysis: Specification Synthesis and True Bug-finding.} 
We use bi-abduction to power automatic synthesis of UX function specifications, obtaining one specification per each constructed execution path. Such function specification synthesis forms the back-end of Pulse-style true bug-finding, where specifications describing erroneous executions, after appropriate filtering, can be reported as bugs. Given the guarantees of UX reasoning, any bug (represented by a synthesised erroneous function specification) found during this process will be a true~bug.



\subparagraph*{Example: Specification Synthesis.} We illustrate how bi-abduction can be used for the synthesis of UX function specifications, using again the simple function $\mathtt{f}(\pvar c, \pvar x)$ from Fig.~\ref{fig:f-exec} (left).
The first and the third branches of Fig.~\ref{fig:f-exec} (middle) yield the following specifications:
\[
\begin{array}{c}
\isltripleerr{\pvar c = c \lstar \pvar x = x}{\mathtt{f}(\pvar c, \pvar x)}{\pvar{err} = [\mathstr{\mathsf{ExprEval}}, \strlit{c $\ge$ 42}] \lstar c \notin \nats} \\[1mm]
\isltripleerr{\pvar c = c \lstar \pvar x = x}{\mathtt{f}(\pvar c, \pvar x)}{\pvar{err} = [\mathstr{\mathsf{Error}}, \strlit{$\pvar c$ less than 42}] \lstar c < 42}
\end{array}
\]
noting that information about local variables is discarded, the error value is returned in the dedicated program variable $\pvar{err}$, and symbolic variables are replaced~with~logical~variables.

Using bi-abduction, the second branch of Fig.~\ref{fig:f-exec} (middle) now becomes Fig.~\ref{fig:f-exec} (right). The second branch of Fig.~\ref{fig:f-exec} (middle) has one branch in which $\pvar x$ is not a heap address (natural number), yielding a type error and the following specification 
\[
\isltripleerr{\pvar c = c \lstar \pvar x = x}{\mathtt{f}(\pvar c, \pvar x)}{\pvar{err} = [\mathstr{\mathsf{Type}}, \strlit{x}, x, \mathstr{\mathsf{Nat}}] \lstar c \ge 42 \lstar x \notin \nats}
\]
and one branch in which $\pvar x$ is a heap address. In that branch, the lookup fails with a missing-resource error as the heap is empty, but in bi-abductive execution, that is, Fig.~\ref{fig:f-exec} (right), instead of failing we first generate the fix $\svar x \mapsto \svar v$, where $\svar v$ is a fresh symbolic variable, and then add it to the heap and re-execute the lookup, which now succeeds. The rest of the function is executed without branching or errors, the function terminates 
 and returns the value of $\pvar r$, which is $\svar v$. This branch results in the following specification:
\[
\isltripleok{\pvar c = c \lstar \pvar x = x \lstar {\color{red}x \mapsto v}}{\mathtt{f}(\pvar c, \pvar x)}{x \mapsto c \lstar c \ge 42 \lstar \pvar{ret} = v }
\]
which illustrates an essential principle of bi-abduction, which is to add the fixes applied during execution (also known as \emph{anti-frame}, highlighted in red) to the specification pre-condition.

%



\subparagraph*{Example: Specification Synthesis with Predicates.} To exemplify how predicates can be useful during specification synthesis, consider the following variant of the standard singly-linked list predicate: $\llist{x; \lstxs, \lstvs}$, where $x$ denotes the starting address of the list, and $\lstxs$ and $\lstvs$ denote node addresses and node values, respectively.\footnote{We use the semicolon notation for predicates to be consistent with the main text, where the notation is used for automation---for the purpose of this section, these semicolons can be read as commas.} Both addresses and values are exposed in the predicate to ensure that no information is lost when the predicate is folded, making it suitable for UX reasoning. The predicate is defined as~follows:
\[
\begin{array}{r@{~}c@{~}l}
\llist{x; \lstxs, \lstvs} & \defeq & (x\doteq\nil \lstar \lstxs \doteq \emplist \lstar \lstvs \doteq \emplist)~\lor \\
                          && (\exsts{v, x', \lstxs', \lstvs'} x \mapsto v, x' \lstar\llist{x'; \lstxs', \lstvs'} \lstar  \lstxs \doteq x \cons \lstxs' \lstar \lstvs \doteq v \cons \lstvs') 
\end{array}
\]
Further, consider the predicate $\listptr{x; \lstxs}$, which tells us that $x$ is the head of the list~$\lstxs$ if $\lstxs$ is not empty and $\nil$ otherwise, defined as
\[
\begin{array}{@{}c@{}}
\listptr{x; \lstxs} \defeq (\lstxs \doteq \emplist \lstar x \doteq \nil) \lor (\exsts{\lstxs'} \lstxs = x \cons \lstxs'),
\end{array}
\]
and the following specifications of two list-manipulating functions (e.g., proven using pen-and-paper), which capture the exact behaviour of inserting a value in the front of a list (\texttt{LInsert}) and swapping of the first two values in a list (\texttt{LSwapFirstTwo}):
\[
\begin{array}{l}
\!\!\!\!\isltripleok{\pvar x\doteq x \lstar \pvar v \doteq v \lstar \llist{x; \lstxs, \lstvs} }{\mathtt{LInsert}(\pvar x, \pvar v)}{\llist{\pvar{ret};\pvar{ret}{\cons}\lstxs, v{\cons}\lstvs} \lstar \listptr{x; \lstxs}} \\[1mm]
 \!\!\!\!{\color{blue} [\pvar x\doteq x \lstar \llist{x; \lstxs, \lstvs}]}~{\mathtt{LSwapFirstTwo}(\pvar x)}~{\color{blue} [\oxerr: \llist{x; \lstxs, \lstvs} \lstar \shade{|\lstvs| < 2} \lstar \pvar{err} \doteq \strlit{List~too~short!}]}
\end{array}
\]
Using these specifications, we can bi-abduce the following UX true-bug specification of client code calling these functions, where the discovered anti-frame is again highlighted in red, but this time contains a predicate:
\[
\begin{array}{l}
{\color{blue} [\pvar x\doteq x \lstar {\color{red} \llist{x; \lstxs, \lstvs}}]} \\
\passign{\pvar x}{\mathtt{LInsert}(\pvar x, 42)}; \passign{\pvar y}{\mathtt{LSwapFirstTwo}(\pvar x)} \\
{\color{blue} [\oxerr: \exists x'.~\llist{x'; x'{\cons}\lstxs, 42{\cons}\lstvs} \lstar \listptr{x; \lstxs} \lstar \shade{|42{\cons}\lstvs| < 2} \lstar \pvar{err} \doteq \strlit{List~too~short!}]}
 \end{array}
\]

%% file: sections/language.tex

\section{Programming Language}%
\label{sec:language-concrete}

\newcommand{\serr}{\sto[\pvar{err} \rightarrow \verr]}

\newcommand{\myunmod}{\texttt{unmod}}

We present a simple imperative heap language with function calls on which
our analysis engine operates. The language is standard, except that, in line
with previous work on compositional reasoning and
incorrectness~\cite{cosette,javert,javert2,jslogic,isl}, we track freed
cells in the heap,
and separate language errors and missing-resource errors to preserve
the compositionality of the semantics.
We sometimes refer to the definitions of this section as \emph{concrete} to differentiate them from the \emph{symbolic} definitions used in the symbolic engine introduced in subsequent sections.

\subparagraph*{Syntax.} 
The values are given by: $ \gv \in \vals  ::= \nv \in \nats \mid \bv \in \bools \mid \sv \in
                       \strings \mid \nil \mid [\lst{v}] $, 
where $\lst{v}$ denotes a vector of values. The types
are given by: $\tau \in \typez ::= \vals \mid \nats \mid \bools \mid
\strings \mid \lists$. The types are used to define the semantics of
the language; the language itself is dynamically typed. The expressions, $\pexp \in \pexps$, comprise the
values, program variables  $\pvar x, \pvar y, \pvar z, \ldots
\in \pvars$, and expressions formed using the standard operators for numerical and Boolean
expressions. The commands are given by the grammar:
\[\begin{array}{r@{~}c@{~}l}
 \cmd  \in \cmds & ::= &\pskip \mid \passign{\pvar{x}}{\expr{\pexp}}
                         \mid \passign{\pvar{x}}{\prandom} \mid
                         \perror(\pexp) \mid
                         \pderef{\pvar{x}}{\expr{\pexp}} \mid
                         \pmutate{\expr{\pexp}}{\expr{\pexp}} \mid
                         \pmyalloc{\pvar{x}}  \mid
                         \\
               &&
             \pdealloc{\expr{\pexp}} \mid \pifelse{\pexp}{\cmd}{\cmd} \mid      \cmd; \cmd \mid
                  \pfuncall{\pvar{y}}{\fid}{\lst{\pexp}} 
\end{array}
\]
comprising the variable assignment, variable assignment of a
non-deterministically chosen natural number, user-thrown error, heap
lookup, heap mutation, allocation, deallocation, command sequencing,
conditional control-flow and function call.
Our results extend to other control-flow commands, e.g. loops, since these can be implemented using conditionals and recursive functions. 
The sets of program variables for expressions $\pv{\pexp}$ and commands $\pv{\cmd}$ are standard.

\subparagraph*{Functions and Function Implementation Contexts.} A function implementation, denoted $\pfunction{\fid}{\lst{\pvar x}}{\cmd; \preturn{\pexp}}$, comprises:
an identifier, $\fid \in \fids \subseteq \strings$;  the 
parameters, given by a list of distinct program variables  $\lst{\pvar x}$; a 
body, $\cmd \in \Cmd$; and a {return expression}, $\pexp \in \PExp$, with  $\pv{\pexp} \subseteq
\{\lst{\pvar x}\} \cup \pv{\cmd}$.
A~function implementation context, $\fictx$,  maps function identifiers to their implementations:
$
	\fictx : \fids \rightharpoonup_{\mathit{fin}} \pvars~\mathsf{List} \times \cmds \times \pexps 
$, where $\rightharpoonup_{\mathit{fin}}$ denotes that the function is
a finite partial function. We often  write $\fid(\lst{\pvar{x}})\{\cmd;
\preturn{\pexp}\} \in \fictx$  for $\fictx(\fid)=(\lst{\pvar{x}}, \cmd, {\pexp})$.


\subparagraph*{Stores, Heaps and States.}
A {variable store}, $\sto : \pvars
\rightharpoonup_{\mathit{fin}} \vals$, is a function from program
variables to values. A {\em partial}   heap, $\hp : \nats
  \rightharpoonup_{\mathit{fin}} (\vals \uplus \{\cfreed\})$, is a function from natural numbers to values extended with a dedicated symbol
  $\cfreed \notin \vals$ recording that a heap cell has been
  freed.
  Two heaps are disjoint, denoted $\dish{h_1}{h_2}$, when their domains are disjoint.
  Heap composition, denoted $h_1 \uplus h_2$, is given by 
the disjoint union of partial functions which is only defined 
when the domains are
disjoint. 
A {\em partial} program state, $\cst = (\sto, \hp)$, is a  pair comprising a store
and a heap. 
State composition, denoted $\cst_1 \cdot \cst_2$,  is given by 
$(s_1, h_1)\cdot (s_2,h_2)\defeq (s_1 \cup s_2,  h_1 \uplus h_2)$  for
$\cst_1 = (s_1, h_1)$ and $\cst_2 = (s_2, h_2)$,
which is only defined when 
variable stores are equal on their intersection
and the~heap~composition~is~defined.

\subparagraph*{Operational Semantics.}

We use a standard expression evaluation function, $\esem{\pexp}{\sto}$,
which evaluates an expression $\pexp$ with respect to a
store $\sto$, assuming that expressions do not affect the heap. 
It results in 
 either a value or a dedicated symbol $\undefd \notin \vals$, denoting, an
 evaluation error, such as a variable not being in the store or a mathematical error.
  The operational semantics of commands is a big-step semantics using  judgements of the form
\(
  \st, \scmd \baction_{\fictx} \outcome : \st'
\)
which reads ``the execution of command $\scmd$ in state
$\st$ and function implementation context $\gamma$ results in a 
state $\st'$ with outcome $o$'', where $o ::= \oxok \mid \oxerr \mid
\oxm$ denotes, respectively, a successful execution, a language error, 
and a missing-resource error  due to the absence of a required 
cell in the {partial} heap. The separation of the missing-resource
errors from the language errors is important for compositional
reasoning, since the language satisfies  both  the standard OX and
UX frame properties when the outcome is not missing. 
The semantics is standard and given in full 
\ifarxiv
in App.~\ref{app:concrete}, along with the frame properties it satisfies.
\fi

%% file: sections/cse.tex

\section{Compositional Symbolic Execution: Core Engine}%
\label{sec:language-symbolic}

\newcommand{\unmod}{\texttt{unmod}}

We present our CSE engine  in two
stages. In this section, we present the   core CSE engine, given 
by a standard compositional symbolic operational semantics
presented here to establish notation and introduce key concepts to the
non-specialist reader: the
definitions are similar to those for
whole-program symbolic execution; the difference is with the use of
partial state which has the effect that we have the  distinction between
language errors and missing-resource errors. 
In \S\ref{ref:compositional-rules}, 
we extend the core engine with our semantic rules for function calls and 
folding/unfolding  predicates,  using an axiomatic description of
 the consume and produce operations given in \S\ref{sec:axfun}. 


\subsection{Symbolic  States}\label{ssec:sstates}
 Let $\svars$ be a set of symbolic variables, disjoint from
 the set of program variables, $\PVar$, and values, $\vals$. 
Symbolic values are~defined~as~follows:
\[
    \begin{array}{l} 
    \sval \in \SVal  ::=  \gv \mid  \svar{x}  \mid \sval + \sval
                               \mid \ldots \mid \sval
                               = \sval \mid \neg \ \sval \mid  \sval
                               \wedge \sval 
                               \mid \sval \in \tau \\
  \end{array}
\]
A {symbolic store}, $\ssto : \PVar
  \rightharpoonup_{\mathit{fin}} \SVal$, is a  function 
  from program variables to symbolic values.
  A partial 
{symbolic heap}, $\smem : \SVal 
\rightharpoonup_{\mathit{fin}} (\SVal \uplus \{ \cfreed\})$, is a function from symbolic values  to symbolic values extended with
 $\cfreed$.
A~{path condition},~$\spc \in \SVal$, is a symbolic Boolean
expression that captures constraints imposed on symbolic variables 
during execution.
A (partial) {symbolic state} is a triple of the form $\sst = (\ssto,
\smem, \spc)$.
Throughout the paper, we only work with well-formed states $\sst$,
denoted $\sinv(\sst)$, the definition is uninformative (it ensures, e.g., that the addresses of the symbolic heap are interpreted as natural numbers), 
\ifarxiv
see App.~\ref{app:symbolic-soundness}.
\fi
We write 
 $\sst.\fieldsst{pc}$ and
$\sst[\sstupdate{pc}{\spc'}]$ to denote, respectively, access and update the state
path condition.
We write 
$\svs{\sf{X}}$ to denote  the set of symbolic
variables of a given construct~$\sf{X}$: e.g., 
$\svs{\ssto}$ for symbolic stores, $\svs{\smem}$
for symbolic~heaps,~etc.

\subparagraph*{Symbolic Interpretations.} A symbolic interpretation, $\vint : \svars \rightharpoonup_{\mathit{fin}} \Val$ maps symbolic variables to concrete values, and is used to define the meaning of symbolic states and state the soundness results of the engine. We lift interpretations to symbolic values, $\vint : \SVal \rightharpoonup_{\mathit{fin}} \Val$, with the property that it is undefined if the resulting concrete evaluation faults. Satisfiability of symbolic values is defined as usual, i.e., $\sat(\spc) \defeq \exists \vint.~\vint(\spc) = \true$. We further lift symbolic interpretations to stores, heaps, and states, overloading the $\vint$ notation.
\ifarxiv
~(see~App.~\ref{app:symbolic-soundness}).
\else
\fi

\subsection{Core Engine} The  symbolic expression evaluation
relation, $\cseeval{\pexp}{\ssto}{\spc}{\sym{w}}{\spc'}$, evaluates a program  expression $\pexp$ with respect to a symbolic store $\ssto$
and path condition $\spc$. It results in either a symbolic value or an
evaluation error, $\sym{w} \defeq \hat v \mid \undefd$, and
a satisfiable path condition $\pc' \Rightarrow \pc$, which may contain
additional constraints arising from the evaluation (e.g., to prevent branching on division by zero).
    The core CSE semantics is described using the usual single-trace semantic
    judgement (below, left) which is used to state UX properties. It
    also  induces
    the collecting semantic judgement (below, right), which is used to
    state OX properties. 
\[
\sst, \scmd \baction_\fictx \outcome : \sst' \qquad\qquad \csesemtranscollect{\sst}{\cmd}{\hat\Sigma'}{\fictx} \iff \hat\Sigma' = \{ (\result, \sst') \mid
\csesemtransabstract{\sst}{\cmd}{\sst'}{\fictx}{\result}\}
\]
We give  the lookup rules for illustration in Fig.~\ref{fig:sexrules}:
for example,  the rule \textsc{Lookup} branches over all possible addresses in the heap that can match the given address.
\ifarxiv
All the  rules of our core CSE semantics  are given in
App.~\ref{app:symbolic}.
\else
\fi

\begin{figure}[!t]
\footnotesize
\begin{mathpar}
\inferrule[\textsc{Lookup}]
 {\cseeval{\pexp}{\ssto}{\spc}{\sval}{\spc'} \quad \smem(\sloc) = \sval_m \\\\ \spc'' \defeq (\sloc = \sval) \land \spc' \quad \sat(\spc'')}
 {\csesemtrans{\ssto, \smem,\spc}{\pderef{\pvar{x}}{\pexp}}{ \ssto[\pvar{x} \mapsto \sval_m], \smem, \spc''}{\fsctx}{\osucc}}
\and
\inferrule[\textsc{Lookup-Err-Val}]
  {\cseeval{\pexp}{\ssto}{\spc}{\undefd}{\spc'}\\\\
   \sverr \defeq [\mathstr{\mathsf{ExprEval}}, \stringify{\pexp}] }
  {\csesemtrans{\ssto, \smem,  \spc}{\pderef{\pvar{x}}{\pexp}}{\ssto_{\oerr}, \smem, \spc'}{\fsctx}{\oerr}}
\and
%
\and
\inferrule[\textsc{Lookup-Err-Missing}]
 {\cseeval{\pexp}{\ssto}{\spc}{\sval}{\spc'} \quad \sat(\spc' \land \sval \in \nats \land \sval \not\in \domain(\smem)) \quad \sverr \defeq [\mathstr{\mathsf{MissingCell}}, \stringify{\pexp}, \sval]}
 {\csesemtrans{\ssto, \smem, \spc}{\pderef{\pvar{x}}{\pexp}}{\ssto_\oerr, \smem, \spc' \land \sval \in \nats \land \sval \not\in \domain(\smem)}{\fsctx}{\omiss}}
\end{mathpar}
\vspace*{-0.3cm}
\caption{Excerpt of symbolic execution rules, where $\ssto_\oerr = \ssto[\pvar{err} \mapsto \verr]$}%
\label{fig:sexrules}
\vspace*{-0.3cm}
\end{figure}

Our CSE semantics is both OX- and UX-sound, which we call exact: OX soundness captures that no paths are dropped by stating that the symbolic semantics includes all behaviour w.r.t. the concrete semantics; UX soundness captures that no information is dropped by stating that the symbolic semantics does not add behaviour w.r.t.~the~concrete~semantics.

\begin{theorem}[OX and UX soundness]\label{thm:ux-ox-sound}
\[
\begin{array}{l@{\qquad}r@{\ \implies\ }l}
\mbox{(\macOX)} &\csesemtranscollect{\sst}{\cmd}{\hat\Sigma'}{\fictx} \land
\vint(\sst),\cmd\baction_\fictx \result : \st'
& \exsts{\sst', \vint' \ge \vint} { (\result, \sst')} \in \hat\Sigma' \land \st' = \vint'(\sst') \\[2mm]
\mbox{(\macUX)}&\csesemtransabstract{\sst}{\cmd}{\sst'}{\fictx}{\result} \land
\sint(\sst') = \st' &
\sint(\sst),\cmd\baction_\fictx \result : \st'
\end{array}
\]
where $\vint' \ge \vint$ denotes that $\vint'$ extends $\vint$, i.e., $\vint'(\svar{x}) = \vint(\svar{x})$ for all $\svar{x} \in \domain(\vint)$.
\end{theorem}
\noindent
\ifarxiv
The proofs are standard, 
with several illustrative cases given in App.~\ref{app:symbolic-soundness}.
\fi

%% file: sections/func-interface.tex

\section{Compositional Symbolic Execution: Full Engine}%
\label{sec:function-interface}

Our core CSE engine  is limited in
that it does not call function
specifications written in a program logic,  and it cannot fold
and unfold user-defined predicates to verify, e.g., specifications
of list algorithms. What is missing is a general description of 
how to update  symbolic state using assertions from the function
specifications and predicate definitions.  In  VeriFast, Viper
and Gillian, this symbolic-state update is given by their   implementations
of the  {consume} and produce operations. 
We instead  give  an
\emph{axiomatic interface} for describing  such 
symbolic-state update by 
providing a general  characterisation
of these consume and produce operations
(\S\ref{sec:axfun}). Using this interface, we are able to give general definitions of 
the function-call rule (\S\ref{ref:compositional-rules}) and
the folding and unfolding of predicates that are independent of the underlying tool implementation.
Assuming the appropriate properties stated in the
axiomatic interface, we prove that the resulting CSE engine is
either 
OX sound or UX sound. Moreover, because the axiomatic interface relates the behaviour of the consume and produce operations to the standard satisfaction relation of SL and ISL, our function-call rule is able to use any function specification proven correct with respect to the standard function specification validity of SL and ISL, including functions specification proven outside our engine. 
In the next section  (\S\ref{sec:function-implementation}),
we demonstrate that the Gillian
implementation of the consume and produce operations are correct with
respect to our axiomatic
interface.\footnote{To our best understanding, there is a large overlap
  between Gillian's consume and produce operations and those  of Viper and VeriFast. We therefore expect them to also satisfy the OX properties of our interface (we have however not proven this fact).}

\subsection{Assertions and Extended Symbolic States}%
\label{sub:funspec}

We present assertions  suitable for both SL and ISL reasoning, and also extend our
symbolic states to account for predicates. 
It is  helpful to make a clear
distinction between the logical 
assertions and symbolic states: since we
  work with the linear heap, the gap between assertions and
  symbolic states is quite small; with more complex memory models and 
  optimised symbolic representations of memory, the gap is larger and
 this distinction becomes essential.

\subparagraph*{Assertions.}  
Let $x, y, z \in \LVar$ denote logical variables, distinct from
program  and symbolic variables. The set of logical
expressions, $E \in \LExp$, extends program expressions to include
logical variables. We work with the following set of assertions (other assertions are derivable):
\[\!
  \begin{array}{r@{~}c@{~}l}
   \pc \in \PAssert  &\defeq  & \lexp \mid \lexp \in \tau \mid \lnot \pc \mid \dots \mid \pc_1 \land \pc_2 \\
   P \in \Assert & \defeq &
   {{\pc}} \mid \AssFalse\!\mid\!{ P_1 \Rightarrow   P_2} \mid P_1 \vee
                            P_2 \mid \exists x \ldotp P \mid \\
                            && \emp \mid \lexp_1 \mapsto \lexp_2 \mid \lexp \mapsto \cfreed \mid  P_1 \lstar P_2 \mid \Predd{\vec{\lexp_1}}{\vec{\lexp_2}}
   \end{array}
 \]
where $\lexp, \lexp_1, \lexp_2  \in \LExp$, $\lvar{x} \in \LVar$, and $\pred \in \strings$.
The assertions should by now be familiar from separation logic. They
comprise the lift of the usual first-order Boolean
assertions $\pi$, assertions built from the usual first-order
connectives and quantifiers,  and assertions well-known from separation logic:
the empty assertion $\emp$, the cell assertion  $\lexp_1 \mapsto
\lexp_2$ describing a heap cell at an address given by $\lexp_1$ with
value given by $\lexp_2$, the less well-known assertion
$\lexp \mapsto \cfreed $ describing a heap cell at address $\lexp$
that has been freed, the separating conjunction $P_1 \lstar P_2$, and
the predicate assertions $\Predd{\vec{\lexp_1}}{\vec{\lexp_2}}$. 
The parameters of predicate assertions
$\Predd{\vec{\lexp_1}}{\vec{\lexp_2}}$ are separated into
in-parameters $\vec{\lexp_1}$ (``ins'') and out-parameters
$\vec{\lexp_2}$ (``outs'') for automation purposes, as we
discuss  in \S\ref{sec:function-implementation}; this separation does
not affect the logical meaning of the predicate assertions. We write $\lv{\mathsf{X}}$ to denote free
logical variables of a construct~$\mathsf{X}$: e.g., $\lv{\lexp}$ for
logical expressions, $\lv{P}$ for assertions,~etc.
We say that an
assertion $P$ is {\it simple} if it does not syntactically feature the
separating conjunction; simple assertions are used in the definition of matching plans (\S\ref{sec:macMP}).

\subparagraph*{Predicates.} Predicate definitions are given by a set
$\preds$ containing elements of type  $\strings \times \vec{\LVar}
\times \vec{\LVar} \times \Assert$,  with the notation
$\Predd{\predin}{\predout}~\{ P \} \in \preds$, where the string
$\pred$ denotes  the predicate name, the lists of disjoint parameters  $\predin, \predout$  denote the predicate ins and outs
respectively, and the assertion $P= \bigvee_i (\exists \vec{y}_i.~P_i)$ denotes the predicate body, which does not contain program variables and whose free logical variables are contained in $\{\predin\} \cup \{\predout\}$ which are disjoint from the bound variables $\vec{y}_i$, and where the $P_i$'s (denoting the assertions of the predicate definition) do not contain disjunctions or existential quantifiers.

\ifarxiv
For example, in the standard ${\sf list}$ predicate
\[
\begin{array}{r@{~}c@{~}l}
\llist{x; \lstvs} & \defeq & (x\doteq\nil \lstar \lstvs \doteq \emplist)~\lor \\
                  && (\exsts{v, x', \lstvs'} x \mapsto v, x' \lstar\llist{x'; \lstvs'} \lstar \lstvs \doteq v \cons \lstvs') 
\end{array}
\]
the logical variable $x$ is an in-parameter and the list $\lstvs$ is an out-parameter. The logical variables $v,x',\lstvs'$ are bound in the right disjunct of the predicate body.
\else
\fi

\subparagraph*{Satisfiability.} The meaning of an assertion $P$  is defined by capturing the models of $P$ using the standard satisfaction relation $\subst, \st \models P$ where 
 $\subst : \LVar \rightharpoonup_{\mathtt{fin}} \Val$ is a logical interpretation represented by a function
 from logical variables to values and $\st$ is a program state (as
 defined in \S\ref{sec:language-concrete}). 
\ifarxiv
This satisfaction relation, given partially below,
 is the  standard relation underpinning all separation logics;
its full definition is given  in
App.~\ref{app:symbolic-soundness-func}.
\[
\begin{array}{l@{\hspace{1cm}}l}
 \subst, (\sto, \hp)  \models  & \\
\begin{array}{rcl}
 \quad {{\pc}} &\Leftrightarrow & \esem{{\pc}}{\subst,\sto} = \true \wedge \hp = \emptyset \\
  \quad P_1 \lor P_2 &\Leftrightarrow& \subst, (\sto, \hp)  \models P_1 \lor  \subst, (\sto, \hp) \models P_2\\
 \quad {\lexp_1 \mapsto \lexp_2} &\hspace{0.05cm}\Leftrightarrow& \hp = \{ \esem{\lexp_1}{\subst, \sto} \mapsto \esem{\lexp_2}{\subst, \sto}\} \\
 \quad P_1 \lstar P_2 &\Leftrightarrow& \exists \hp_1, \hp_2 \ldotp \hp = \hp_1 \uplus \hp_2 \land \subst, (\sto, \hp_1) \models P_1 \wedge \subst, (\sto, \hp_2) \models P_2 
\end{array}
\end{array}
\]
\fi


%

\subparagraph*{Function Specifications.}
The quadruples $\quadruple{\vec{\pvar{x}} = \vec{x} \lstar
  P}{\fid(\vec{\pvar{x}})}{\Qok}{\Qerr}$ 
 and $\islquadruple{\vec{\pvar{x}} =
  \vec{x} \lstar P}{\fid(\vec{\pvar{x}})}{\Qok}{\Qerr}$ denote, respectively,  a
SL and an ISL function specification, as explained in \S\ref{ki:cse}.  We write $\genquadruple{\vec{\pvar{x}} = \vec{x} \lstar
  P}{\fid(\vec{\pvar{x}})}{\Qok}{\Qerr}$ to refer to either. Both
quadruples record successful executions and language errors. They
are unable to record missing-resource errors, as these errors do not satisfy the
OX and UX frame properties. Missing errors can be removed automatically via UX bi-abduction (see \S\ref{sec:bi-abduction}).

Formally, we define function specifications using internalisation~\cite{esl}. In short, internalisation relates internal specifications, which describe the internal behaviour of functions, to external specifications, which describe the external behaviour of functions. Internalisation is needed for ISL to allow the logic to drop information about function-local program variables at function boundaries, since dropping information is in general not allowed in ISL. The full definitions of function specifications and internalisation are included in 
\ifarxiv
App.~\ref{app:symbolic-soundness-func}.
\fi

A function specification context, $\fsctx \in \fids \rightharpoonup_{\mathit{fin}} \mathcal{P} (\especs)$, maps function identifiers to a finite set of external specifications $\especs$. To simplify the presentation of the paper, we assume existential quantifiers only occur at the top level of external specifications. We denote the validity of $\fsctx$ with respect to $\fictx$ by $\models (\fictx, \fsctx)$, and validity of a function specification with respect to $\fsctx$ by $\fsctx \models \genquadruple{\vec{\pvar{x}} = \vec{x} \lstar P}{\fid(\vec{\pvar{x}})}{\Qok}{\Qerr}$.

\subparagraph*{Extended Symbolic States.}

To reason about unbounded execution to verify, for example,
specifications of  list algorithms, 
we extend the partial symbolic states defined in \S\ref{ssec:sstates} with symbolic predicates of
the form $\Predd{\vec{\sval}_1}{\vec{\sval}_2}$, with $\pred \in
\strings$ and $\vec{\sval}_1, \vec{\sval}_2 \in \vec{\SVal}$. An
extended symbolic state $\sst$ is a tuple $(\ssto, \smemb,\spc)$
comprising 
a partial symbolic state $\ssto$, a symbolic resource 
$\smemb=(\smem,\sps)$ with   symbolic
heap  $\smem$ and  multiset of symbolic predicates $\sps$, 
and a symbolic path condition $\spc$. 
\ifarxiv
We provide the extended definitions of well-formedness of
symbolic state and symbolic interpretations in
App.~\ref{app:symbolic-soundness-func}.
\fi
 We define 
$\sint, \st \smodels (\ssto, \smemb, \spc)$  analogously to assertion
satisfaction since the interpretation of a symbolic state with respect
to symbolic interpretation $\vint : \svars
\rightharpoonup_{\mathit{fin}} \Val$ is a relation and not a function (cf. \S\ref{sec:language-symbolic}) due to the presence of symbolic predicates.

The composition of two extended symbolic states is defined by:
\[
(\ssto_1, \smemb_1, \spc_1) \cdot (\ssto_2, \smemb_2,  \spc_2) \defeq (\ssto_1 \cup \ssto_2, \smemb_1 \cup \smemb_2, \spc_1 \wedge \spc_2 \land \sinvc(\smemb_1 \cup \smemb_2))
\]
where $\smemb_1 \cup \smemb_2$ denotes the pairwise union of the
components of the symbolic resource and  $\sinvc(\smemb_1 \cup
  \smemb_2)$ ensures that the composition is well-formed.
\ifarxiv
  (as defined in App.~\ref{appssec:extsstates}).
\else
\fi

\subsection{Axiomatic Interface for Consume and Produce}%
\label{sec:axfun}

We present our axiomatic interface for the consume and produce
operations, used to update the symbolic state during function call and
to fold and unfold  
the predicates. 
Given the  substitution $\ssubst :
\LVar \rightharpoonup_{\mathtt{fin}} \SVal$, 
 the \mac and \produce ~operations have the signatures:
\[
  \mac(m, P, \ssubst, \sst) \rightsquigarrow (\ssubst', \sst_f) \mid
  \oxabort(\sval)
  \qquad\qquad
  \produce(P, \ssubst, \sst) \rightsquigarrow \sst'
\]
Recall the use of the consume and produce operations in the function call illustrated in 
Fig.~\ref{fig:ki-interface} of  \S\ref{ki:cse}. For  $\mac(m, P, \ssubst, \sst) $,
the initial
substitution $\ssubst$ comes from replacing the function parameters
with symbolic values given by the arguments in the function call;
consume matches the precondition $P$ and  substitution $\ssubst$
against part of $\sst$,
removing the appropriate resource $\sst_P$ and returning the
frame $\sst_f$ and the substitution $\ssubst'$ which extends $\ssubst$ 
with further information given by the match. For 
  $\produce(Q_{ok}, \ssubst', \sst_f)$, 
the
produce takes the postcondition $Q_{ok}$ and this resulting
substitution
$\ssubst'$ and  creates a symbolic state
which is composed with  $\sst_f$ to obtain $\sst'$. Notice that the
consume operation can abort with error information if no match is found. The
$\produce$ operation does not abort, but it may render states
unsatisfiable, in which case the branch is cut.

\begin{figure}[!t]
\begin{minipage}{14cm}
\begin{mdframed}[innerrightmargin=0.1cm, rightmargin=0,innerleftmargin=4,linewidth=.4pt]

  \small

We assume the following holds initially: state $\sst$
and  substitution $\ssubst$  are well-formed;  and  $\ssubst$ covers~$P$ for \produce, that is, $\lv{P} \subseteq \dom(\ssubst)$. For properties 1--4 below, consider~the~following~executions:
\[
\begin{array}{l@{\hspace{3cm}}l}
\mac(m, P, \ssubst, \sst) \rightsquigarrow (\ssubst', \sst_f)  & \textcolor{gray}{where  \quad \sst=(\ssto, \smemb, \spc)~and~\sst_f=(\ssto', \smemb_f, \spc')}\\
\produce(P, \ssubst, \sst) \rightsquigarrow \sst' & \textcolor{gray}{where \quad\sst=(\ssto, \smemb, \spc)~and~\sst'=(\ssto', \smemb', \spc')}
\end{array}
\]

\setcounter{property}{0}
\begin{property}[{Well-formedness}]\label{prop:wf} The variable store is not altered: that is, $\ssto' = \ssto$ and 

\begin{description}[font=\normalfont\textit]
\item[(\mac)] $\sinv(\sst_f)$ and  $\sinv(\ssubst', \spc')$ \qquad \qquad \textcolor{darkgray}{$(\produce)$}  $\sinv(\sst')$
\end{description}
\end{property}

\begin{property}[Path Strengthening]\label{prop:path_strength}
$\spc' \Rightarrow \spc$
\end{property}


\begin{property}[Consume Covers $P$]\label{prop:coverage} $\ssubst' \geq
\ssubst $ and $\dom(\ssubst') \supseteq \lv{P}$
\end{property}

\begin{property}[Soundness] \label{prop:soundness} \quad 

\begin{description}[font=\normalfont\textit]
\item[(\mac)] \(
\exists \smemb_P.~\smemb = \smemb_f \cup \smemb_P 
\land (\forall \vint, \st.~\sint, \st \smodels \sst_P \implies \sint(\sym\theta'), \st \models P)
\)
 \ \textcolor{gray}{where $\sst_P \defeq
   (\emptyset,\smemb_P,\spc')$}\footnote{We choose the empty symbolic
   store for $\sst_P$; $P$ does not have program variables so this
   choice is arbitrary. The symbolic state $\sst_P$ is  the one used in Prop.~\ref{prop:mac-ux-comp}.}
\item[(\produce)]  ~\(
\exists \smemb_P.~\smemb' = \smemb \cup \smemb_P 
\land (\forall \vint, \st.~\sint, \st \smodels \sst_P \implies \sint(\sym\theta), \st \models P)\)  ~\kern 0.5em \textcolor{gray}{where  \ $\sst_P \defeq (\emptyset,\smemb_P,\spc')$}
\end{description}
\end{property}

\begin{property}[Completeness: \macOX consume]\label{prop:mac-branch-comp}
 If \(  \oxabort \not\in \mac(\macEX, P, \ssubst, \sst) \) and \(\vint, \st \smodels \sst\), then
\[
  \begin{array}{l}
   \exists \ssubst', \st_f, \sst_f.~\mac(\macEX, P, \ssubst, \sst) \rightsquigarrow (\ssubst', \sst_f) \land
  \vint, \st_f \smodels \sst_f
  \end{array}
  \]
\end{property}
\begin{property}[Completeness: \macUX consume]\label{prop:mac-ux-comp}
If \(\mac(\macUX, P, \ssubst, \sst)\rightsquigarrow (\ssubst', \sst_f) \), \(\sint(\spc')=\true\) and \(\sint(\ssubst'), (\emptyset ,h_P)\models P \),~then
\[
\begin{array}{@{}l}
 \sint, (\emptyset, \hp_P) \smodels \sst_P \gand 
(\forall \hp_f.~\sint, (\sto, \cmem_f) \smodels \sst_f \gand \dish{\hp_P}{\cmem_f} \implies \sint, (\sto, \hp_f \uplus \hp_P) \smodels \sst)
\end{array}
\]
\end{property}
\begin{property}[Completeness: produce]\label{prop:prod-compl}
 If \(\sint, (\sto, \cmem) \smodels\sst  \) and \( \sint(\ssubst), (\emptyset ,h_P)\models P \) and \( \dish{\cmem}{\hp_P} \), then
\[
\begin{array}{l} 
\exists \sst_P.~\produce(P, \ssubst, \sst) \rightsquigarrow \sst \cdot \sst_P \gand \sint, (\emptyset, \hp_P) \smodels \sst_P
\end{array}
\]
\end{property}
\end{mdframed}
\end{minipage}
\caption{The axiomatic interface for the \mac and \produce~operations}\label{fig:interface}
\end{figure}

In Fig.~\ref{fig:interface}, we present the 
axiomatic  interface  of the consume and
produce operations, identifying  sufficient properties 
to prove OX and UX soundness  for 
the function-call rule and the folding and unfolding of predicates, as
we demonstrate in
the next section (\S\ref{ref:compositional-rules}). 
Properties 1--3 ensure that the operations are compatible with the
expected 
properties of symbolic execution, including well-formedness
$\sinv(\ssubst', \spc')$ of the symbolic substitutions with respect to
path conditions.  This property  guarantees that the $\spc'$
implies that $\ssubst'$ does not map logical variables into $\undefd$,
that is,  $\spc'\vDash\codom(\ssubst')\subseteq \Val$.
\ifarxiv
 (the formal definition is in App.~\ref{appssec:extsstates}).
\else
\fi

 Properties 4--7 give conditions
for \mac and \produce to soundly decompose and compose symbolic states
respectively, while being compatible with symbolic and logical
interpretations. These properties will come as no surprise to those
with a formal 
knowledge of symbolic execution.
However, their identification was not easy, requiring a considerable amount of
back and forth between the soundness proof and the properties to
pin them down properly. 
We  describe the more interesting of the properties described in Fig.~\ref{fig:interface}:

\begin{itemize}[leftmargin=*]
\item  Prop.~\ref{prop:path_strength} states that path conditions may
  only get strengthened: 
    OX  and UX \mac  may strengthen $\spc$ to due to branching; additionally UX \mac
    may strengthen $\spc$ arbitrarily due to cutting; and \produce may add constraints to $\spc$ arising from~$P$.

\item Prop.~\ref{prop:soundness}, for consume (and similarly for produce), states that the operation is sound: the symbolic resource of $\sst$ can be decomposed as $\smemb = \smemb_f \cup \smemb_P$, i.e., it consists of the symbolic resources of $\sst_P$ and $\sst_f$, respectively, and $\forall \vint, \st.~\sint, \st \smodels \sst_P \implies \sint(\sym\theta''), \st \models P$ states that  all models of $\sst_P$ are models of $P$. 

\item Prop.~\ref{prop:mac-branch-comp} captures that successful OX consumptions do not drop paths: if no branch aborts, and we have a model $\sint, \st \models \sst$, then there exists a branch $\sst_f$ with a model using the same $\sint$, i.e., there exist $\st_f$ such that $\sint, \st_f \models \sst_f$. 
\ifarxiv
(Specifically, $\sint(\sst.\fieldsst{pc}) = \true$ implies $\sint((\sst_f).\fieldsst{pc}) = \true$.)
\else
\fi
\item Prop.~\ref{prop:mac-ux-comp} is as follows: in successful UX \mac, any model of the consumed assertion $P$ must also model the consumed state $\sst_P$ (obtained from Prop.~\ref{prop:soundness}), and when extended with a compatible model of the output state $\sst_f$ it must model the input state $\sst$.

\end{itemize}

Assuming these properties of the consume and produce operations, we are
able to  prove that the function-call rule and
the predicate folding and unfolding are sound, and thus that our whole
CSE engine is sound (Thm.~\ref{thm:ux-ox-sound-func}).
In \S\ref{ref:compositional-rules}, we give an example (Ex. \ref{ex:fcall-ux}) illustrating some of the properties in action during a function call. 
\ifarxiv
In addition,  in App~\ref{app:rules-consume}, there is an example that demonstrates the satisfaction relation described by Props.~\ref{prop:soundness} and \ref{prop:mac-ux-comp} for consume in UX mode  (see  Ex.~\ref{ex:illustrate_props}).
\else
\fi


 

\subsection{Full CSE Engine}%
\label{ref:compositional-rules}

We introduce our full CSE engine, 
extending the  core CSE engine (\S\ref{sec:language-symbolic}) with
the ability to soundly call  valid SL and ISL function specifications (\S\ref{sub:funspec}), 
and to fold/unfold predicates. 
We extend and adapt
the compositional symbolic operational semantics to
carry a specification context $\fsctx$ and the mode of execution $m$
(OX or UX), obtaining the judgement
$\csesemtransabstractm{\sst}{\cmd}{\sst'}{\fsctx}{m}{\result}$, and
extend the possible outcomes with $\oxabort$, as \mac can abort. The
rules are analogous except for the rules for function calls and predicate folding/unfolding, as detailed below.\footnote{The satisfiability check $\sat(\spc)$ used by the rules over-approximates the existence of a model for $(\ssto, \smemb, \spc)$, due to the presence of symbolic predicates; for sound reasoning in UX mode, our engine addresses this source of over-approximation by under-approximating the satisfiability check once by bounded unfolding of predicates at the~end~of~execution.}

\begin{figure}[!t]
\begin{mathpar}
\footnotesize
\infer
{\csesemtransabstractm{\sst}{\passign{\pvar{y}}{\fid(\vec{\pexp})}}{\sst''[\sstupdate{sto}{\ssto[\pvar{y} \mapsto \svar{r}]}]}{\fsctx}{m}{\osucc}}
{
\begin{array}{l@{\hspace*{-0.2cm}}r}
(1)~\genquadruple{\vec{\pvar{x}} = \vec{x} \lstar P}{\fid(\vec{\pvar{x}})}{\Qok}{\Qerr} \in \fsctx(f)|_m & \text{get function specification} \\
(2)~\cseeval{\vec{\pexp}}{\ssto}{\spc}{\vec{\sval}}{\spc'} \text{ and }\hat{\theta} \defeq \{\vec{\sval}/\vec{x} \}  & \text{evaluate function parameters} \\
(3)~\mac(m, P, \hat{\theta}, \sst[\sstupdate{pc}{\spc'}]) \rightsquigarrow (\hat{\theta}', \sst')& \text{consume pre-condition} \\
\textcolor{gray}{(4)~\Qok = \exists \vec{y}.~\Qok'} & \textcolor{gray}{\text{as $\Qok$ is a post-condition}} \\
\textcolor{gray}{(5)~\ssubst'' \defeq \ssubst{'}\cup\{\vec{\hat{z}}/\vec{y}\}}&\textcolor{gray}{\text{ extend substitution to cover } \Qok}\\
\textcolor{gray}{(6)~\vec{\hat{z}}, r, \hat{r} \text{ fresh}} & \textcolor{gray}{\text{fresh variables}} \\
\textcolor{gray}{(7)~\Qok'' = \Qok'\{r/\pvar{ret}\} \text{~and~} \ssubst''' = \ssubst''\cup\{\hat{r}/r\}}& \textcolor{gray}{\text{set up return value}} \\
(8)~\produce(\Qok'', \ssubst''', \sst') \rightsquigarrow  \sst'' & \text{produce post-condition}
\end{array}
}
\end{mathpar}
\vspace{-0.65cm}
\caption{Unified function-call rule for CSE: success case, where $\sst
  = (\ssto, \smemb, \spc)$}
\label{fig:fcall}
\end{figure}

\subparagraph*{Function-Call Rule.} 
The unified success rule for a function call is in Fig.~\ref{fig:fcall}, using the notation $\fsctx(f)|_m$ to isolate the $m$-mode specifications of $f$. 
 A description of each step is included in the rule itself. In short, given an initial state $\sst$, the rule selects a function specification, consumes the specification pre-condition from $\sst$, resulting in $\sst'$, and produces the post-condition of the specification into $\sst'$, resulting in the final state $\sst''$. The steps in \textcolor{gray}{grey} are uninteresting (about renamings and fresh variables) and can be ignored on~a~first~reading.

\begin{example}\label{ex:fcall-ux}
We show a possible execution of a function call using a function specification, where we assume we have been given \mac and \produce example implementations that satisfy the axiomatic interface. Consider the function ${\sf f}$ given in \S\ref{ki:core} and the ISL specification given in \S\ref{ki:biabduction}: 
$\isltripleok{\pvar c = c \lstar \pvar x = x \lstar P}{\mathtt{f}(\pvar c, \pvar x)}{x \mapsto c \lstar c \ge 42 \lstar \pvar{ret} = v }$
where $P$ is $x \mapsto v$, here assumed to be in the function specification context $\Gamma$. Suppose the symbolic execution is in a state $\sst=(\ssto,\smemb,\spc)$ and that the next step is $\sst, \pvar{y}:={\sf f}(50,1)\Downarrow^{\macUX}_{\Gamma}ok:\sst'$. Let  $\smemb=(\{\svar{x}\mapsto \svar{c},\svar{y}\mapsto1, 3\mapsto 5\},\emptyset)$ be the symbolic resource, and $\spc=\svar{c}\geq 42\land\svar{x}\neq \svar{y}\land \svar{x},\svar{y}\in \Nat$ be the  symbolic path condition  (the symbolic store $\ssto$ is irrelevant to this computation and left opaque). 
We now follow the steps (1) - (8) described in the function-call rule in Fig.~\ref{fig:fcall}.

Step (1) is above. Step (2) evaluates the parameters of the function call  which  in this case yields the initial
substitution is  $\ssubst=\{50/c,1/x\}$. Step (3) is to consume the
pre-condition of ${\sf f}$: $\ssubst$ identifies the logical
variable $x$ with $1$, and thus, this $\ssubst$ maps $P$ into $1\mapsto v$; now we check whether
there exists a resource in $\smemb$ that matches this. There are two
possibilities: either $\svar{x}=1$ and $v=\svar{c}$; or $\svar{y}=1$
and $v=1$. Let us choose the first match. Thus, with our axiomatic description of a consume operation,   $\mac(\macUX, x\mapsto v, \ssubst, \sst)$
gives  the pair $(\ssubst',\sst_f)$, with substitution
$\ssubst'=\{50/c,1/x, \svar{c}/v\}$ and  symbolic frame
$\sst_f=(\ssto,(\{\svar{y}\mapsto 1, 3\mapsto 5\},\emptyset ),
\spc\land \svar{x}=1)$. Here $\ssubst'\geq \ssubst$ as described by
Prop.~\ref{prop:coverage} of Fig.~\ref{fig:interface}, the new path
condition $\spc\land \svar{x}=1$ is stronger than the initial
$\spc$, as required by Prop.~\ref{prop:path_strength}.

Steps (4) - (5)~are straightforward: $\Qok$ is not existentially quantified and the domain of $\ssubst'$ covers $\Qok=x \mapsto c \lstar c \ge 42 \lstar \pvar{ret} = v$.  Steps (6) - (7)  set up the return value by renaming ${\sf ret}$ with a fresh logical variable $r$ as in $\Qok''=\Qok\{r/{\sf ret }\}$ and defining the substitution $\ssubst''=\ssubst'\uplus \{\svar{r}/r\}$, with  $\svar{r}$ a fresh symbolic variable.  Step (8) produces the post-condition
which results in $(\ssto,(\{\svar{y}\mapsto 1, 3\mapsto 5,1\mapsto 50\},\emptyset),\spc')$, for some $\spc'$  that  is satisfiable.
\ifarxiv
 The full computation is in App.~\ref{app:prod_rules}, see Ex.~\ref{ex:fcall-ux-complete}.
\else
\fi
\end{example}

We illustrate the general execution of the function-call rule in Fig.~\ref{fig:branching}. Successful \mac may branch (in the figure: $\sst_{f_1},\ldots, \sst_{f_k}$) due to different ways of matching with the symbolic state $\sst$, and the function call will branch accordingly. In both modes, in each successful branch, say with frame state $\sst_{f_i}$, the function-call rule will call \produce, which will produce both $\Qok$  and $\Qerr$ postconditions of the function specification. The function call propagates errors from \mac, whose error handling can depend on the mode of reasoning. In OX mode, all errors must be reported; the figure shows an example with two $abort$ outcomes, one with a symbolic value $\sst_{miss}$, representing a missing outcome, and another abort $\sst_{\sval}$. In UX mode, in contrast, errors can be cut: e.g., a \mac implementation may choose to report missing errors (to be used in e.g. bi-abduction, see \S\ref{sec:bi-abduction}), but cut other errors, as illustrated in the figure. Lastly, note that \mac implementations must represent missing-resource errors as abort errors. To see why, consider the function $\pfunction{\textsf{do\_nothing}}{}{\pskip; \preturn{\nil}}$ and the (nonsensical but valid) specification $\isltripleok{5 \mapsto 0}{\textsf{do\_nothing}()}{5 \mapsto 0 \lstar \pvar{ret} = \nil}$. Of course, in the concrete semantics, calling the function will never result in a miss. Now, say the symbolic engine calls the function using the provided function specification. If the the resource of the pre-condition is not available in the current symbolic heap, then the consumption of the pre-condition will fail. Because no concrete execution of the function results in a miss, it would be unsound for the consumption to report a missing-resource~error~in~this~case.

\subparagraph*{Predicate Rules.} To handle the folding and unfolding of predicates in symbolic states, we extend the language syntax with the following two ghost commands (also known as tactic commands):
\ifarxiv
\[
\cmd \in \cmds ::= \dots \mid \pfold{p}{\vec{\pexp}} \mid \punfold{p}{\vec{\pexp}} 
\]
\fi
\noindent where $\vec{\pexp} \in \PExp$ specifies the values of the in-parameters of the predicate $p$. In the concrete semantics, these commands are no-ops, as they are ghost commands. The symbolic-semantics rules are similar to the function-call rule: in short, a fold of a predicate consumes the body of the predicate, learns the out-parameters of the predicate, and adds the predicate (with the specified in-parameters and learnt out-parameters) to the symbolic state; and an unfold of a predicate finds a corresponding predicate in the symbolic state, learns the out-parameters of the predicate, and produces the body of the predicate.

\ifarxiv
We give the full symbolic rules in  App.~\ref{ssec:fold_unfold}.
\fi

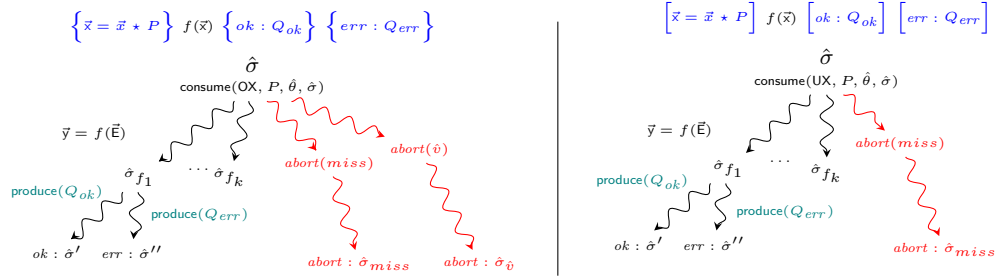
\begin{figure}[!t]
\tiny
\begin{tikzpicture}[scale=0.6]
\node (fcall) at (-2.5,1) {\tiny$\vec{\pvar{y}}=\fid (\vec{\pvar{E}})$};
\node (sst) at (1,2.5){\footnotesize$\sst$};
\node[yshift=-0.5cm] (spec) [above of=sst] {\tiny$ \quadruple{\vec{\pvar{x}} = \vec{x} \lstar P}{\fid(\vec{\pvar{x}})}{\Qok}{\Qerr}$};
  \node (root) at (1,2) {\tiny$\mac(\macOX,P,\ssubst,\sst)$};
 \node (frame1)   at (-1.5,0)       {\tiny$\sst_{f_1}$};
 \node[xshift=-.1cm] (sproderr)  [below of= frame1] {\tiny$\oxerr:\sst''$};
  \node (sprodok)  [left of= sproderr] {\tiny$\osucc:\sst'$};
  \node (frame4) at (0.5,0) {\tiny$\sst_{f_k}$};
   \node[red] (abort1) at (4.7,0.6) {\tiny$\oxabort(\sval)$};
\node[teal,yshift=0.2cm,xshift=-0.5cm] (prodok) [below of=fcall]  {\tiny $\produce(\Qok)$};
\node[teal,xshift=0.8cm,yshift=-0.3cm,xshift=0.1cm] (proderr) [right of=prodok]  {\tiny $\produce(\Qerr)$};
\node[yshift=1mm, xshift=-0.2cm] (dots1) [right of=frame1] {\tiny\ldots};
  \node[red, xshift=0.8cm, yshift=-0.51cm] (abortf1)  [below of= abort1]{\tiny$\oxabort:\sst_{\sval}$};
    \node[red,xshift=-0.6cm] (abortf2)  [left of= abortf1]{\tiny$\oxabort:\sst_{miss}$};
   \node[red] (abort4) at (2.7,.3) {\tiny$\oxabort(miss)$};
  \draw [->,decorate,decoration={snake,aspect=0},>=stealth]      (root) -- (frame1);
 \draw[->,decorate,decoration={snake,aspect=0},>=stealth]       (frame1) -- (sprodok);
  \draw[->,decorate,decoration={snake,aspect=0},>=stealth]       (frame1) -- (sproderr);
  \draw[->,decorate,decoration={snake,aspect=0},>=stealth]       (root) -- (frame4);
   \draw[->,red, decorate,decoration={snake,aspect=0},shorten <= .1mm,>=stealth]       (root) -- (abort1);
   \draw[->,red, decorate,decoration={snake,aspect=0}, >=stealth]       (root) -- (abort4);
   \draw[->, red,decorate,decoration={snake,aspect=0},>=stealth]       (abort1) -- (abortf1);
     \draw[->, red,decorate,decoration={snake,aspect=0},>=stealth]       (abort4) -- (abortf2);
\end{tikzpicture}
\hspace{0.3cm}
\begin{tikzpicture}[scale=0.65]
\draw (-4.5,3.2) -- (-4.5,-2);
\node (fcall) at (-2,1) {\tiny$\vec{\pvar{y}}=\fid (\vec{\pvar{E}})$};
\node (sst) at (1,2.5){\footnotesize$\sst$};
\node[yshift=-0.5cm] (spec) [above of=sst] {\tiny$ \islquadruple{\vec{\pvar{x}} = \vec{x} \lstar P}{\fid(\vec{\pvar{x}})}{\Qok}{\Qerr}$};
  \node (root) at (1,2) {\tiny$\mac(\macUX,P,\ssubst,\sst)$};
 \node (frame1)   at (-1,0.2)       {\tiny$\sst_{f_1}$};
 \node[xshift=-.2cm] (sproderr)  [below of= frame1] {\tiny$\oxerr:\sst''$};
  \node (sprodok)  [left of= sproderr] {\tiny$\osucc:\sst'$};
  \node (frame4) at (1,0.1) {\tiny$\sst_{f_k}$};
\node[teal,yshift=0.3cm,xshift=-.5cm] (prodok) [below of=fcall]  {\tiny $\produce(\Qok)$};
\node[teal,xshift=0.8cm,yshift=-0.4cm,xshift=0.1cm] (proderr) [right of=prodok]  {\tiny $\produce(\Qerr)$};
\node[yshift=1mm, xshift=-0.3cm] (dots1) [right of=frame1] {\tiny\ldots};
 \node[red,xshift=-0.4cm,yshift=2mm] (abortf2)  [left of= abortf1]{\tiny$\oxabort:\sst_{miss}$};
   \node[red] (abort4) at (2.5,.7) {\tiny$\oxabort(miss)$};
  \draw [->,decorate,decoration={snake,aspect=0},>=stealth]      (root) -- (frame1);
 \draw[->,decorate,decoration={snake,aspect=0},>=stealth]       (frame1) -- (sprodok);
  \draw[->,decorate,decoration={snake,aspect=0},>=stealth]       (frame1) -- (sproderr);
  \draw[->,decorate,decoration={snake,aspect=0},>=stealth]       (root) -- (frame4);
   \draw[->, red,decorate,decoration={snake,aspect=0},>=stealth]       (abort4) -- (abortf2);
   \draw[->,red, decorate,decoration={snake,aspect=0}, >=stealth]       (root) -- (abort4);
\end{tikzpicture}
\vspace*{-0.1cm}
\caption{Branching in \macOX and \macUX function calls}\label{fig:branching}
\vspace*{-0.3cm}
\end{figure}

\subparagraph*{Soundness.} Our CSE engine is sound: OX soundness, expectedly, has no restrictions on the predicates; UX soundness, on the other hand, allows only strictly exact predicates (i.e., predicates whose bodies are satisfiable by at most one heap~\cite{Yangphd}) to be folded to ensure that no~information~is~dropped.
\begin{theorem}[Compositional OX and UX soundness]\label{thm:ux-ox-sound-func}
If all UX predicate foldings are limited to strictly exact predicates, then the following hold:
\[
\begin{array}{l}
\models (\fictx, \fsctx) \land
\csesemtranscollectm{\sst}{\cmd}{\hat\Sigma'}{\fsctx}{\macOX} \land
\oxabort \not\in \hat\Sigma'\, \land
\vint, \st \smodels \sst \land
\st, \cmd\baction_{\fictx} \result : \st'
\implies \\
\qquad \exsts{\sst', \vint' \ge \vint} (\result, \sst') \in \hat\Sigma' \land \vint', \st' \smodels \sst' \\[2mm]
\models (\fictx, \fsctx) \land
\csesemtransabstractm{\sst}{\cmd}{\sst'}{\fsctx}{\macUX}{\result} \land
\vint, \st' \smodels \sst' \implies
\exists \st.~\vint, \st \smodels \sst \land \st, \cmd \baction_{\fictx} \result : \st'
\end{array}
\]
\end{theorem}
\begin{proof}
Proofs of rules not related to function calls and predicates are the same as for Thm.~\ref{thm:ux-ox-sound}. OX soundness of function call follows from soundness of \mac (Prop.~\ref{prop:soundness}), OX completeness of \mac (Prop.~\ref{prop:mac-branch-comp}), and completeness of \produce (Prop.~\ref{prop:prod-compl}). UX soundness of function call follows from the soundness of \mac and \produce (Prop.~\ref{prop:soundness}) and UX completeness of \mac (Prop.~\ref{prop:mac-ux-comp}).
Full details are given in 
\ifarxiv
App.~\ref{app:symbolic-soundness-func}.
\fi
\end{proof}

%% file: sections/func-implementation.tex

\section{Consume and Produce Implementations}%
\label{sec:function-implementation}

We provide implementations for the \mac and \produce~operations and
prove that they satisfy the properties 1--7 of the axiomatic interface
(\S\ref{sec:axfun}). We give the complete set of rules implementing
these operations in
\ifarxiv
App.~\ref{app:mac-produce-soundness}
\fi
and only discuss the interesting rules here. Our implementations are
inspired by the Gillian OX implementations, although previous work has only given a brief informal sketch of these implementations~\cite{javert2}.

\subsection{Implementations}\label{sec:macMP}

\subparagraph*{Consume Implementation.} As is typical for SL-based analysis tools, our \mac~operation works with a fragment of the assertions with no implications, disjunctions, or existentials (which are handled outside consumption, see, e.g., the function-call rule); which means that input assertions for consumption are $\lstar$-separated lists of simple assertions. Following the implementation of Gillian, our \mac~operation works by consuming one simple assertion at a time and is split into two phases, a planning phase and a consumption~phase:
\[
\begin{array}{l}
\mac(m, P, \hat{\theta}, \sst) \defeq \textsf{let } mp \textsf{ = } \plan(\dom(\ssubst),P) \textsf{ in } \macMP(m, mp, \hat{\theta}, \sst)
\end{array}
\]
Here our interest lies in the consumption phase: the planning phase of Gillian has been formalised and discussed by Lööw et al.~\cite{matchplanning}. We, however, repeat the necessary background of the planning phase here to keep this paper self-contained.

\subparagraph*{Consumption Planning.}
The $\plan$~operation has two responsibilities: to resolve the order of
consumption and the unknown variables. The
operation takes a set of known logical variables (above,
$\dom(\ssubst)$) and an assertion $P$ to plan and returns a
\emph{matching plan} (\mps) of  the form $[(\Assert, [(\LVar, \LExp)])]$. An MP for an assertion $P=P_1 \lstar \ldots \lstar P_n$ ensures that (1) the simple assertions $P_i$ of  $P$ are consumed in an order such that the \emph{in-parameters} (\textit{ins}) of each simple assertion, i.e., the parameters (logical variables) that must be known to consume the simple assertion, have been learnt during previous consumption; (2) specifies how \emph{out-parameters} (\textit{outs}) are learnt during consumption, i.e., the remaining parameters (logical variables). For instance, the in-parameter of the cell-assertion $x\mapsto z+1$ is $x$ and the out-parameter is $z$, where the value of $z$ can be learnt by inspecting the heap and subtracting $1$. Another example is given by the pure assertion $x+1=y+3$: here, what the in- and out-parameters are depend on what variables are known, e.g., if we know $x$ we can learn $y$ and vice versa.
\begin{example}\label{ex:mp}
  Say we are to plan the assertion $x \le 10 \lstar x \mapsto y \lstar y = z - 10$ knowing that $\ssubst = \{ \svar x/x \}$, that is, $x$ is known but $y$ and $z$ are not. One MP for this assertion is $[(x \le 10, []), (x \mapsto y, [(y, \mathsf{O})]), (y =z-10, [(z,y+10)])]$, where $\mathsf{O}$ is used to refer to the value of the consumed heap cell. First, by consuming $x \le 10$ we learn nothing ($x$ is already known); second, when consuming $x \mapsto y$ we learn $y$ (from the consumed heap cell); third, since we learn $y$ in the previous step, we can learn $z$ by manipulating the assertion to $z=y+10$. Another MP is $[(x\mapsto y, [(y, {\sf O})]), (x\le 10, []), (y=z-10,[(z, y+10)])]$. Note that there is no MP starting with assertion $y=z-10$, because $y$ and $z$ are not known initially.
\end{example}

\subparagraph*{Consuming Assertions.} Having discussed the planning phase, we now discuss how pure assertions, cell assertions, and predicate~assertions~are~consumed.

\subparagraph*{Consuming Pure Assertions.}
Fig.~\ref{fig:mac-pure} contains the \macMP rules for consuming pure assertions.
The rules are defined in terms of the helper operation $\consPure(m, \spc, \spc') = \spc'' \mid abort$ which depends on the current reasoning mode $m$:
\begin{enumerate}[label=(\roman*)]
\item for $m=\macOX$, we check $\neg\sat(\spc \land \neg \spc')$ which is equivalent to $\spc \Rightarrow \spc'$, hence, the SAT check corresponds to the entailment check seen in traditional OX reasoning;

\item for $m=\macOX$, if $\sat(\spc \land \neg \spc')$, that is, $\neg(\spc \Rightarrow \spc')$, consumption, of course, must abort;

\item for $m=\macUX$, $\consPure$ instead cuts all paths of $\spc$ that are not compatible with the input pure assertion $\spc'$, i.e., forms $\spc \land \spc'$, and then checks if there are any paths left after the cut, i.e., checks $\sat(\spc \land \spc')$.
\end{enumerate}

\begin{example}\label{ex:ux_vs_ox}
To exemplify the difference between OX and UX consume, consider calling a function \texttt{foo(y)} with the precondition $\pvar{y} = y \lstar y \geq 0$. The first step of calling a function using its function specification is to consume its precondition, which we now illustrate. Say we are in a symbolic state with path condition $\spc = \sval > 5$ and are calling the function with an argument that symbolically evaluates to $\sval$, i.e., we know $\ssubst(y) = \sval$. In OX mode, the function call aborts: \macMP's pure consumption error rule is applicable because $\consPure(\macOX, \spc, \ssubst(y) \geq 10) = abort$ since $\sat(\spc \wedge \neg(\sval\geq 10))$. Intuitively, this means that not all paths described by $\spc$ are described by $y\geq 10$, i.e., we are ``outside'' the precondition of the function. Differently, in UX mode, a call to $\consPure(\macUX, \spc, \ssubst(y) \geq 10)$ cuts the incompatible paths by strengthen the path condition to $\spc \wedge \sval\geq 10$. That is, instead of as in OX mode where execution must abort, in UX mode the execution can continue.
\end{example}

\begin{figure}[t]
\small
\begin{mathpar}
\begin{aligned}
&\consPure(m, \spc, \spc') = \\[-0.4cm]
&\quad\begin{cases}
\spc \land \spc', & \text{if $m=\macUX$ and $\sat(\spc \land \spc')$} \\
\spc, & \text{if $m = \macOX$ and $\neg\sat(\spc \land \neg \spc')$} \\
abort, & \text{if $m = \macOX$ and $\sat(\spc \land \neg \spc')$}
\end{cases}
\end{aligned}
\quad
\begin{aligned}
\inferrule{
 P \text{ is pure } \quad outs=[(\lvar{x}_i, \lexp_i)|_{i=1}^n]\\\\
 \ssubst'=\ssubst \uplus \{( \ssubst (\lexp_i)/\lvar{x}_i)|_{i=1}^n\} \\\\
 \consPure(m,\spc,\ssubst'(P)) = \spc' }
 {\macMP( m, [(P, outs)], \ssubst ,\sst)\rightsquigarrow ( \ssubst', \sst[\sstupdate{pc}{\spc'}])) }
\end{aligned}

\inferrule
{P \text{ is pure }\\
 outs=[(\lvar{x}_i, \lexp_i)|_{i=1}^n]\\\\
 \ssubst'=\ssubst \uplus \{( \ssubst (\lexp_i)/\lvar{x}_i)|_{i=1}^n\} \\
   \consPure(\macEX, \spc, \ssubst'(P))= abort }
   {\macMP(\macEX, [(P, outs)], \ssubst, \sst) \rightsquigarrow abort([\mathstr{\sf consPure}, \ssubst'(P), \spc])}
\end{mathpar}

\vspace*{-0.40cm}
\caption{Rules for \consPure and \macMP (excerpt), where $\sst = (\ssto, \smemb,  \spc)$}\label{fig:mac-pure}
\end{figure}

\begin{figure}[t]
\small
    \begin{mathpar}
        \inferrule{\smem=\smem_f \uplus \{\sval_1 \mapsto \sval_2\}\\ 
          \spc'=\spc\wedge (\sval = \sval_1) \\ \sat(\spc')}
        {\conscell( \sval,\sst)
        \rightsquigarrow ({\sval_2}, \sst[\sstupdate{heap}{\smem_f}, \sstupdate{pc}{\spc'}])}
\and
        \inferrule{\sat(\spc \wedge \sval\notin \dom(\smem))}
         {\conscell(\sval,\sst) \rightsquigarrow \oxabort}
\and 
\ifarxiv
 \inferrule
 {
  \macMP(m, [(P,outs)], \ssubst, \sst)\rightsquigarrow abort(\sval)}
 {\macMP(m, (P,outs)::mp, \ssubst,  \sst)\rightsquigarrow abort(\sval)}
\and
\else
\fi

\infer{
\macMP(m, [(\lexp_a\mapsto \lexp_v, outs)],  \ssubst,  \sst) 
       \rightsquigarrow  (\ssubst', \sst'[ \sstupdate{pc}{\spc''}])
}
{
\begin{array}{lr}
\consPure(m, \spc, \ssubst(\lexp_a)\in \Nat) = \spc'& \text{check for evaluation error} \\
\conscell(\ssubst(\lexp_a), \sst[\sstupdate{pc}{\spc'}]) \rightsquigarrow (\sval,\sst') &  \text{branch over all cell consumptions}\\
      \textcolor{gray}{     \ssubst_{subst}=\{\sval/\mathsf{O}\}} &     \textcolor{gray}{ \text{substitution with cell contents}}\\ 
    \textcolor{gray}{   outs=[(\lvar{x}_i, \lexp_i)|_{i=1}^n] \text{ and } ((\ssubst\uplus \ssubst_{subst})(E_i)=\hat v_i)|_{i=1}^n }&  \textcolor{gray}{\text{collect and instantiate outs}}\\
     \textcolor{gray}{   \ssubst'=\ssubst \uplus \{(\hat v_i/x_i)|_{i=1}^n\} } &    \textcolor{gray}{  \text{extend substitution with outs}}\ \\
        \consPure(m, (\sst').\fieldsst{pc}, \ssubst' (\lexp_v)=\sval)= \spc'' &   \text{consume  cell contents}
\end{array}
}
\ifarxiv
\and
\inferrule
{\conscell( \ssubst(\lexp_a),\sst) = abort}
{\macMP(m, [(\lexp_a\mapsto \lexp_v, outs)],  \ssubst,\sst) \rightsquigarrow  abort([\mathstr{\sf MissingCell}, \ssubst(\lexp_a), \spc \land \ssubst(\lexp_a) \not\in \dom(\smem)])}
\else
\fi
\end{mathpar} 

\vspace*{-0.40cm}
\caption{Rules for \conscell and \macMP (excerpt), where $\sst = (\ssto, \smemb,  \spc)$}
\label{fig:mac-cell}
\end{figure}

\subparagraph*{Consuming Cell  Assertions.} Fig.~\ref{fig:mac-cell}
contains some of the \macMP rules for consuming a cell assertion. The rules are defined using the helper operation $\conscell(\hat v, \sst) \rightsquigarrow (\sval', \sst') \mid abort$, which tries to consume the cell at address~$\hat v$ in mode $m$ by branching over all possible addresses in the heap, returning the corresponding value in the heap, $\sval'$, and the rest of the state, $\sst'$, and returns $\oxabort$ if it is possible for the address $\hat v$ to point outside of heap. In the successful \macMP rule (featured in Fig.~\ref{fig:mac-cell}), the operation $\consPure$ is used to consume the contents of the matched cells.
\ifarxiv
The erroneous \macMP rules for consuming cell assertions propagate errors from the components of the successful rule; e.g., as shown in Fig.~\ref{fig:mac-cell}, if the \conscell call aborts, then an abort representing a missing-cell-resource~error~is~reported.
\fi
\ifarxiv
\begin{example}\label{ex:consP_modes} We now revisit the consumption example of \S\ref{ki:cse}. Recall, we start in the symbolic state $\sst=(\emptyset, \smem, \spc)$, where $\spc=(\hat{x} > 0 \land \hat{v} > 5)$ and $\smem = \{1\mapsto \hat{v}, 2\mapsto 10, 3\mapsto 100\}$. We are to consume $P = x \mapsto y \lstar y \geq 10$ given $\ssubst=\{ \hat{x}/x\}$. Further recall that we in total get four consumption branches in both OX and UX mode: one branch per address in heap (where, respectively, $\hat{x} = 1$, $\hat{x} = 2$, and $\hat{x} = 3$) and one missing-cell branch (where $\hat{x} \not\in \{1, 2, 3\}$). Now, using the rules we have introduced for \macMP, we explain in more detail why the first branch of the example differs in OX and UX mode, that is, the branch where $\hat{x} = 1$ (called branch 0 and branch 1 in the example). First, note that the only MP for $P$ is $[(x \mapsto y, [(y,{\sf O})]), (y\geq 10, [])]$. Hence, $x \mapsto y$ is consumed first, using the \macMP rules for cells just introduced. Both the OX and UX consumption result in the same set of states: \conscell branches over all heap addresses and the missing-cell case. For the branch with $\hat{x} = 1$, we learn that $\ssubst(y) = \hat{v}$ (from the value of the matching cell in the heap). It is now time to consume $y \geq 10$, using the \macMP rules for pure consumption. From the definition of the $\consPure$~operation: in OX mode, the operation returns $\oxabort$ because it is satisfiable that $\hat{v} > 5$ and $\neg(\hat{v} \geq 10)$, whereas in UX mode, the path condition is strengthened with $\hat{v} \geq 10$ and execution~can~continue.
\end{example}
\else

\fi

\subparagraph*{Consuming Predicate Assertions.} The \macMP rules for predicate assertions are generalisations of the rules for cell assertions, with two main differences: predicates may have multiple \textit{ins} and \textit{outs} whereas cells have a single \textit{in} 
and single \textit{out},
and predicate assertions refers to symbolic predicates whereas cell assertions refer to the symbolic heap. 

\subparagraph*{Produce Implementation.}\label{sec:produce}

\begin{figure}[!t]
\begin{mathpar}
\small
\inferrule
{\smem' \defeq \smem \uplus \{\svar{l}\mapsto \svar{v}_\cfreed\} \\\\
\spc' \defeq \spc\wedge \svar{l}\notin \domain(\smem) \quad \sat(\spc') \\\\
\sst' \defeq \sst[\sstupdate{heap}{\smem'}, \sstupdate{pc}{\spc'}]}
{\funcAlg{prodCell}{\svar{l}, \svar{v}_\cfreed,\sst}= \sst'}
\quad
\inferrule
{P ~\text{pure} \\\\ \spc' \defeq \spc \wedge  \ssubst(P) \quad \sat(\spc') \\\\ \sst' \defeq \sst[\sstupdate{pc}{\spc'}]}
{\funcAlg{produce}{P,\ssubst, \sst}\rightsquigarrow \sst'}
\quad
\inferrule
{ \spc' \defeq \spc \wedge \ssubst(\lexp_a)\in \Nat\wedge \ssubst(\lexp_v)\in \Val \\\\
\prodcell(\ssubst(\lexp_a), \ssubst(\lexp_v),\sst[\sstupdate{pc}{\spc'}]) =\sst'}
{\produce(\lexp_a\mapsto \lexp_v, \ssubst, \sst) \rightsquigarrow \sst'}
\end{mathpar}

\caption{Rules for \prodcell and \produce (excerpt), where $\sst = (\ssto, \smemb, \spc)$ and $\sym{v}_\cfreed$ denotes a symbolic value or $\cfreed$}\label{fig:part-produce}
\end{figure}

The implementation of \(\produce(Q, \ssubst, \sst)\) is straightforward: it extends $\sst$ with the symbolic state corresponding to $Q$ given $\ssubst$, ensuring well-formedness. Unlike \mac, \produce does not require planning and is not dependent on the mode of execution. An excerpt of rules for \produce is given in Fig.~\ref{fig:part-produce}. Like \mac, \produce does not support assertion-level implications, which are usually not found in function specifications or predicate definitions. However, \produce, unlike \mac, supports assertion-level disjunctions since the function specification we synthesise using bi-abduction contains disjunctions~(cf.~\S\ref{sec:applications}).

\subsection{Correctness of Implementations}

The correctness of the consume and produce implementations amount to
showing that they satisfy properties of the axiomatic interface for consume and produce.
\ifarxiv
The proofs are given in App.~\ref{sec:macMP-correct} for \mac and App.~\ref{sec:produce-correct} for \produce.
\else
\fi
\begin{theorem}[Correctness]\label{thm:cp_correct}
The \mac and \produce~operations satisfy properties \ref{prop:wf}-\ref{prop:prod-compl}~(\S\ref{sec:axfun}).
\end{theorem}

%% file: sections/bi-abduction.tex

\section{Bi-abduction}%
\label{sec:bi-abduction}

To enable hosting Pulse-style true bug-finding on top of our CSE engine, it must support UX bi-abduction. In this section, we show how the engine presented in \S\ref{sec:function-interface} can be extended to support UX bi-abduction by catching missing-resource errors that happen during execution and applying \emph{fixes} to  enable uninterrupted  execution instead of faulting. These fixes, stored in an \emph{anti-frame}, add the missing resource to the current state and allow execution to continue. This style of bi-abduction was introduced in the OX setting by JaVerT 2.0~\cite{javert2}. Here we show that it can also be applied to the UX setting. We focus on UX bi-abduction for true bug-finding, but also discuss in \S\ref{sec:applications} how the obtained UX results can be used in~an~OX~setting.

We introduce the  judgement for the bi-abductive symbolic~engine,
$
\biab{\sst}{\scmd}{(\sst',\smemb)}{\fsctx}{\outcome}, 
$
with outcomes $o ::= \oxok \mid \oxerr$, and the anti-frame $\smemb=(\smem,
\sps)$ containing the anti-heap $\smem$ and the
anti-predicates  $\sps$. We do not need $\omiss$ or $\oxabort$ as possible outcomes since if they happen during execution they will be either fixed by bi-abduction or cut if not. The new judgement is defined in terms of the judgement $\csesemtransabstractm{\sst}{\scmd}{\sst'}{\fsctx}{\macUX}{\outcome}$ and a partial function \texttt{fix} as:
\begin{mathpar}
\footnotesize
\inferrule[\textsc{Biab}]
 {\csesemtransabstractm{\sst}{\scmd}{\sst'}{\fsctx}{\macUX}{\outcome} \\\\ \mathsf{not\_Seq}(\scmd) \quad \outcome \notin \{ \omiss, \oxabort \}}
 {\biab{\sst}{\scmd}{(\sst', (\emptyset, \emptyset))}{\fsctx}{\outcome}}
\and
\inferrule[\textsc{Biab-Miss}]
 {\csesemtransabstractm{\sst}{\scmd}{\sst'}{\fsctx}{\macUX}{\outcome} \quad
  \mathsf{not\_Seq}(\scmd) \quad
  \outcome \in \{ \omiss, \oxabort \} \\\\
  \texttt{fix}(\sst') = (\smemb, \spc) \quad
  \biab{\sst \cdot (\smemb, \spc)}{\scmd}{(\sst'', \smemb')}{\fsctx}{\outcome'}}
 {\biab{\sst}{\scmd}{(\sst'', \smemb \cup \smemb')}{\fsctx}{\outcome'}}
\and
\inferrule[\textsc{Biab-Seq-Err}]
 {\biab{\sst}{\scmd_1}{(\sst', \smemb)}{\fsctx}{\outcome} \quad \outcome \not= \osucc}
 {\biab{\sst}{\scmd_1; \scmd_2}{(\sst', \smemb)}{\fsctx}{\result}}
\and
\inferrule[\textsc{Biab-Seq}]
 {\biab{\sst}{\scmd_1}{(\sst', \smemb_1)}{\fsctx}{\osucc} \qquad \biab{\sst'}{\scmd_2}{(\sst'', \smemb_2)}{\fsctx}{\result} \\\\ 
  \svs{\domain(\smemb_2)}\cap (\svs{\sst'} \setminus \svs{\sst}) = \emptyset}
 {\biab{\sst}{\scmd_1; \scmd_2}{(\sst'', \smemb_1\cup \smemb_2))}{\fsctx}{\result}}
\end{mathpar}
where $\mathsf{not\_Seq}(\scmd)$ denotes that $\scmd$ is not a sequence command (i.e., does not have the form~$\scmd_1; \scmd_2$), $\sst \cdot (\smemb, \spc)$ denotes $\sst \cdot (\emptyset, \smemb, \spc \land\sinvc(\smemb))$, $\domain(\smemb_2)=\svs{\domain(\smem_2)}  \cup  \svs{\domain(\sps_2)}$, and $\domain(\sps)$ denotes all symbolic variables of the $\ins$ of the predicates in $\sps$. The rule \textsc{Biab} states that for non-erroneous outcomes, the bi-abductive engine has the same semantics as the underlying UX engine it is built on top of. The rule \textsc{Biab-Miss}, which is the most interesting rule, catches missing-resource errors from the underlying UX engine and uses the \texttt{fix} function to add the missing resource to the current symbolic state and anti-frame, such that execution can continue. The two rules \textsc{Biab-Seq-Err} and \textsc{Biab-Seq} are two straightforward sequencing rules for the engine, where the symbolic-variable condition of \textsc{Biab-Seq} ensures that the anti-frame $\smemb_2$ does not clash with resource allocated by $\scmd_1$.

To exemplify, say the engine is in symbolic state $(\{ \pvar v \mapsto 0, \pvar a \mapsto 13 \}, (\emptyset, \emptyset), \true)$ and is about execute $\pderef{\pvar v}{\pvar a}$, i.e., about to retrieve the value of the heap cell with address $\pvar a$. Since this cell is not in the heap, the rule \textsc{Lookup-Err-Missing} from Fig.~\ref{fig:sexrules} is applicable, which sets the variable $\pvar{err}$ to $[\mathstr{\mathsf{MissingCell}}, \mathstr{a}, 13]$ and gives outcome $\omiss$. Now, in the rule \textsc{Biab-Miss}, given the data in the $\pvar{err}$ variable, the \texttt{fix} function constructs a fix $((\{ 13 \mapsto \hat{v} \}, \emptyset), \true)$ where $\hat{v}$ is a fresh variable. The rule adds this fix to both the current symbolic state and the anti-frame, resulting in the symbolic state $(\{ \pvar v \mapsto \hat{v}, \pvar a \mapsto 13 \}, (\{ 13 \mapsto \hat{v} \}, \emptyset), \true)$ and outcome $\oxok$. Other cases are similar. E.g., when abort outcomes from \mac represent missing resource (e.g., when invoked in a function call), \texttt{fix} returns the resources needed for the~execution~to~continue.
The following theorem captures the essence of bi-abduction:
\ifarxiv
its proof is given in~App.~\ref{app:bi-abduction-correct}:
\fi
\begin{theorem}[CSE with Bi-Abduction: UX Soundness]\label{thm:bi-abduction}
\[
\biab{\sst}{\scmd}{(\sst', \smemb)}{\fsctx}{\outcome} \implies \csesemtransabstractm{{\sst \cdot (\smemb,\true)}}{\scmd}{\sst'}{\fsctx}{\macUX}{\outcome}
\]
\end{theorem}

%% file: sections/applications.tex

\section{Analysis Applications}\label{sec:applications}

We discuss the three analysis applications we have built on top of our unified CSE engine, to demonstrate its wide applicability: EX whole-program automatic symbolic testing (\S\ref{sec:testing}); OX semi-automatic verification (\S\ref{sec:verification}); and UX automatic true bug-finding (\S\ref{sec:bug-finding}).

We have gathered these three analyses from different corners of the literature. EX symbolic testing is well understood in the first-order symbolic execution literature. OX verification is well understood in the consume-produce symbolic execution literature. However, one novelty here is that the correctness proof of the analysis is established with respect to our axiomatic interface rather than the consume/produce operations directly, allowing us to show that function specifications are valid w.r.t. the standard SL definition of validity. In contrast to the other two analyses, UX bug-finding has not previously been implemented in consume-produce style, making this a novel contribution.
To simplify the presentation, we consider only non-recursive functions. All applications can be extended to handle bounded recursion by adding a fuel parameter. Unbounded recursion can be handled in verification via user-provided annotations, but is not a good fit for automatic bug-finding.

\subsection{EX Whole-program Symbolic Testing}%
\label{sec:testing}

Our EX core engine allows us to implement simple non-compositional analyses, such as whole-pro\-gram symbolic testing, in the style of CBMC~\cite{cbmc} and Gillian~\cite{gillianpldi}. For this analysis, we augment the input language with three additional commands: $\pvar x := \mathsf{sym}$, for creating symbolic variables; $\passume{\pexp}$, for imposing a constraint $\pexp$ on the current state; and $\passert{\pexp}$, for checking that $\pexp$ is true in the current state.
The operational semantics for these commands is given
\ifarxiv
in App.~\ref{app:symbolic}.
\fi

The testing algorithm is as follows. Given a command $C$ and implementation context $\fictx$, the analysis starts from the state $\sst \defeq (\{ \pvar x \mapsto \nil \mid \pvar x \in \pv C \}, \emptyset, \true)$, and executes $C$ to completion. 
The analysis reports back any violations of the $\mathsf{assert}$ commands encountered during execution. Given the core engine is UX sound, any bug found will be a true bug. Moreover, if the analysed code contains no unbounded recursion, given the core engine is OX sound, all existing bugs will be found modulo the ability of the underlying SMT solver. 

\subsection{OX Verification}%
\label{sec:verification}

We formalise an OX verification procedure, $\texttt{verifyOX}$, on top of our CSE engine. 
Given a specification context $\fsctx$, a function $\pfunction{f}{\vec{\pvar{x}}}{\scmd; \preturn{\pexp}}$ with $f \notin \dom(\Gamma)$, and an OX specification $t_f = \quadruple{\vec{\pvar{x}} = \vec{x} \lstar P}{\!\!\!}{\Qok}{\Qerr}$, if $\texttt{verifyOX}(\fsctx, f, t_f)$ terminates successfully, then we can soundly extend $\fsctx$ to $\fsctx' = \fsctx[f \mapsto t_f]$. 
\ifarxiv
The algorithm is given below and the proof in App.~\ref{app:applications-proofs}:
\fi
\begin{enumerate}[leftmargin=*]
\item Let $\ssubst \defeq \{ \svar{x}/x \mid x \in \lv{\vec{\pvar{x}} = \vec{x} \lstar P} \}$, $\ssto \defeq \{ \pvar{x} \mapsto \svar{x} \mid \pvar{x} \in \vec{\pvar{x}} \} \cup \{ \pvar{x} \mapsto \nil \mid \pvar{x} \in \pv{\scmd} \setminus \vec{\pvar{x}} \}$, and $\sst = (\ssto, \emptyset, \true)$.
\item Set up symbolic state corresponding to pre-condition: $\produce(P, \ssubst, \sst) \rightsquigarrow \sst'$.
\item Execute the function to completion: $\csesemtranscollectm{\sst'}{\cmd; \passign{\pvar{ret}}{\pexp}}{\hat\Sigma'}{\fsctx}{\macOX}$. 
Then, for every $(\outcome, \sst'') \in \hat\Sigma'$:
\begin{enumerate}[label=(\alph*)]
\item If $\outcome = \omiss$ or $\outcome=\oxabort$, abort with an error. 
\item If $\outcome = \osucc$, then let $\ssubst' =  \ssubst \uplus \{(\sst''.\texttt{sto})(\pvar{ret})/r\}$ for a fresh $r$ and let $Q' = \Qok\{r/\pvar{ret}\}$. Otherwise, $\outcome = \oerr$, in which case let $\ssubst' = \ssubst$ and $Q' = \Qerr$.
\item Consume the post-condition: $\mac(\macEX, Q'', \ssubst', \sst'') \rightsquigarrow (\ssubst'', \sst''')$, where 
 $Q' = \exists \vec{y}.~Q''$.
\item If consumption fails or the final heap is not empty, abort with an error.\footnote{This check is required due to our classical (linear) treatment of resource, appropriate for languages with explicit deallocation rather than garbage collection.}
\end{enumerate}
\end{enumerate}

\subsection{UX Specification Synthesis and True Bug-finding}%
\label{sec:bug-finding}

\newcommand{\toasrt}[0]{\mathsf{toAsrt}}

Recall that Pulse-style UX bug-finding is powered by UX specification synthesis, where, after appropriate filtering, synthesised erroneous specification can be reported as bugs. UX specification synthesis, in turn, is powered by UX bi-abduction (as introduced in \S\ref{sec:bi-abduction}).

To formalise the specification synthesis procedure, we first define the $\toasrt$ function, which takes a symbolic state $\sst = (\ssto, \smemb, \spc)$, where $\smemb=(\smem,\sps)$, and returns the corresponding assertion. The function is simple to implement: $\ssto$ becomes a series of equalities,  $\smem$ becomes a series of cell assertions, $\sps$ are lifted to predicate assertions and $\spc$ is lifted to a pure assertion. Using $\toasrt$, we can also transform multiple symbolic states into an assertion by transforming them individually and gluing together the obtained~assertions~using~disjunction.

We generate UX function specifications using the $\mathsf{synthesise}(\fsctx, f, P)$ algorithm, which takes a specification context $\fsctx$, a function $\pfunction{f}{\vec{\pvar{x}}}{\scmd; \preturn{\pexp}}$ and its candidate pre-condition, $\vec{\pvar x} = \vec x \lstar P$, and uses bi-abduction to generate a set of UX specifications describing the behaviour of $f$ starting from~$P$. As $P = \emp$ is a valid starting point, \textsf{synthesise} can be applied to any function without a priori knowledge. The \textsf{synthesise} algorithm is as~follows:
\begin{enumerate}[leftmargin=*]
\item Let $\ssubst \defeq \{  \svar{x}/x \mid x \in \lv{\vec{\pvar{x}} = \vec{x} \lstar P} \}$, $\ssto \defeq \{ \pvar{x} \mapsto \svar{x} \mid \pvar{x} \in \vec{\pvar{x}} \} \cup \{ \pvar{x} \mapsto \nil \mid \pvar{x} \in \pv{\scmd} \setminus \vec{\pvar{x}} \}$, and $\sst = (\ssto,  \emptyset, \true)$.
\item Add the symbolic representation of $P$ to $\sst$: $\produce(P, \ssubst, \sst) \rightsquigarrow \sst'$
\item Execute the function, obtaining a set of traces: $\csesemtranscollectm{\sst'}{\cmd, \pvar{ret} := \pexp}{\{ (o_i, (\sst_i', \smemb_i))|_{i \in I} \}}{\fsctx}{\text{bi}}$.
\item Then, for every obtained $(o_i, ((\ssto_i', \smemb_i', \spc_i'), \smemb_i))$:
\begin{enumerate}
\item Complete the candidate pre-condition: $P_i \defeq P \lstar \toasrt((\emptyset, \smemb_i, \true))$.
\item Restrict the final store to the return/error variable: $\ssto_i'' \defeq \ssto_i'|_{\{\pvar x\}}$, where $\pvar x = \pvar{ret}$ if $o_i = \oxok$ and $\pvar x = \pvar{err}$ otherwise.
\item Create the post-condition: $Q_i \defeq \toasrt(\ssto_i'', \smemb_i', \spc_i')$.
\item Return $\isltriplex{P_i}{\fid(\vec{\pvar{x}})}{o_i}{\exists \vec y.~Q_i}$, where $\vec y \defeq \lv{Q_i} \setminus \lv{P_i}$.
\end{enumerate}
\end{enumerate}
%
%
\begin{theorem}[Correctness of ${\sf synthesise}$]\label{thm:synthesise-correct}
\[
\isltriplex{P'}{\fid(\vec{\pvar{x}})}{o}{\exists \vec y.~Q} \in {\sf synthesise}(\fsctx,f,P) \implies \fsctx \models \isltriplex{P'}{\fid(\vec{\pvar{x}})}{o}{\exists \vec y.~Q}
\]
\end{theorem}
\ifarxiv
\begin{proof}
We present the full proof in App.~\ref{app:applications-proofs}. In essence, the proof follows from the definition of UX specification internalisation and the soundness of bi-abductive execution (Thm.~\ref{thm:bi-abduction}).
\end{proof}
\else
\fi

\begin{remark}
\textnormal{Step 4b corresponds to forgetting the local variables when moving from internal to external post-condition since symbolic states only have program variables in the store.}
\end{remark}

\begin{remark} 
\textnormal{Front-end heuristics to filter out ``interesting'' bugs, i.e., synthesised erroneous specification, can be easily implemented on top of our bi-abduction (e.g., filtering for manifest bugs as per Lee et al.~\cite{Le22}); for this paper, however, we are foremost interested in back-end engine development and therefore consider such front-end issues~out~of~scope.}
\end{remark}

\begin{remark} 
\textnormal{Specifications with the same anti-frame can be coalesced into one via disjunction of their post-conditions. Moreover, if a specification does not branch on symbolic variables created by the execution of $C$, the pure part of the post-condition can be lifted to the pre-condition to create an EX specification~\cite{esl}, which can then be used both in OX verification and UX true bug-finding.}
\end{remark}

\begin{remark} 
\textnormal{Automatic predicate folding and unfolding may be required in some cases to prevent redundant fixes: e.g., if $\pfuncall{\pvar{y}}{\pvar{g}}{}; \pderef{\pvar{x}}{\pvar{y}}$ and $\pvar{g}$ has post-condition $\llist{\pvar{ret}, \lstvs}$, the list predicate should be unfolded for the lookup to access the first value in the list. Gillian has heuristics-based automatic folding and unfolding, but we leave its description and the evaluation of its compatibility with bi-abduction for future work. Without automatic folding and unfolding, code must not break the interface barrier of the data structures it~uses.}
\end{remark}

%% file: sections/evaluation.tex

\section{Evaluation}%
\label{sec:evaluation}

We have evaluated our CSE engine in the following two practical ways. 

\subsection{Companion Haskell Implementation}
With our  engine formalism, we have developed a companion Haskell implementation to demonstrate that the formalism is implementable (and to catch errors early by executing simple examples). The Haskell implementation follows the inference rules of $\csesemtransabstractm{\sst}{\cmd}{\sst'}{\fsctx}{m}{\result}$ given in \S\ref{sec:function-interface}, and the specific $\mac$ and $\produce$ operations in~\S\ref{sec:function-implementation}. We have implemented the search through the inference rules inside a symbolic execution monad, similar to other monads in the literature~\cite{darais:icfp:2017,Mensing19}; the monad handles, e.g., demonic non-determinism (branching), angelic non-determinism (backtracking), per-branch state and global state.

\subsection{Gillian OX and UX Compositional Analysis Platform} 

Our unified CSE engine took direct inspiration from the Gillian compositional OX platform.
Using the ideas presented here, we returned to Gillian and adapted its CSE engine to handle both SL and ISL function specifications with real-world consume-produce implementations. Leveraging the identified difference between 
OX and UX reasoning, we were able to introduce UX reasoning to Gillian by adding, in essence, a OX/UX flag to the corresponding function. As these changes were isolated, existing analyses implemented in Gillian remain unaffected, including Gillian's whole-program symbolic testing, previously evaluated on the Collections-C library~\cite{gillianpldi}, and Gillian's compositional OX verification, previously evaluated on AWS code~\cite{gilliancav}. In addition, we have implemented UX bi-abduction in Gillian, following the fixes-from-errors approach presented in \S\ref{sec:bi-abduction}, where functions are evaluated bottom-up along the call graph and previously generated specifications are used at call sites. 

\newcolumntype{H}{>{\setbox0=\hbox\bgroup}c<{\egroup}@{}}
\begin{table}[!t]
\centering
\footnotesize
\begin{tabular}{@{}lrrHrrr@{}}
\toprule
\textbf{Library} & \textbf{Functions} & \textbf{GIL Inst.} & \textbf{Tests} & \textbf{Succ. Specs} & \textbf{Err. Specs} & \textbf{Time (s)} \\ \midrule
array & 45 & 1784 & 204 & 251 & 260 & 1.36 \\
deque & 47 & 2312 & 213 & 271 & 210 & 2.25 \\
hashset & 14 & 160 & 46 & 7 & 112 & 9.43 \\
hashtable & 28 & 1527 & 97 & 31 & 147 & 15.67 \\
list & 66 & 2977 & 314 & 454 & 615 & 5.59 \\
pqueue & 10 & 557 & 28 & 90 & 51 & 3.96 \\
queue & 16 & 85 & 92 & 133 & 67 & 1.36 \\
rbuf & 9 & 181 & 22 & 9 & 17 & 0.07 \\
slist & 52 & 2269 & 278 & 292 & 1873 & 24.49 \\
stack & 16 & 85 & 92 & 136 & 88 & 0.50 \\
treeset & 17 & 214 & 68 & 28 & 106 & 0.36 \\
treetable & 36 & 1601 & 154 & 144 & 276 & 1.55 \\
other & 8 & 139 & 21 & 14 & 11 & 0.03 \\
\textbf{Total} & \textbf{364} & \textbf{13891} & \textbf{1629} & \textbf{1860} & \textbf{3833} & \textbf{66.62} \\ \bottomrule
\end{tabular}
\vspace*{0.2cm}
\caption{Aggregated results of synthesising function specifications for the Collections-C library (commit \texttt{584e113}). Results were obtained by setting the loop and recursive call unrolling limit to 3, on a MacBook Pro 2019 laptop with 16~GB memory and a 2.3~GHz Intel Core i9 CPU.}%
\label{tbl:collections-c-results}
\vspace{-0.3cm}
\end{table}

To evaluate Gillian's new support for UX reasoning, we have tested its new UX bi-abduction analysis on real-world code, specifically, the Collections-C~\cite{collections} data-structure library for C. As discussed in \S\ref{sec:applications}, specification synthesis using UX bi-abduction constitutes the back-end of Pulse-style bug-finding and is its most time-consuming part. The Collections-C library has 2.6K stars on GitHub and approximately 5.2K lines of code, and it uses many C constructs and idioms such as structures and pointer arithmetic. The data structures it provides include, e.g., dynamic arrays, linked lists, and hash tables. To carry out the evaluation, we extended previous work  where the Gillian platform has been instantiated to the C programming language, called Gillian-C. Tbl.~\ref{tbl:collections-c-results} presents the results of our new UX bi-abduction analysis, grouped by the data structures of the library: the numbers of associated functions; number of corresponding GIL instructions (GIL is the intermediate language used by Gillian); the number of success and error specifications; and the analysis time. Since one specification is synthesised per execution path, the number of specifications reflect the number of execution paths the Gillian engine was able to construct using bi-abduction. In summary, Gillian-C synthesises specifications for 364 functions of the Collections-C library, producing 5693 specifications in 66.92 seconds. We believe the results are promising both in terms of performance and number of specifications synthesised. One anomaly is that 58\% of the execution time is spent on 3 of the 343 functions, leading to the creation of 1640 specifications. This anomaly arises because of the memory model currently in use by Gillian-C and not from a limitation of our formalisation or of the Gillian engine. More detailed analysis and a selection of generated specifications 
\ifarxiv
can be found in~App.~\ref{app:gillian}.
\fi

%% file: sections/related-work.tex

\section{Related Work}\label{sec:related-work}

\subparagraph*{First-order Compositional Symbolic Execution.} Static symbolic execution tools and frameworks based on first-order logic, such as CBMC~\cite{cbmc} and Rosette~\cite{torlak:pldi:2014,Porncharoenwase22},
 can be made functionally compositional with respect to the variable store but not with respect to arbitrary state
 because they are not able to specify functions that manipulate memory in a way that would make the reasoning scalable.

\subparagraph*{Compositional Symbolic Execution.} We work with a  CSE engine with consume and produce operations, as found in, e.g., the OX tools VeriFast~\cite{verifast}, Viper~\cite{Muller16}, and Gillian~\cite{gillianpldi,gilliancav}. In contrast, an alternative approach is to describe CSE inside a separation logic using proof search, as found in, e.g., Smallfoot~\cite{berdine:fmco:2005,berdine:aplas:2005}, Infer~\cite{calcagno:nasa:2011}, and Infer-Pulse~\cite{Le22}.

Here, we focus on formalising a CSE  engine with consumers and producers, which more accurately models tool implementations. All the current consume-produce tools are based on OX reasoning. Some have detailed work on formalisation: Featherweight VeriFast~\cite{Jacobs15} provides a Coq mechanisation inspired by VeriFast; Schwerhoff's PhD thesis~\cite{Schwerhoff16} and Zimmerman et al.~\cite{Zimmerman24} provide detailed accounts of Viper's symbolic execution backend.
Previous work has not, like us, introduced an axiomatic interface for their consume and produce operations. Because of the interface, our results are established using function specifications whose meaning is defined in standard SL/ISL-style, in particular, using the standard satisfaction relation for assertions defined independent of the choice of our CSE engine. This means that we can use specifications developed outside of our engine, e.g., using theorem provers, and vice versa. In contrast, the work on Featherweight VeriFast does not define an assertion satisfaction relation independent of their consume and produce operations. Schwerhoff does not give a soundness theorem at all. Lastly, Zimmerman et al. give a standard satisfaction relation (for implicit dynamic frames~\cite{Smans09}, a variant of SL~\cite{Parkinson11}) but only embed this relation inside their concrete semantics instead of working with the standard definitions of function specifications (see Fig.~11 of their paper). Finally, we have demonstrated that our engine semantics provides a common foundation for OX and UX reasoning, with the difference in the underlying engine only amounting to the choice to use satisfiability or validity.
This allows a straightforward extension of Gillian to support UX reasoning.

\subparagraph*{Bi-abduction.} Bi-abduction was originally introduced for  OX reasoning~\cite{calcagno:popl:2009,calcagno:jacm:2011} and led to Meta's automatic Infer tool for bug-finding~\cite{calcagno:nasa:2011}. Recently, it was reworked for UX  reasoning and led to Meta's Infer-Pulse for true bug-finding~\cite{isl,Le22}. In these works, compositional symbolic execution is formalised using proof search, in the style of the  Smallfoot description of symbolic execution~\cite{berdine:fmco:2005,berdine:aplas:2005}, with bi-abduction embedded into that proof search.
In contrast, our UX  bi-abduction is formulated as a separate layer on top of our CSE engine, establishing fixes from missing-resource errors using an idea  introduced for OX bi-abduction by~JaVerT~2.0~\cite{javert2}.

\subparagraph*{Alternating between OX and UX Reasoning.} Smash by Godefroid et al.~\cite{smash} is the most well-known tool to combine OX and UX reasoning. It is a  first-order tool which  ``alternates'' between OX and UX reasoning 
  to speed up the program analysis implemented by the tool.
   However, 
citing Le et al.~\cite{Le22},  who in turn report on personal communication with Godefroid, Smash-style analyses seem to have faced obstacles when put into practice in that they  were ``used in production at Microsoft, but are not used by default widely in their deployments, because other techniques were found which were better for fighting path explosion.'' 
Our CSE engine, in contrast,  has a completely different motivation in that its purpose is to  host different types of OX and UX analyses.

\subparagraph*{Program Correctness and Incorrectness.} We know of two program logics that work with both program correctness  and incorrectness. Exact separation logic (ESL)~\cite{esl} combines the guarantees of both SL and ISL, providing exact function specifications compatible with both OX verification and UX true bug-finding. Exact specifications are compatible with our CSE engine, e.g., in UX mode the engine can call exact unbounded function specifications of list algorithms and still preserve true bug-finding.
Outcome logic (OL)~\cite{Zilberstein23} is based on OX Hoare logic with a different approach to handling incorrectness based on the reachability of sets of states. No tool is currently based on OL; in the future, we hope to be able to extend straightforwardly our unified approach to~incorporate~OL.

%% file: sections/conclusion.tex

\section{Conclusions}


We have introduced a  compositional symbolic execution engine
capable of creating and using function
specifications arising from an underlying separation logic.
Our engine is formally defined using a novel axiomatic interface which ensures a sound link between
the execution engine and the function specifications using consume and produce operations.
Thus, our engine creates function specifications usable by other tools, and uses function specifications 
from various sources, including theorem provers and 
pen-and-paper proofs. Additionally, we have captured the essence of the Gillian consume and produce implementations both operationally, using inference rules, and via an accompanying Haskell implementation, and shown that our operational description satisfies the properties of the axiomatic interface.
In this way, we offer a degree of assurance that the real-world, heavily-optimised Gillian implementation is correct.
 
A surprising property of our semantics is that it provides a common foundation for both OX reasoning based on SL, and UX reasoning based on ISL. By leveraging the minimal  differences between the OX and UX engines,
we have extended the OX Gillian platform to support UX  reasoning.  This extension includes function specifications underpinned by ISL,  enabling automatic true bug-finding using UX  bi-abduction which our engine incorporates by creating fixes from missing-resource errors. 
We evaluate our extension using the Gillian instantiation to C, 
the first real-world tool to support both compositional correctness and incorrectness reasoning, 
grounded on a common formal compositional symbolic execution engine.
Our instantiation preserves the previous OX verification evaluated on AWS code~\cite{gilliancav} and now 
automatically synthesises UX function specifications for the real-world Collections-C library using our UX bi-abduction technique.

We believe that our axiomatic interface and formalisation of UX bi-abduction serve as re-usable techniques, which we hope will provide valuable guidance for the implementation of~the next-generation~compositional~symbolic~execution~engines.

%% file: sections/app-concrete.tex
\section{Concrete Semantics}\label{app:concrete}

This appendix provides an exhaustive list of the concrete semantics rules. The corresponding judgement is:
\[
	\st, \scmd \baction_{\fictx} \outcome: \st'
\]
which, informally, means that executing command $\scmd$ starting with state $\st$ ends in state $\st'$ with outcome $\outcome := \osucc \mid \oerr \mid \omiss$, where function calls are resolved using the function context $\fictx$.
 
\subsection{Expression Evaluation}

Expression evaluation is trivially defined, we provide some illustrative cases:
\[
\footnotesize
\begin{array}{rcl}
\esem{\gv}{\sto} & = &\gv \\
\esem{\pvar{x}}{\sto}& = &\sto(\pvar{x})\\
\esem{\pexp_1 ~{+}~ \pexp_2}{\sto} & = &
	\begin{cases}
		\esem{\pexp_1}{\sto} + \esem{\pexp_2}{\sto}, &
		\esem{\pexp_1}{\sto} \in \nats, \esem{\pexp_2}{\sto} \in \nats\\
		\undefd, & \mbox{otherwise} 
	\end{cases}\\
\esem{\pexp_1 ~{/}~ \pexp_2}{\sto} & = &
	\begin{cases}
		\esem{\pexp_1}{\sto} / \esem{\pexp_2}{\sto}, &
		\esem{\pexp_1}{\sto} \in \nats, \esem{\pexp_2}{\sto} \in \nats,\esem{\pexp_2}{\sto} \neq 0 \\
		\undefd, & \mbox{otherwise} 
	\end{cases}\\
\end{array}
\]

\subsection{Core Rules}
{\footnotesize
\begin{mathparpagebreakable}
       \infer{
       \st, \pskip \baction_{\fictx} \oxok: \st
       }{ \ } \and
       \infer{
          \sthreadp{ \sto }{ \hp }, \passign{\pvar{x}}{\pexp} \baction_{\fictx} \oxok: \sthreadp{\sto [\pvar{x} \storearrow v] }{ \hp }
       }{
       \esem{\pexp}{\sto} = v
       } \and
       \infer{
          \sthreadp{ \sto }{ \hp }, \passign{\pvar x}{\pexp} \baction_{\fictx} {\oxerr} : \sthreadp{ \sto_{\oxerr} }{ \hp }
       }{
       \begin{array}{c}
       \esem{\pexp}{\sto} = \undefd \quad \verr = [``\mathsf{ExprEval}", \stringify {\pexp}]
       \end{array}
       } \and
       \infer{
          \sthreadp{ \sto }{ \hp }, \passign{\pvar{x}}{\prandom} \baction_{\fictx} \oxok:\sthreadp{\sto [\pvar{x} \storearrow n] }{ \hp }
       }{
       n \in \Nat
       } \and
       \infer{
       \sthreadp{ \sto }{ \hp }, \perror(\pexp) \baction_\fictx \oxerr :  \sthreadp{ \sto_{\oxerr} }{ \hp }
       }{\esem{\pexp}{\sto} = \gv \quad \verr = [``\mathsf{Error}", \gv]}
       \and
       \infer{
          \sthreadp{ \sto }{ \hp }, \perror(\pexp) \baction_{\fictx} {\oxerr} : \sthreadp{ \sto_{\oxerr} }{ \hp }
       }{
       \begin{array}{c}
       \esem{\pexp}{\sto} = \undefd \quad \verr = [``\mathsf{ExprEval}", \stringify {\pexp}]
       \end{array}
       } \and
       \infer{
          \st, \pifelse{\pexp}{\scmd_1}{\scmd_2} \baction_{\fictx} \outcome: \st'
       }{
       \esem{ \pexp }{ \st } = \true \quad  \st, \scmd_1 \baction_{\fictx} \outcome: \st'
       } 
       \and
       \infer{
          \st, \pifelse{\pexp}{\scmd_1}{\scmd_2} \baction_{\fictx} \outcome: \st'
       }{
          \esem{ \pexp }{ \st } = \false \quad  \st, \scmd_2 \baction_{\fictx} \outcome: \st'
       } \and
       \infer{
          \sthreadp{ \sto }{ \hp }, \pifelse{\pexp}{\scmd_1}{\scmd_2} \baction_{\fictx} {\oxerr} : \sthreadp{ \sto_{\oxerr} }{ \hp }
       }{
       \begin{array}{c}
       \esem{\pexp}{\sto} = \undefd \quad \verr = [``\mathsf{ExprEval}", \stringify {\pexp}]
       \end{array}
       }
       \and
       \infer{
          \sthreadp{ \sto }{ \hp },  \pifelse{\pexp}{\scmd_1}{\scmd_2}
          \baction_{\fictx} {\oxerr} : \sthreadp{ \sto_{\oxerr} }{ \hp }
       }{
              \begin{array}{c}
       \esem{\pexp}{\sto} = v \notin \Bool  \\ \verr = [``\mathsf{Type}", \stringify {\pexp}, v, ``\mathsf{Bool}"]
       \end{array}
       } \and
       \infer{
    \st, \scmd_1 ; \scmd_2 \baction_{\fictx} \outcome : \st'
       }{
       \begin{array}{c}
       \st, \scmd_1 \baction_{\fictx} \st'' \quad  \st'', \scmd_2 \baction_{\fictx} \outcome:  \st' \end{array}
       } \and
       \infer{
    \st, \scmd_1 ; \scmd_2 \baction_{\fictx} \outcome: \st'
       }{
       \begin{array}{c}
       \st, \scmd_1 \baction_{\fictx} \outcome:   \st' \quad  \outcome \neq \oxok \end{array}
       } \and
\and
%
 \infer{
     \sthreadp{ \sto }{ \hp },
     \pfuncall{\pvar{y}}{\procname}{\vec{\pexp}} \baction_{\fictx} \oxok : 
     \sthreadp{\sto [\pvar{y} \storearrow v' ] }{ \hp' }
  }{
    \begin{array}{c}
    \pfunction{\procname}{\vec{\pvar{x}}}{\scmd; \preturn{\pexp'}} \in \scontext
    \\
    \esem{\vec{\pexp}}{\sto} = \vec{v} 
     \quad \pv{\scmd} \setminus
      \{\vec{\pvar{x}}\} = \{\vec{\pvar{z}} \}
    \\
 \sto_p  =  \emptyset [ \vec{\pvar{x}} \storearrow \vec{v}] [ \vec{\pvar{z}} \storearrow \nil]
     \\ (\sto_p, \hp), \scmd
      \baction_{\fictx} \oxok:\sthreadp{ \sto_q }{ \hp' }  \quad  \esem{\pexp'}{\sto_q} =v' 
    \end{array}
  }
  \and
  \infer{
     \sthreadp{ \sto }{ \hp },
     \pfuncall{\pvar{y}}{\procname}{\vec{\pexp}} \baction_{\fictx}
     \oxerr: \sthreadp{\sto_{\oxerr} }{ \hp' }
  }{
    \begin{array}{c}
    \pfunction{\procname}{\vec{\pvar{x}}}{\scmd; \preturn{\pexp'}} \in \scontext
    \\
    \esem{\vec{\pexp}}{\sto} = \vec{v} 
     \quad \pv{\scmd} \setminus
      \{\vec{\pvar{x}}\} = \{\vec{\pvar{z}} \}
    \\
 \sto_p  =  \emptyset [ \vec{\pvar{x}} \storearrow \vec{v}] [ \vec{\pvar{z}} \storearrow \nil]
     \\ (\sto_p, \hp), \scmd
      \baction_{\fictx} \oxok:\sthreadp{ \sto_q }{ \hp' }  \quad 
      \esem{\pexp'}{\sto_q} = \undefd \\ \verr = [``\mathsf{ExprEval}", \stringify {\pexp'}]
    \end{array}
  } \and
  \infer{
     \sthreadp{ \sto }{ \hp },
     \pfuncall{\pvar{y}}{\procname}{\vec{\pexp}} \baction_{\fictx}
     \oxerr: \sthreadp{\sto_{\oxerr} }{ \hp }
  }{
    \begin{array}{c}
    \pfunction{\procname}{\vec{\pvar{x}}}{\scmd; \preturn{\pexp'}} \in \scontext
    \\
    |\vec{\pvar x}| \neq |\vec{\pexp}| \\
    \verr = [``\mathsf{ParamCount}", f]
    \end{array}
  } \and
  \infer{
     \sthreadp{ \sto }{ \hp },
     \pfuncall{\pvar{y}}{\procname}{\vec{\pexp}} \baction_{\fictx}
     \outcome: \sthreadp{\sto [\pvar{err} \storearrow \sto_q(\pvar{err}) ] }{ \hp }
  }{
    \begin{array}{c}
    \pfunction{\procname}{\vec{\pvar{x}}}{\scmd; \preturn{\pexp'}} \in \scontext
    \\
    \esem{\vec{\pexp}}{\sto} = \vec{v} 
     \quad \pv{\scmd} \setminus
      \{\vec{\pvar{x}}\} = \{\vec{\pvar{z}} \}
    \\
 \sto_p  =  \emptyset [ \vec{\pvar{x}} \storearrow \vec{v}] [ \vec{\pvar{z}} \storearrow \nil]
     \\ (\sto_p, \hp), \scmd
      \baction_{\fictx} \outcome: \sthreadp{ \sto_q }{ \hp' }  \quad 
      \outcome \neq \oxok
    \end{array}
  }\and
  \infer{
     \sthreadp{ \sto }{ \hp },
     \pfuncall{\pvar{y}}{\procname}{\pexp_1, \ldots \pexp_n} \baction_{\fictx}
     \oxerr: \sthreadp{\sto_{\oxerr} }{ \hp }
  }{
    \begin{array}{c}
    \pfunction{\procname}{\vec{\pvar{x}}}{\scmd; \preturn{\pexp'}} \in \scontext 
    \\
    k \in \{ 1, \ldots n \} \quad (\esem{\pexp_i}{\sto} = v_i)|_{i=1}^{k-1}  \quad \esem{\pexp_k}{\sto} = \undefd  \\
    |\vec{\pvar{x}}| = n \quad \verr = [``\mathsf{ExprEval}", \stringify {\pexp_k}]
    \end{array}
  } \and
        \infer{
     \sthreadp{ \sto }{ \hp }, \pfuncall{\pvar{x}}{\procname}{\vec{\pexp}} \baction_{\fictx} {\oxerr} :  (\sto_{\oxerr}, \hp)
  }{
    \fid \notin \dom(\scontext) \quad
    \verr = [``\mathsf{NoFunc}", \fid]
  }
\and
       \infer{
          \sthreadp{ \sto }{ \hp }, \pderef{\pvar{x}}{\pexp} \baction_{\fictx} \oxok:\sthreadp{\sto [\pvar{x} \storearrow v] }{ \hp }
       }{
       \esem{\pexp}{\sto} = n \quad \hp( n) = v
       }  \and
       \infer{
          \sthreadp{ \sto }{ \hp }, \pderef{\pvar x}{\pexp} \baction_{\fictx} {\oxerr} : \sthreadp{ \sto_{\oxerr} }{ \hp }
       }{
       \begin{array}{c}
       \esem{\pexp}{\sto} = \undefd \\ \verr = [``\mathsf{ExprEval}", \stringify {\pexp}]
       \end{array}
       }
\and
       \infer{
          \sthreadp{ \sto }{ \hp }, \pderef{\pvar x}{\pexp}
          \baction_{\fictx} {\oxerr} : \sthreadp{ \sto_{\oxerr} }{ \hp }
       }{
              \begin{array}{c}
       \esem{\pexp}{\sto} = v \notin \nats  \\ \verr = [``\mathsf{Type}", \stringify {\pexp}, v, ``\mathsf{Nat}"]
       \end{array}
       } \and
        \infer{
          \sthreadp{ \sto }{ \hp }, \pderef{\pvar x}{\pexp}
          \baction_{\fictx} {\oxm} : \sthreadp{ \sto_{\oxerr} }{ \hp }
       }{
       \begin{array}{c}\esem{\pexp}{\sto} = n  \notin \dom (\hp) \\ \verr = [``\mathsf{MissingCell}", \stringify{\pexp}, n]\end{array}
       }
       \and
       \infer{
          \sthreadp{ \sto }{ \hp }, \pderef{\pvar x}{\pexp}
          \baction_{\fictx} {\oxerr} : \sthreadp{ \sto_{\oxerr} }{ \hp }
       }{
       \begin{array}{c}\esem{\pexp}{\sto} = n  \quad \hp(n) = \cfreed \\ \verr = [``\mathsf{UseAfterFree}", \stringify{\pexp}, n]\end{array}
       } \and
       \infer{
          \sthreadp{ \sto }{ \hp }, \pmutate{\pexp_1}{\pexp_2}
          \baction_{\fictx} \oxok: \sthreadp{ \sto }{ \hp [n \mapsto v] }
       }{
       \esem{\pexp_1}{\sto} = n \quad \hp(n) \in \Val \quad \esem{\pexp_2}{\sto} = v}
\and
\infer{
          \sthreadp{ \sto }{ \hp }, \pmutate{\pexp_1}{\pexp_2} \baction_{\fictx} {\oxerr} : \sthreadp{ \sto_{\oxerr} }{ \hp }
       }{
       \begin{array}{c}
       \esem{\pexp_1}{\sto} = \undefd \\ \verr = [``\mathsf{ExprEval}", \stringify {\pexp_1}]
       \end{array}
       }
       \and
       \infer{
          \sthreadp{ \sto }{ \hp }, \pmutate{\pexp_1}{\pexp_2}
          \baction_{\fictx} {\oxerr} : \sthreadp{ \sto_{\oxerr} }{ \hp }
       }{
              \begin{array}{c}
       \esem{\pexp_1}{\sto} = v \notin \nats  \\ \verr = [``\mathsf{Type}", \stringify {\pexp_1}, v, ``\mathsf{Nat}"]
       \end{array}
       } 
       \and
        \infer{
          \sthreadp{ \sto }{ \hp }, \pmutate{\pexp_1}{\pexp_2}
          \baction_{\fictx} {\oxm} : \sthreadp{ \sto_{\oxerr} }{ \hp }
       }{
       \begin{array}{c}\esem{\pexp_1}{\sto} = n  \notin \dom (\hp) \\ \verr = [``\mathsf{MissingCell}", \stringify{\pexp_1}, n]\end{array}
       } 
       \and
       \infer{
          \sthreadp{ \sto }{ \hp }, \pmutate{\pexp_1}{\pexp_2}
          \baction_{\fictx} {\oxerr} : \sthreadp{ \sto_{\oxerr} }{ \hp }
       }{
       \begin{array}{c}\esem{\pexp_1}{\sto} = n  \quad \hp(n) = \cfreed \\ \verr = [``\mathsf{UseAfterFree}", \stringify{\pexp_1}, n]\end{array}
       } \and
       \infer{
          \sthreadp{ \sto }{ \hp }, \pmutate{\pexp_1}{\pexp_2}
          \baction_{\fictx} {\oxerr} : \sthreadp{ \sto_{\oxerr} }{ \hp }  }{
       \begin{array}{c}
       \esem{\pexp_1}{\sto} = n \quad \hp(n) \in \Val \quad \esem{\pexp_2}{\sto} = \undefd \\ \verr = [``\mathsf{ExprEval}", \stringify {\pexp_2}]
       \end{array}
       } \and
       \infer{
          \sthreadp{ \sto }{ \hp }, \palloc{\pvar{x}}{n} \baction_{\fictx} \oxok:\sthreadp{\sto [\pvar{x} \storearrow n'] }{ \hp' }
       }{
       \begin{array}{c}
       		\qquad (n' + i \notin \dom (\hp))|_{0 \le i < n} \\
       		\hp' = \hp[n' \mapsto \nil]\cdots[n' + n - 1 \mapsto \nil]
       \end{array}
       }
       \and
%
%
       \infer{
          \sthreadp{ \sto }{ \hp }, \pdealloc{\pexp} \baction_{\fictx} \oxok:\sthreadp{ \sto }{ \hp[n \mapsto \cfreed]}
       }{
       \esem{\pexp}{\sto} = n \quad  \hp(n) \in \Val 
       } \and
       \infer{
          \sthreadp{ \sto }{ \hp }, \pdealloc{\pexp} \baction_{\fictx} {\oxerr} : \sthreadp{ \sto_{\oxerr} }{ \hp }
       }{
       \begin{array}{c}
       \esem{\pexp}{\sto} = \undefd \\ \verr = [``\mathsf{ExprEval}", \stringify {\pexp}]
       \end{array}
       }
       \and
       \infer{
          \sthreadp{ \sto }{ \hp }, \pdealloc{\pexp}
          \baction_{\fictx} {\oxerr} : \sthreadp{ \sto_{\oxerr} }{ \hp }
       }{
              \begin{array}{c}
       \esem{\pexp}{\sto} = v \notin \nats  \\ \verr = [``\mathsf{Type}", \stringify {\pexp}, v, ``\mathsf{Nat}"]
       \end{array}
       } \and
        \infer{
          \sthreadp{ \sto }{ \hp }, \pdealloc{\pexp}
          \baction_{\fictx} {\oxm} : \sthreadp{ \sto_{\oxerr} }{ \hp }
       }{
       \begin{array}{c}\esem{\pexp}{\sto} = n  \notin \dom (\hp) \\ \verr = [``\mathsf{MissingCell}", \stringify{\pexp}, n]\end{array}
       } 
       \and
       \infer{
          \sthreadp{ \sto }{ \hp }, \pdealloc{\pexp}
          \baction_{\fictx} {\oxerr} : \sthreadp{ \sto_{\oxerr} }{ \hp }
       }{
       \begin{array}{c}\esem{\pexp}{\sto} = n  \quad \hp(n) = \cfreed \\ \verr = [``\mathsf{UseAfterFree}", \stringify{\pexp}, n]\end{array}
       } 	
    \end{mathparpagebreakable}
}

\vspace{1em}
\noindent
where $\sto_{\oxerr} \defeq \serr$.

\subsection{Predicate Folding and Unfolding}

The predicate folding and unfolding commands are ghost commands and therefore have no effect on the program state:
\begin{mathpar}
\footnotesize
\infer{\st, \pfold{\pred}{\vec{\pexp}} \baction_{\fictx} \oxok: \st}{ \ }
\and
\infer{\st, \punfold{\pred}{\vec{\pexp}} \baction_{\fictx} \oxok: \st}{ \ }
\end{mathpar}

\subsection{Symbolic Testing}%
\label{app:concretesymtesting}

We provide the rules for the classic commands supported by symbolic testing tools and bounded model checkers.
\begin{itemize}
\item	$\mathsf{assert}$ checks that the provided boolean expression is true
\item $\mathsf{assume}$ cuts execution branches which do not satisfy the provided boolean expression
\item $\mathsf{sym}$ is like $\mathsf{nondet}$, but it can yield any value instead of just a $\nats$
\end{itemize}
 
\begin{mathpar}
\footnotesize
\infer{ \sthreadp{\sto}{\hp}, \mathsf{assume}(\pexp) \baction_{\fictx} \osucc : (\sto, \hp)}{\esem{\pexp}{\sto} = \true}
\and
\infer{\sthreadp{\sto}{\hp}, \mathsf{assume}(\pexp) \baction_{\fictx} \oerr : (\sto_{\oerr}, \hp) }{ \esem{\pexp}{\sto} = \undefd \quad v_{\oerr} \defeq [``\mathsf{ExprEval}", \stringify {\pexp}]}
\and
\infer{\sthreadp{\sto}{\hp}, \mathsf{assert}(\pexp) \baction_{\fictx} \osucc : (\sto, \hp)}{\esem{\pexp}{\sto} = \true}
\and
\infer{\sthreadp{\sto}{\hp}, \mathsf{assert}(\pexp) \baction_{\fictx} \oerr : (\sto_{\oerr}, \hp)}{\esem{\pexp}{\sto} = \false \quad \sverr \defeq [\mathstr{\mathsf{Assert}}, \stringify{\pexp}]}
\and
\infer{\sthreadp{\sto}{\hp}, \mathsf{assert}(\pexp) \baction_{\fictx} \oerr : (\sto_{\oerr}, \hp) }{ \esem{\pexp}{\sto} = \undefd \quad v_{\oerr} \defeq [``\mathsf{ExprEval}", \stringify {\pexp}]}
\and
\infer{\sthreadp{\sto}{\hp}, \pvar{x} := \mathsf{sym} \baction_{\fictx} \osucc: (\sto[\pvar{x} \mapsto v],\hp)}{v \in \vals}
\end{mathpar}
\vspace{1em}
\noindent
where $\sto_{\oxerr} \defeq \serr$.

\subsection{OX and UX Frame Properties}


The language is compositional in that it obeys the standard OX and UX frame properties when the outcome is successful or a language error. As we now discuss.

We denote the set of modified variables for commands $\cmd$ be denoted by $\updt(\cmd)$. We use state composition to frame on and frame off state $\st_f$ as long as its variables are not modified by the current command $C$: formally, we write $\myunmod{(\st_f, \cmd)}$ if and only if $\updt(\cmd) \cap \pv{\st_f} = \emptyset$:
\begin{lemma}[OX and UX Frame properties]\label{lem:ux-ox-frame}

\hspace{2em} \\
  
\begin{tabular}{ll}
OX: & $\st \cdot \st_f, \cmd \baction_{\fictx} \result : \st'
            \land \pv{\cmd} \subseteq \pv{\st} \land
            \myunmod{(\st_f, C)} \implies $\\
& $\qquad \exists \st'', \result'.
~\st, \cmd \baction_{\fictx} \result' : \st''~\land 
 (\result' \neq \omiss
 \Rightarrow 
 (\st' = \st'' \cdot \st_f \land \result' = \result))$
\\[1mm]
UX:  & $\st, \cmd \baction_{\fictx} \result : \st' \land \result
             \neq \omiss
             \land \myunmod{(\st_f  , C)} 
  \land \st'\cdot\st_f \text{ defined} \implies 
\st \cdot \st_f, \cmd \baction_{\fictx} \result : \st'\cdot\st_f $
\end{tabular}
\end{lemma}
%
%
%
\noindent The frame properties do not hold when the outcome denotes a missing-resource error since the outcome can change if the missing resource is added by the frame. The OX frame property is standard but might appear surprising, but the  simpler property analogous to the  UX frame property does not hold as the frame might interfere with allocations created in the given~execution~trace~\cite{Yang02}.

%% file: sections/app-symbolic.tex

\section{Symbolic Semantics}\label{app:symbolic}
This appendix provides an exhaustive list of the symbolic execution semantics rules. The corresponding judgement is
\[
	\csesemtransabstract{\sst}{\scmd}{\sst'}{\fsctx}{\outcome}
\]
which intuitively says that symbolically executing command $\scmd$ in state $\sst$ ends in state $\sst'$ with outcome $\outcome := \osucc \mid \oerr \mid \omiss$, in the specification context $\fsctx$.

In cases where a rule is only applicable to one mode of reasoning, that is, OX and UX, we write:
\[
	\csesemtransabstractm{\sst}{\scmd}{\sst'}{\fsctx}{\macUX}{\outcome} \qquad \csesemtransabstractm{\sst}{\scmd}{\sst'}{\fsctx}{\macOX}{\outcome}
\]
For rules applicable to both OX and UX reasoning, one needs to keep in mind is that for the recursively defined rules, all recursive executions must happen in the same mode as the initial executions.

These three judgements corresponds to the \textit{symbolic semantics with function calls by specifications}. There is also an alternative judgement for the \textit{symbolic semantics with function calls by inlining}, which uses a function context $\scontext$ instead of a specification context.

        
\subsection{Expression Evaluation}

The symbolic semantics of expression evaluation may branch, we provide a few rules here, matching the ones given in the concrete semantics.
{\footnotesize
\begin{mathpar}
\inferrule[Eval-PVar]{}
{\cseeval{\pvar{x}}{\ssto}{\spc}{\ssto(x)}{\spc}}
\and
\inferrule[Eval-Value]{}
{\cseeval{v}{\ssto}{\spc}{v}{\spc}}
\and
\inferrule[Eval-Plus]{
	\cseeval{\pexp_1}{\ssto}{\spc}{\sval_1}{\spc'}\quad
	\cseeval{\pexp_2}{\ssto}{\spc'}{\sval_2}{\spc''}\\\\
	\spc''' = \spc'' \land \sval_1 \in \nats \land \sval_2 \in \nats\\\\
	\sat(\spc''')
}
{\cseeval{\pexp_1~{+}~\pexp_2}{\ssto}{\spc}{\sval_1 + \sval_2}{\spc'''}}
\and
\inferrule[Eval-Div]{
	\cseeval{\pexp_1}{\ssto}{\spc}{\sval_1}{\spc'}\quad
	\cseeval{\pexp_2}{\ssto}{\spc'}{\sval_2}{\spc''}\\\\
	\spc''' = \spc'' \land \sval_1 \in \nats \land \sval_2 \in \nats \land \sval_2 \neq 0\\\\
	\sat(\spc''')
}
{\cseeval{\pexp_1~{/}~\pexp_2}{\ssto}{\spc}{\sval_1 / \sval_2}{\spc'''}}
\end{mathpar}}

\subsection{Core Rules}\label{ssec:sym_rules}

The rules given here are the ones common to all kinds of symbolic execution. Although we provide them using the judgement with specification context $\fsctx$, these rules are also valid for the semantics with function inlining using a function context $\scontext$.

{
\footnotesize
\begin{mathparpagebreakable}
\inferrule[\textsc{Skip}]
{}
{\csesemtransabstract{\sst}{\pskip}{\sst}{\fsctx}{\osucc}}
\and
\inferrule[\textsc{Assign}]
{\cseeval{\pexp}{\ssto}{\spc}{\sval}{\spc'} \quad \ssto' \defeq \ssto[\pvar{x} \mapsto \sval] }
{\csesemtrans{\ssto, \smem, \spc}{\passign{\pvar{x}}{\pexp}}{\ssto', \smem, \spc'}{\fsctx}{\osucc}}
\and
\inferrule[\textsc{Assign (Error)}]
{\cseeval{\pexp}{\ssto}{\spc}{\undefd}{\spc'} \\\\ \sverr \defeq [\mathstr{\mathsf{ExprEval}}, \stringify{\pexp}] }
{\csesemtrans{\ssto, \smem, \spc}{\passign{\pvar{x}}{\pexp}}{\ssto_{\oerr}, \smem, \spc'}{\fsctx}{\oerr}}
\and
\inferrule[\textsc{Nondet}]
{\svar{r} \text{ fresh} \quad \spc' \defeq \svar{r} \in \nats \land \spc \quad \ssto' = \ssto[\pvar{x} \mapsto \svar{r}]}
{\csesemtrans{\ssto, \smem, \spc}{\passign{\pvar{x}}{\prandom}}{\ssto', \smem, \spc'}{\fsctx}{\osucc}}
\and
\inferrule[\textsc{Error}]
{\cseeval{\pexp}{\ssto}{\spc}{\sval}{\spc'} \quad \sverr \defeq [\mathstr{\mathsf{Error}}, \sval]}
{\csesemtrans{\ssto, \smem, \spc}{\perror(\pexp)}{\ssto_{\oerr}, \smem, \spc'}{\fsctx}{\oerr}}
\and
\inferrule[\textsc{Error (Error)}]
{\cseeval{\pexp}{\ssto}{\spc}{\undefd}{\spc'} \quad \sverr \defeq [\mathstr{\mathsf{ExprEval}}, \stringify{\pexp}]}
{\csesemtrans{\ssto, \smem, \spc}{\perror(\pexp)}{\ssto_{\oerr}, \smem, \spc'}{\fsctx}{\oerr}}
\and
\inferrule[\textsc{If-Then}]
{\cseeval{\pexp}{\ssto}{\spc}{\sval}{\spc'} \quad \spc'' \defeq \spc' \land \sval \quad \sat(\spc'') \\\\ \csesemtrans{\ssto, \smem, \spc''}{C_1}{\ssto', \smem', \spc'''}{\fsctx}{\result}}
{\csesemtrans{\ssto, \smem, \spc}{\pifelse{\pexp}{C_1}{C_2}}{ \ssto', \smem', \spc'''}{\fsctx}{\result}}
\and
\inferrule[\textsc{If-Else}]
{\cseeval{\pexp}{\ssto}{\spc}{\sval}{\spc'} \quad \spc'' \defeq \spc' \land \neg\sval \quad \sat(\spc'') \\\\ \csesemtrans{\ssto, \smem, \spc''}{C_2}{\ssto', \smem', \spc'''}{\fsctx}{\result}}
{\csesemtrans{\ssto, \smem, \spc}{\pifelse{\pexp}{C_1}{C_2}}{ \ssto', \smem', \spc'''}{\fsctx}{\result}}
\and
\inferrule[\textsc{If-Err-Val}]
{\cseeval{\pexp}{\ssto}{\spc}{\undefd}{\spc'} \quad \sverr \defeq [\mathstr{\mathsf{ExprEval}}, \stringify{\pexp}]}
{\csesemtrans{\ssto, \smem, \spc}{\pifelse{\pexp}{C_1}{C_2}} {\ssto_{\oerr}, \smem,  \spc'}{\fsctx}{\oerr}}
\and
\inferrule[\textsc{If-Err-Type}]
{\cseeval{\pexp}{\ssto}{\spc}{\sval}{\spc'} \quad \spc'' \defeq \spc' \land \sval \not\in \bools \quad \sat(\spc'') \\\\ \sverr \defeq [\mathstr{\mathsf{Type}}, \stringify{\pexp}, \sval, \mathstr{\mathsf{Bool}}]}
{\csesemtrans{\ssto, \smem, \spc}{\pifelse{\pexp}{C_1}{C_2}} {\ssto_{\oerr}, \smem,  \spc''}{\fsctx}{\oerr}}
\and
\inferrule[\textsc{Seq}]
{\csesemtransabstract{\sst}{C_1}{\sst'}{\fsctx}{\osucc} \\\\ \csesemtransabstract{\sst'}{C_2}{\sst''}{\fsctx}{\result}}
{\csesemtransabstract{\sst}{C_1;C_2}{\sst''}{\fsctx}{\result}}
\and
\inferrule[\textsc{Seq-Err}]
{\csesemtransabstract{\sst}{C_1}{\sst'}{\fsctx}{\result} \quad \result \neq \osucc}
{\csesemtransabstract{\sst}{C_1;C_2}{\sst'}{\fsctx}{\result}}
\and
\inferrule[\textsc{Lookup}]
 {\cseeval{\pexp}{\ssto}{\spc}{\sval}{\spc'} \quad \smem(\sloc) = \sval_m \\\\ \spc'' \defeq (\sloc = \sval) \land \spc' \quad \sat(\spc'')}
 {\csesemtrans{\ssto, \smem,\spc}{\pderef{\pvar{x}}{\pexp}}{ \ssto[\pvar{x} \mapsto \sval_m], \smem, \spc''}{\fsctx}{\osucc}}
\and
\inferrule[\textsc{Lookup-Err-Val}]
  {\cseeval{\pexp}{\ssto}{\spc}{\undefd}{\spc'}\quad \sverr \defeq [\mathstr{\mathsf{ExprEval}}, \stringify{\pexp}] }
  {\csesemtrans{\ssto, \smem,  \spc}{\pderef{\pvar{x}}{\pexp}}{\ssto_{\oerr}, \smem, \spc'}{\fsctx}{\oerr}}
\and
\inferrule[\textsc{Lookup-Err-Type}]
 {\cseeval{\pexp}{\ssto}{\spc}{\sval}{\spc'} \quad \spc'' \defeq \sval \not\in \nats \land \spc' \quad \sat(\spc'') \quad \sverr \defeq [\mathstr{\mathsf{Type}}, \stringify{\pexp}, \sval, \mathstr{\mathsf{Nat}}]}
 {\csesemtrans{\ssto, \smem,\spc}{\pderef{\pvar{x}}{\pexp}}{\ssto_{\oerr}, \smem, \spc''}{\fsctx}{\oerr}}
\and
\inferrule[\textsc{Lookup-Err-Use-After-Free}]
 {\cseeval{\pexp}{\ssto}{\spc}{\sval}{\spc'} \quad \smem(\sloc) = \cfreed \quad \spc'' \defeq \sval \in \nats \land (\sloc = \sval) \land \spc' \\\\ \sverr \defeq [\mathstr{\mathsf{UseAfterFree}}, \stringify{\pexp}, \sval] \quad \sat(\spc'')}
 {\csesemtrans{\ssto, \smem, \spc}{\pderef{\pvar{x}}{\pexp}}{\ssto_\oerr , \smem, \spc''}{\fsctx}{\oerr}}
\and
\inferrule[\textsc{Lookup-Err-Missing}]
 {\cseeval{\pexp}{\ssto}{\spc}{\sval}{\spc'} \quad \spc'' \defeq \sval \in \nats \land \sval \not\in \domain(\smem) \land \spc' \\\\ \sverr \defeq [\mathstr{\mathsf{MissingCell}}, \stringify{\pexp}, \sval]\quad \sat(\spc'') }
 {\csesemtrans{\ssto, \smem, \spc}{\pderef{\pvar{x}}{\pexp}}{\ssto_\oerr, \smem, \spc''}{\fsctx}{\omiss}}
\and
\inferrule[\textsc{Mutate}]
 {\cseeval{\pexp_1}{\ssto}{\spc}{\sval_1}{\spc'} \quad \smem(\sloc) = \sval_m \quad \spc'' \defeq (\sloc = \sval_1) \land \spc' \\\\ \sat(\spc'') \quad \cseeval{\pexp_2}{\ssto}{\spc''}{\sval_2}{\spc'''} \quad \smem' = \smem[\sloc \mapsto \sval_2]  }
 {\csesemtrans{\ssto, \smem,  \spc}{\pmutate{\pexp_1}{\pexp_2}}{\ssto, \smem',\spc'''}{\fsctx}{\osucc}}
\and
\inferrule[\textsc{Mutate-Err-Val-1}]
  {\cseeval{\pexp_1}{\ssto}{\spc}{\undefd}{\spc'} \quad 
  \sverr \defeq [\mathstr{\mathsf{ExprEval}}, \stringify{\pexp_1}]}
  {\csesemtrans{\ssto, \smem,  \spc}{\pmutate{\pexp_1}{\pexp_2}}{\ssto_{\oerr}, \smem, \spc'}{\fsctx}{\oerr}}
\and
\inferrule[\textsc{Mutate-Err-Type}]
  {\cseeval{\pexp_1}{\ssto}{\spc}{\sval_1}{\spc'} \quad \spc'' \defeq \sval_1 \not\in \nats \land \spc' \quad \sat(\spc'') \\\\ \sverr \defeq [\mathstr{\mathsf{Type}}, \stringify{\pexp_1}, \sval_1, \mathstr{\mathsf{Nat}}]}
  {\csesemtrans{\ssto, \smem, \spc}{\pmutate{\pexp_1}{\pexp_2}}{\ssto_\oerr, \smem,  \spc''}{\fsctx}{\oerr}}
\and
\inferrule[\textsc{Mutate-Err-Missing}]
 {\cseeval{\pexp_1}{\ssto}{\spc}{\sval_1}{\spc'} \quad \spc'' \defeq \sval_1 \in \nats \land \sval_1 \not\in \domain(\smem) \land \spc' \\\\ \sat(\spc'') \quad \sverr \defeq [\mathstr{\mathsf{MissingCell}}, \stringify{\pexp_1}, \sval_1]}
 {\csesemtrans{\ssto, \smem,\spc}{\pmutate{\pexp_1}{\pexp_2}}{\ssto_\oerr, \smem,  \spc''}{\fsctx}{\omiss}}
\and
\inferrule[\textsc{Mutate-Err-Use-After-Free}]
 {\cseeval{\pexp_1}{\ssto}{\spc}{\sval}{\spc'} \quad \smem(\sloc) = \cfreed \\\\ \spc'' \defeq (\sloc = \sval) \land \spc' \quad \sat(\spc'') \\\\ \sverr \defeq [\mathstr{\mathsf{UseAfterFree}}, \stringify{\pexp_1}, \sval]}
 {\csesemtrans{\ssto, \smem, \spc}{\pmutate{\pexp_1}{\pexp_2}}{\ssto_\oerr, \smem,  \spc''}{\fsctx}{\oerr}}
\and
\inferrule[\textsc{Mutate-Err-Val-2}]
  {\cseeval{\pexp_1}{\ssto}{\spc}{\sval_1}{\spc'} \quad \cseeval{\pexp_2}{\ssto}{\spc'}{\undefd}{\spc''} \quad \spc''' = \sval_1 \in \nats \land \spc'' \\\\ \sat(\spc''') \quad \sverr \defeq [\mathstr{\mathsf{ExprEval}}, \stringify{\pexp_2}]}
  {\csesemtrans{\ssto, \smem, \spc}{\pmutate{\pexp_1}{\pexp_2}}{\ssto_\oerr, \smem,  \spc'''}{\fsctx}{\oerr}}
  \and
\inferrule[\textsc{New-Zero}]
{\sloc \text{ fresh} \quad \spc' \defeq \sloc \in \nats \land \spc}
{\csesemtrans{\ssto, \smem, \spc}{\palloc{\pvar{x}}{0}}{\ssto[\pvar{x} \mapsto \sloc], \smem, \spc'}{\fsctx}{\osucc}}
\and
\inferrule[\textsc{New-Nonzero}]
 {n \neq 0 \quad \sloc \text{ fresh} \quad \spc' \defeq \sloc,\dots,\sloc + n - 1 \in \nats \land \sloc,\dots,\sloc + n - 1 \not\in \domain(\smem) \land \spc}
 {\csesemtrans{\ssto, \smem, \spc}{\palloc{\pvar{x}}{n}}{\ssto[\pvar{x} \mapsto \sloc], \smem[\sloc \mapsto \texttt{null}]\cdots[\sloc + n - 1 \mapsto \texttt{null}], \spc'}{\fsctx}{\osucc}}
\and
\inferrule[\textsc{Free}]
 {\cseeval{\pexp}{\ssto}{\spc}{\sval}{\spc'} \quad \smem(\sloc) = \sval_m \\\\ \spc'' \defeq (\sloc = \sval) \land \spc' \\\\ \sat(\spc'') \quad
 \smem' = \smem[\sloc \mapsto \cfreed]}
 {\csesemtrans{\ssto, \smem, \spc}{\pdealloc{\pexp}}{\ssto, \smem', \spc''}{\fsctx}{\osucc}}
\and
\inferrule[\textsc{Free-Err-Eval}]
  {\cseeval{\pexp}{\ssto}{\spc}{\undefd}{\spc'} \quad \sverr \defeq [\mathstr{\mathsf{ExprEval}}, \stringify{\pexp}]}
  {\csesemtrans{\ssto, \smem, \spc}{\pdealloc{\pexp}}{\ssto_\oerr, \smem,  \spc'}{\fsctx}{\oerr}}
\and
\inferrule[\textsc{Free-Err-Type}]
  {\cseeval{\pexp}{\ssto}{\spc}{\sval}{\spc'} \quad \spc'' \defeq \sval \not\in \nats \land \spc' \quad \sat(\spc'') \\\\ \sverr \defeq [\mathstr{\mathsf{Type}}, \stringify{\pexp}, \sval, \mathstr{\mathsf{Nat}}]}
  {\csesemtrans{\ssto, \smem, \spc}{\pdealloc{\pexp}}{\ssto_\oerr, \smem,  \spc''}{\fsctx}{\oerr}}
\and
\inferrule[\textsc{Free-Err-Missing}]
 {\cseeval{\pexp}{\ssto}{\spc}{\sval}{\spc'} \quad \spc'' \defeq \sval \in \nats \land \sval \not\in \domain(\smem) \land \spc'  \\\\ \sverr \defeq [\mathstr{\mathsf{MissingNegCell}}, \stringify{\pexp}, \sval]\quad \sat(\spc'')}
 {\csesemtrans{\ssto, \smem,  \spc}{\pdealloc{\pexp}}{\ssto_\oerr, \smem,  \spc''}{\fsctx}{\omiss}}
\and
\inferrule[\textsc{Free-Err-Use-After-Free}]
 {\cseeval{\pexp}{\ssto}{\spc}{\sval}{\spc'} \quad \smem(\sloc) = \cfreed \\\\
  \spc'' \defeq \sval \in \nats \land (\sloc = \sval) \land \spc' \\\\ 
 \sat(\spc'') \quad \sverr \defeq [\mathstr{\mathsf{UseAfterFree}}, \stringify{\pexp}, \sval]}
 {\csesemtrans{\ssto, \smem, \spc}{\pdealloc{\pexp}}{\ssto_\oerr, \smem,  \spc''}{\fsctx}{\oerr}}
\and
\inferrule[\textsc{Fcall-Err-ParamCount}]
  {\uquadruple{\vec{\pvar{x}} = \vec{x} \lstar P}{\fid(\vec{\pvar{x}})}{\Qok}{\Qerr} \in \fsctx \quad |\vec{\pexp}| \neq |\vec{\pvar{x}}| \\\\
   \sverr \defeq [\mathstr{\mathsf{ParamCount}}, \fid]}
  {\csesemtrans{\ssto, \smem, \spc}{\passign{\pvar y}{\fid(\vec{\pexp})}}{\ssto_\oerr, \smem, \spc}{\fsctx}{\oerr}}
\and
\inferrule[\textsc{Fcall-Err-Val}]
  {\uquadruple{\vec{\pvar{x}} = \vec{x} \lstar P}{\fid(\vec{\pvar{x}})}{\Qok}{\Qerr} \in \fsctx \quad |\vec{\pvar{x}}| = n \\\\
   1 \leq m \leq n \quad \spc_0 = \spc \quad (\cseeval{\pexp_i}{\ssto}{\spc_{i-1}}{\sym v_i}{\spc_i})|_{i=1}^{m-1} \quad \cseeval{\pexp_m}{\ssto}{\spc_{m-1}}{\undefd}{\spc'} \quad \sverr \defeq [\mathstr{\mathsf{ExprEval}}, \stringify{\pexp_m}]}
  {\csesemtrans{\ssto, \smem, \spc}{\passign{\pvar y}{\fid(\pexp_1, \ldots, \pexp_n)}}{\ssto_\oerr, \smem,  \spc'}{\fsctx}{\oerr}}
\end{mathparpagebreakable}
}

\vspace{1em}
\noindent
where $\ssto_{\oerr} \defeq \ssto[\pvar{err} \rightarrow \sverr]$.

\subsection{Function Call: Function Inlinling}

The following rules only apply in symbolic execution with function call by inlining. The corresponding rules for function call by specifications are given right after.

\begin{mathpar}
\footnotesize
  \inferrule[\textsc{Fcall-Inline}]
  {\pfunction{\procname}{\vec{\pvar{x}}}{\scmd; \preturn{\pexp'}} \in \scontext\quad
  \cseeval{\vec{\pexp}}{\ssto}{\spc}{\vec{\sval}}{\spc'}\\\\
  \pv{\scmd} \setminus \{\vec{\pvar{x}}\} = \{\vec{\pvar{z}}\} \quad
  \ssto_p = \emptyset[\vec{\pvar{x}} \storearrow{} \vec{\sval}][\vec{\pvar{z}} \storearrow{} \nil]\\\\
  \csesemtrans{\ssto_p, \smem, \spc'}{\scmd}{\ssto_q, \smem', \spc''}{\scontext}{\osucc}\\
  \cseeval{\pexp'}{\ssto_q}{\spc''}{\svar{r}}{\spc'''}}
  {\csesemtrans{\ssto, \smem, \spc}{\passign{\pvar{y}}{\fid(\vec{\pexp})}}{\ssto[\pvar{y} \mapsto \svar{r}], \smem', \spc'''}{\scontext}{\osucc}}
  \hfill
  \inferrule[\textsc{Fcall-Inline-Err}]
  {\pfunction{\procname}{\vec{\pvar{x}}}{\scmd; \preturn{\pexp'}} \in \scontext\quad
  \cseeval{\vec{\pexp}}{\ssto}{\spc}{\vec{\sval}}{\spc'}\\\\
  \pv{\scmd} \setminus \{\vec{\pvar{x}}\} = \{\vec{\pvar{z}}\} \quad
  \ssto_p = \emptyset[\vec{\pvar{x}} \storearrow{} \vec{\sval}][\vec{\pvar{z}} \storearrow{} \nil]\\\\
  \csesemtrans{\ssto_p, \smem, \spc}{\scmd}{\ssto_p', \smem', \spc'''}{\scontext}{\result} \quad \result \neq \osucc}
  {\csesemtrans{\ssto, \smem, \spc}{\passign{\pvar{y}}{\fid(\vec{\pexp})}}{\ssto[\pvar{err} \mapsto \ssto_q(\pvar{err})], \smem', \spc'''}{\scontext}{\result}}
  \and
  \inferrule[\textsc{Fcall-Inline-Err-RetVal}]
  {\pfunction{\procname}{\vec{\pvar{x}}}{\scmd; \preturn{\pexp'}} \in \scontext\\
  \cseeval{\vec{\pexp}}{\ssto}{\spc}{\vec{\sval}}{\spc'} \\
  \pv{\scmd} \setminus \{\vec{\pvar{x}}\} = \{\vec{\pvar{z}}\} \quad
  \ssto_p = \emptyset[\vec{\pvar{x}} \storearrow{} \vec{\sval}][\vec{\pvar{z}} \storearrow{} \nil]\\
  \csesemtrans{\ssto_p, \smem, \spc'}{\scmd}{\ssto_q, \smem', \spc''}{\scontext}{\osucc}\\\\
  \cseeval{\pexp'}{\ssto_q}{\spc''}{\undefd}{\spc'''}\\
  \sverr \defeq [\mathstr{\mathsf{ExprEval}}, \stringify{\pexp'}]
  }
  {\csesemtrans{\ssto, \smem, \spc}{\passign{\pvar{y}}{\fid(\vec{\pexp})}}{\ssto'_{q_\oerr}, \smem', \spc'''}{\scontext}{\oerr}}
\end{mathpar}

\subsection{Function Call: Function Specifications}

The following rules only apply in symbolic execution with function call by specification.
\begin{mathpar}
\footnotesize

\inferrule[\textsc{Fcall}]
{
\cseeval{\vec{\pexp}}{\ssto}{\spc}{\vec{\sval}}{\spc'} \\
\genquadruple{\vec{\pvar{x}} = \vec{x} \lstar P}{\fid(\vec{\pvar{x}})}{\Qok}{\Qerr} \in \fsctx(f)|_m \\
\hat{\theta} \defeq [\vec{x} \mapsto \vec{\sval}] \\\\
\mac(m, P, \hat{\theta}, \sst[\sstupdate{pc}{\spc'}]) \rightsquigarrow (\hat{\theta}', \sst') \\
\Qok = \exists \vec{y}.~\Qok' \\
\ssubst'' \defeq \ssubst{'}[\vec{y}\mapsto \vec{\hat{z}}] \\\\
r, \hat{r},\vec{\hat{z}}  \text{ fresh} \\
\Qok'' = \Qok'[r/\pvar{ret}] \text{~and~} \ssubst''' = \ssubst''[r \mapsto \hat{r}] \\
\produce(\Qok'', \ssubst''', \sst') \rightsquigarrow  \sst''
}
{\csesemtransabstractm{\sst}{\passign{\pvar{y}}{\fid(\vec{\pexp})}}{\sst''[\sstupdate{sto}{\ssto[\pvar{y} \mapsto \svar{r}]}]}{\fsctx}{m}{\osucc}}
\end{mathpar}
\begin{mathpar}
\footnotesize

\inferrule[\textsc{Fcall-Qerr}]
{
\cseeval{\vec{\pexp}}{\ssto}{\spc}{\vec{\sval}}{\spc'} \\
\genquadruple{\vec{\pvar{x}} = \vec{x} \lstar P}{\fid(\vec{\pvar{x}})}{\Qok}{\Qerr} \in \fsctx(f)|_m \\
\hat{\theta} \defeq [\vec{x} \mapsto \vec{\sval}] \\\\
\mac(m, P, \hat{\theta}, \sst[\sstupdate{pc}{\spc'}]) \rightsquigarrow (\hat{\theta}', \sst') \\
\Qerr = \exists \vec{y}.~\Qerr'  \\
\vec{\hat{z}}, r, \hat{r} \text{ fresh} \\
\ssubst'' \defeq \ssubst{'}[\vec{y}\mapsto \vec{\hat{z}}] \\\\
\Qerr'' = \Qerr'[r/\pvar{err}] \\ \ssubst''' = \ssubst''[r \mapsto \hat{r}] \\
\produce(\Qerr'', \ssubst''', \sst') \rightsquigarrow  \sst''
}
{\csesemtransabstractm{\sst}{\passign{\pvar{y}}{\fid(\vec{\pexp})}}{\sst''[\sstupdate{sto}{\ssto[\pvar{err} \mapsto \svar{r}]}]}{\fsctx}{m}{\oerr}}
\and
\inferrule[\textsc{Fcall-Abort}]
{\cseeval{\vec{\pexp}}{\ssto}{\spc}{\vec{\sval}}{\spc'} \\ \genquadruple{\vec{\pvar{x}} = \vec{x} \lstar P}{\fid(\vec{\pvar{x}})}{\Qok}{\Qerr} \in \fsctx(f)|_{m} \\
\hat{\theta} = [\vec{x} \mapsto \vec{\sval}]\\
\mac(m, P, \hat{\theta}, \sst[\sstupdate{pc}{\spc'}]) \rightsquigarrow \oxabort(\sverr)}
{\csesemtransabstractm{\sst}{\passign{\pvar{y}}{\fid(\vec{\pexp})}}{\sst[\sstupdate{pc}{\spc'}, \sstupdate{sto}{\ssto[\pvar{err} \mapsto \sverr]}]}{\fsctx}{m}{\oxabort}}
\end{mathpar}
where $\sst = (\ssto, \smem, \spred, \spc)$.

\subsection{Predicate Folding and Unfolding}\label{ssec:fold_unfold}


The premise $\pred(\predin; \predout)~\{\bigvee_{i \in I} (\exists \vec{x}_i.~A_i)\} \in \preds $ is required for the rules below.

{\footnotesize

\begin{mathpar}
\mprset{flushleft}
\inferrule[Fold]{
 \cseeval{\vec{\pexp}}{\ssto}{\spc}{\vec{\sval}}{\spc} \\
 \ssubst \defeq [\predin \mapsto \vec{\sval} ] \\\\
 \mac(m, A_i, \ssubst, \sst) \rightsquigarrow  (\ssubst', \sst') \\\\
 \hat{\mathcal{P}}'' \defeq \{ \pred(\vec{\sval}; \ssubst'(\predout)) \} \cup \sst'.\fieldsst{preds}}
 {\csesemtransabstractm{\sst}{\pfold{\pred}{\vec{\pexp}}}{\sst'[\sstupdate{preds}{\sps''}]}{\fsctx}{m}{\osucc}}
\and
\inferrule[Fold-Err-ParamCount]
 {
 |\vec{\pexp}| \neq |\predin|}
 {\csesemtransabstractm{\sst}{\pfold{\pred}{\vec{\pexp}}}{\sst_{\oerr}}{\fsctx}{m}{\oxabort}}
\and
\inferrule[Fold-Err-Eval]
 {
  |\vec{\pexp}| = |\predin| \\
  \cseeval{\vec{\pexp}}{\ssto}{\spc}{\undefd}{\spc'}}
 {\csesemtransabstractm{\sst}{\pfold{\pred}{\vec{\pexp}}}{\sst_{\oerr}[\sstupdate{pc}{\spc'}]}{\fsctx}{m}{\oxabort}}
\and
\inferrule[Fold-Err]
 {
  \cseeval{\vec{\pexp}}{\ssto}{\spc}{\vec{\sval}}{\spc} \\
  \ssubst \defeq [\predin \mapsto \vec{\sval} ] \\\\
  \forall i \in I.~\mac(m, A_i, \ssubst, \sst) \rightsquigarrow  \oxabort}
 {\csesemtransabstractm{\sst}{\pfold{\pred}{\vec{\pexp}}}{\sst_{\oerr}}{\fsctx}{m}{\oxabort}}
\and
\mprset{flushleft}
\inferrule[Unfold]{
 \cseeval{\vec{\pexp}}{\ssto}{\spc}{\vec{\sval}}{\spc} \quad 
 \conspred(\pred, \vec{\sval},\sst)\rightsquigarrow ( \spredout,\sst') \\\\
 \vec{\sym{z}} \text{ fresh} \quad 
 \ssubst \defeq [\predin \mapsto \vec{\sval}, \predout\mapsto \spredout, \vec{x}_i \mapsto \vec{\sym{z}}] \\\\
 \produce(A_i,\ssubst,\sst') \rightsquigarrow \sst''}
 {\csesemtransabstractm{\sst}{\punfold{\pred}{\vec{\pexp}}}{\sst''}{\fsctx}{m}{\osucc}}
\and 
\inferrule[Unfold-Err-ParamCount]
 {
   |\vec{\pexp}| \neq |\predin|}
 {\csesemtransabstractm{\sst}{\punfold{\pred}{\vec{\pexp}}}{\sst_{\oerr}}{\fsctx}{m}{\oxabort}}
\and
\inferrule[Unfold-Err-Eval]
 {
  |\vec{\pexp}| = |\predin| \\
  \cseeval{\vec{\pexp}}{\ssto}{\spc}{\undefd}{\spc'}}
 {\csesemtransabstractm{\sst}{\punfold{\pred}{\vec{\pexp}}}{\sst_{\oerr}[\sstupdate{pc}{\spc'}]}{\fsctx}{m}{\oxabort}}
\and
\inferrule[Unfold-Err]
{
 |\vec{\pexp}| = |\predin| \\
 \cseeval{\vec{\pexp}}{\ssto}{\spc}{\vec{\sval}}{\spc} \\\\
 \conspred(\pred, \vec{\sval},\sst)\rightsquigarrow \oxabort}
 {\csesemtransabstractm{\sst}{\punfold{\pred}{\vec{\pexp}}}{\sst_{\oerr}}{\fsctx}{m}{\oxabort}}
\end{mathpar}}
where $\sst = (\ssto, \smem, \spred, \spc)$, $\sst_{\oerr} \defeq \sst[\sstupdate{sto}{\ssto[\pvar{err} \rightarrow \mathstr{\mathsf{fold/unfold}}]}]$, and $\conspred$ is defined as follows:
{\footnotesize
\begin{mathpar}
\mprset{flushleft}
\inferrule
{\spred = \{ \pred(\spredin; \spredout) \} \cup \spred' \\\\
\spc' \defeq \spc \wedge \vec{\sval}= \spredin \qquad \sat(\spc')}
{\conspred(\pred,\vec{\sval},\sst)\rightsquigarrow (\spredout, \sst[\sstupdate{preds}{\sps'}, \sstupdate{pc}{\spc'}])}
\and
\mprset{flushleft}
\inferrule
{\spred= \{\pred(\spredini; \spredouti) \mid i \in I \}\uplus\spred' \quad \pred\notin \spred' \\\\
\sat(\spc \wedge \neg(\bigwedge_{i \in I}\vec{\sval}=  \spredini)) }
{\conspred(\pred,\vec{\sval},\sst)\rightsquigarrow abort}
\end{mathpar}}

\subsection{Symbolic Testing}

The following rules correspond to the symbolic semantics of the three commands added for symbolic testing. See App.~\ref{app:concretesymtesting} for more details.

\begin{mathpar}
\footnotesize
\inferrule[Assume]{
    \cseeval{\pexp}{\ssto}{\spc}{\sval}{\spc'} \quad \sat(\sval)
}{
    \csesemtrans{(\ssto, \smem, \spc)}{\passume{\pexp}}{(\ssto, \smem, \spc' \land \sval)}{\scontext}{\osucc}
}
\hfill
\inferrule[Assume-Err]{
    \cseeval{\pexp}{\ssto}{\spc}{\undefd}{\spc'} \quad
    \sverr \defeq [\mathstr{\mathsf{ExprEval}}, \stringify{\pexp}]
}{
    \csesemtrans{(\ssto, \smem, \spc)}{\passume{\pexp}}{(\ssto_{\oerr}, \smem, \spc')}{\scontext}{\oerr}
}
\and
\inferrule[Assert-Ok] {
  \cseeval{\pexp}{\ssto}{\spc}{\sval}{\spc'} \quad \sat(\sval)
}{
  \csesemtrans{(\ssto, \smem, \spc)}{\passert{\pexp}}{(\ssto, \smem, \spc' \land \sval)}{\scontext}{\osucc}
}
\hfill
\inferrule[Assert-Fail]{
  \cseeval{\pexp}{\ssto}{\spc}{\sval}{\spc'} \quad \sat(\neg \sval) \\\\ \sverr \defeq [\mathstr{\mathsf{Assert}}, \stringify{\pexp}]
}{
  \csesemtrans{(\ssto, \smem, \spc)}{\passert{\pexp}}{(\ssto_{\oerr}, \smem, \spc' \land \neg \sval)}{\scontext}{\osucc}
}
\and
\inferrule[Assert-Err]{
  \cseeval{\pexp}{\ssto}{\spc}{\undefd}{\spc'}\quad
  \sverr \defeq [\mathstr{\mathsf{ExprEval}}, \stringify{\pexp'}]
}{
  \csesemtrans{(\ssto, \smem, \spc)}{\passert{\pexp}}{(\ssto_{\oerr}, \smem, \spc')}{\scontext}{\oerr}
}
\hfill
\inferrule[Sym]{
	\sval \text{  fresh} \quad \spc' = \spc \land \sval \in \vals
}{
\csesemtrans
	{(\ssto, \smem, \spc)}
	{\pvar{x} := \textsf{sym}}
	{\ssto[\pvar{x} \mapsto \sval], \smem, \spc')}
	{\scontext}
	{\osucc}
}
\end{mathpar}

%% file: sections/app-symbolic-soundness.tex

\section{Core Symbolic Execution: Soundness}\label{app:symbolic-soundness}

This appendix provides the relevant definitions, lemmas and proofs for the symbolic soundness theorems.

\subsection{Further Definitions}\label{appssec:furtherdef}

Symbolic states $\sst$ must satisfy certain well-formedness constraints, denoted by $\sinv(\sst)$. To define well-formedness, we need the following definition: $\spc_1 \models \spc_2 \defeq \forall \vint.~ \vint(\spc_1) = \true \Rightarrow \vint(\spc_2) = \true$. Also note that symbolic values are countable; we enumerate them using the notation $\SVal = \{\hat v_i:  i \in \nats\}$. Now, we can define well-formedness: the store co-domain must not contain $\undefd$, given by $\sinvc(\ssto)$; the heap domain must contain disjoint addresses and its co-domain must not contain $\undefd$, given by $\sinvc(\smem)$; and the path condition must be satisfiable and include all the symbolic variables of the state, ensuring that when the path condition can be interpreted so can the entire state:
\[
\begin{array}{l}
\sinvc(\ssto) \defeq \codomain(\ssto) \subseteq \Val \\
\sinvc(\smem) \defeq \domain(\smem) \subseteq \nats \land \codomain(\smem) \subseteq \vals \wedge (\forall \hat{v}_{i}, \hat{v}_{j}\in \domain(\smem). \; i \neq j \Rightarrow \hat{v}_{i} \neq \hat{v}_{j}) \\
\sinv((\ssto, \smem, \spc)) \defeq \sat(\spc)  \land   (\svs{\ssto} \cup \svs{\smem} \subseteq \svs{\spc}) \land \spc \models (\sinvc(\ssto) \land \sinvc(\smem))
\end{array}
\]

Given well-formedness, the definition of symbolic value interpretation is naturally extended to symbolic store, heap and state interpretation. The formal definitions are provided below:
\[
\begin{array}{r@{~\defeq~}l@{~}l}
\vint(\ssto) & \{ \pvar x \mapsto \vint (\ssto(\pvar x)) \mid \pvar x \in \domain(\ssto) \}, & \text{if~} \vint(\sinvc(\ssto)) = \true  \\
\vint(\smem) & \{ \vint(\sval) \mapsto \vint(\smem(\sval)) \mid \sval \in \domain(\smem) \}, & \text{if~} \vint(\sinvc(\smem)) = \true \\
\sint((\ssto, \smem, \spc)) & 
(\sint(\ssto), \sint(\smem)), & \text{if~} \vint(\pc) = \true \land \sinv((\ssto, \smem, \spc))
\end{array}
\]

\subsection{Evaluation Soundness}

\begin{lemma}[Symbolic evaluation: Symbolic variables]\label{thm:eval-symvar}
\[
\cseeval{\pexp}{\ssto}{\spc}{\sym{w}}{\spc'} \implies \svs{\sym{w}} \subseteq \svs{\ssto}
\]
\end{lemma}

\begin{lemma}[Symbolic evaluation: Satisfiable outcome]\label{thm:eval-sat}
\[
\cseeval{\pexp}{\ssto}{\spc}{\sym{w}}{\spc'} \implies \sat(\spc')
\]
\end{lemma}

\begin{lemma}[Symbolic evaluation: Path condition]\label{thm:eval-pc}
\[
\cseeval{\pexp}{\ssto}{\spc}{\sym{w}}{\spc'} \implies \sint(\spc') = \true \implies \sint(\spc) = \true
\]
\end{lemma}

An aside: Note that it is, to be UX sound, important to not overreport errors: For example, for a languages with lazy conjunction semantics, evaluating the expression $\pvar c \land (5~/~\pvar{x} = 4)$ with path condition $\spc$ and a store $\ssto$ with $\ssto(\pvar x) = \svar x$ and $\ssto(\pvar c) = \svar c$, evaluation will branch into $\undefd$ with $\spc'_1 = \spc \land \neg \svar c \land \svar x = 0$ and $\svar c \land (5~/~\svar x = 4)$ with $\spc'_2 = \spc \land \neg(\neg \svar c \land \svar x = 0)$ which simplifies to $\spc \land (\svar c \lor \svar x \neq 0)$. Note that the simpler path condition $\spc \land \svar x = 0$ for $\undefd$ would overapproximate the error, and therefore not a possible path condition for a UX-sound symbolic expression evaluator. Here, however, we do not dwell into details.

\begin{lemma}[Symbolic evaluation: UX soundness]\label{thm:ux-sound-eval}
\begin{gather*}
\cseeval{\pexp}{\ssto}{\spc}{\sym{w}}{\spc'} \land \sint(\sinvc(\ssto)) = \true \land \sint(\spc')  = \true \implies \ceval{\pexp}{\sintp(\ssto)} = \sintp(\sym{w})
\end{gather*}
where $\sym w\in\SVal\cup\{\undefd\}$.
\end{lemma}

\begin{lemma}[Symbolic evaluation: OX soundness]\label{thm:ox-sound-eval}
\begin{gather*}
 \ceval{\pexp}{\sint(\ssto)} = w \gand \sint(\sinvc(\ssto)) = \true \gand \sint(\spc)  = \true \implies \\
 \exists \spc', \sym{w}.~\cseeval{\pexp}{\ssto}{\spc}{\sym{w}}{\spc'} \land \sint(\spc') = \true \gand \sint(\sym{w}) = w
\end{gather*}
where $w \in \Val \cup \{\undefd\}$ and $\sym w\in\SVal\cup\{\undefd\}$.
\end{lemma}

\subsection{Execution Soundness}

We give the proof of Thm.~\ref{thm:ux-ox-sound} -- the OX and UX soundness of the engine -- for a selection of the rules that are representative. All other cases work analogously.
\subsubsection*{OX Soundness}
\subparagraph{Mutate.}
%
%
\begin{mathpar}
    \inferrule[\textsc{Mutate}]
 {\cseeval{\pexp_1}{\ssto}{\spc}{\sval_1}{\spc'} \quad  \quad \smem(\sloc) = \sval_m \quad \spc'' \defeq (\sloc = \sval_1) \land \spc' \\\\ \sat(\spc'') \quad \cseeval{\pexp_2}{\ssto}{\spc''}{\sval_2}{\spc'''} \quad \smem' = \smem[\sloc \mapsto \sval_2]  }
 {\csesemtrans{\ssto, \smem,  \spc}{\pmutate{\pexp_1}{\pexp_2}}{\ssto, \smem',\spc'''}{\fictx}{\osucc}}
\end{mathpar}
We assume a successful execution of the mutate command, i.e. some $\sto$, $\hp$, $n$ and $v$ such that

\begin{minipage}{7cm} 
\begin{description}
    \item[(H1)] $\esem{\pexp_1}{s}=n$
    \item[(H2)] $\esem{\pexp_2}{s}=v$
\end{description}
\end{minipage}
\begin{minipage}{7cm} 
\begin{description}
    \item[(H3)] $h(n)\in\Val$
    \item[(H4)] $(\sto,\hp), \pmutate{\pexp_1}{\pexp_2} \baction_{\fictx} \osucc: (\sto,\hp[n \mapsto v])$
\end{description}
\end{minipage}

\noindent and some $\sint$, $\ssto$, $\smem$ and $\spc$ such that \textbf{(H5)}~$\sint(\ssto,\smem,\spc)=(\sto,\hp)$, i.e.

{\qquad 
{\bf (H5a)} $\sint(\ssto)=\sto$ \quad 
{\bf (H5b)} $\sint(\smem)=\hp$ \quad 
{\bf (H5c)} $\sint(\spc)=\spc$
}

\noindent (H1), (H5) and Lemma \ref{thm:ox-sound-eval} imply the existence of some $\sym{n}$ and $\spc'$ such that
\qquad      {\bf (H6a) }$\cseeval{\pexp_1}{\ssto}{\spc}{\sym{n}}{\spc'}$ \quad 
    {\bf (H6b) } $\sint(\spc') = \true$ \quad 
    {\bf (H6c) } $\sint(\sym{n}) = n$
(H3), (H5b) and (H6c) imply the existence of some $\sym{m}\in\dom(\smem)$ such that
\begin{description}
    \item[(H7a)] $\sint(\sym{n})=\sint(\sym{m})$
    \item[(H7b)] $\spc''\eqdef(\sym{n}=\sym{m})~\land~\spc'$ {\it (This is merely a new definition.) }
    \item[(H7c)] $\sint(\spc'')=\true$. {\it (This is a trivially implication of the above two, and it implies $\sat(\spc'')$.) }
    \item[(H7d)] $\smem(\sym{m})\neq\sym{w}$ for some $\sym{w}\in\SVal$ (i.e. it is not negative resource.)
\end{description}
(H2), (H5), (H7c) and Lemma \ref{thm:ox-sound-eval} imply the existence of some $\sym{v}$ and $\spc'''$ such that
\begin{description}
    \item[(H8a)] $\cseeval{\pexp_2}{\ssto}{\spc''}{\sym{v}}{\spc'''}$
    \item[(H8b)] $\sint(\spc''') = \true$
    \item[(H8c)] $\sint(\sym{v}) = v$
\end{description}
We define \textbf{(H9a)}~$\smem'\eqdef\smem[\sym{m}\mapsto\sym{v}]$. 
\begin{align*}
    \sint(\smem') &= \sint(\smem[\sym{m}\mapsto\sym{v}]) 
    = \sint(\smem)[\sint(\sym{m}) \mapsto \sint(\sym{v})] 
    = \hp[n \mapsto v] ~\textbf{(H9b)}
\end{align*}
where the last equality follows from (H5b), (H6c), (H7a) and (H8c). 

(H5a), (H8b) and (H9b) imply \textbf{(H10)}~ $\sint(\ssto,\smem',\spc''') = (\sto, \hp[n\mapsto v])$.

(H6a), (H7b), (H7c), (H7d), (H8a) and (H9a) together with the symbolic execution rule for mutate yield $\csesemtrans{\ssto, \smem,  \spc}{\pmutate{\pexp_1}{\pexp_2}}{\ssto, \smem',\spc'''}{\fictx}{\osucc}$ which, toegether with (H10) implies the desired result.

\subsubsection*{UX Soundness} 

\subparagraph{Mutate.} Rule:
\[
        \inferrule[\textsc{Mutate}]
        {\cseeval{\pexp_1}{\ssto}{\spc}{\sval_1}{\spc'} \quad  \quad \smem(\sloc) = \sval_m \quad \spc'' = (\sloc = \sval_1) \land \spc' \\\\ \sat(\spc'') \quad \cseeval{\pexp_2}{\ssto}{\spc''}{\sval_2}{\spc'''} \quad \smem' = \smem[\sloc \mapsto \sval_2]  }
        {\csesemtrans{\ssto, \smem,  \spc}{\pmutate{\pexp_1}{\pexp_2}}{\ssto, \smem',\spc'''}{\fictx}{\osucc}}
\]
We assume
\[
\csesemtrans{\ssto, \smem,  \spc}{\pmutate{\pexp_1}{\pexp_2}}{\ssto, \smem',\spc'''}{\fictx}{\osucc}
\]
which yields

\begin{minipage}{6cm}
\begin{description}
        \item[(H0)] $\sinv_{\spc}(\ssto)$ and $\sinv_{\spc}(\smem)$
        \item[(H1)] $\cseeval{\pexp_1}{\ssto}{\spc}{\sval_1}{\spc'}$
        \item[(H2)] $\smem(\sloc) = \sval_m$
        \item[(H3)] $\spc'' = (\sloc = \sval_1) \land \spc'$
\end{description}
\end{minipage}
\begin{minipage}{6cm}
\begin{description}
        \item[(H3)] $\spc'' = (\sloc = \sval_1) \land \spc'$
        \item[(H4)] $\sat(\spc'')$
        \item[(H5)] $\cseeval{\pexp_2}{\ssto}{\spc''}{\sval_2}{\spc'''}$
        \item[(H6)] $\smem' = \smem[\sloc \mapsto \sval_2]$
\end{description}
\end{minipage}

Now, let $\sint(\ssto,\smem',\spc''')=(\sto, \hp')$, for some $\sint$, $\sto$ and $\hp'$. i.e. (given (H6))
\begin{description}
        \item[(H7a)] $\sint(\ssto) = \sto$
        \item[(H7b)] $\sint(\smem[\sval_l\mapsto\sval_2]) = \hp'$
        \item[(H7c)] $\sint(\spc''') = \true$
\end{description}

(H1), (H3) and (H5) imply \textbf{(H8)} $\spc'''\Rightarrow\spc''\Rightarrow\spc'\Rightarrow\spc$ and (H7c) yields \textbf{(H9)} $\sint(\spc''')=\sint(\spc'')=\sint(\spc')=\sint(\spc)=\true$.

(H0), (H2) and (H9) implies \textbf{(H10)} $\sint(\sval_l)\in\Nat$.

(H3) and (H9) imply \textbf{(H11a)} $\sint(\sval_l)=\sint(\sval_1)$ and we define  \textbf{(H11b)} $n=\sint(\sval_l)=\sint(\sval_1)\in\Nat$, given (H10).

We define \textbf{(H12a)} $\hp=\sint(\smem)$ and (H11b) and (H12) then implies \textbf{(H12b)} $\hp(n)\in\Val$.

(H6), (H7b), (H11a) and (H12a) implies \textbf{(H13)} $\hp'=\hp[n\mapsto\sint(\sval_2)]$.

Given (H0), (H1), (H7a), (H9) and (H11a), Lemma~\ref{thm:ux-sound-eval} implies \textbf{(H14)} $\esem{\pexp_1}{\sto}=n$.

Given (H5), (H7a), (H9) and Lemma \ref{thm:ux-sound-eval} yields \textbf{(H16)} $\esem{\pexp_2}{\sto}=\sintpe{\vint'}(\sval_2)$.

Given (H12b), (H13), (H14) and (H16), the concrete semantics yields
\[
(\sto,\hp), \pmutate{\pexp_1}{\pexp_2} \baction_{\fictx} (\sto,\hp')
\]
and since (H7a), (H9) and (H13) implies $(\sto,\hp)=\sint(\ssto,\smem,\spc)$, concluding the proof.

\subparagraph{Free.} Rule:
\[
        \inferrule[\textsc{Free}]
        {\cseeval{\pexp}{\ssto}{\spc}{\sval}{\spc'} \quad \smem(\sloc) = \sval_m \\\\ \spc'' = (\sloc = \sval) \land \spc' \\\\ \sat(\spc'') \quad
        \smem' = \smem[\sloc \mapsto \cfreed]}
        {\csesemtrans{\ssto, \smem, \spc}{\pdealloc{\pexp}}{\ssto, \smem', \spc''}{\fictx}{\osucc}}
\]

We assume
\[
\csesemtrans{\ssto, \smem, \spc}{\pdealloc{\pexp}}{\ssto, \smem[\sloc \mapsto \cfreed], \spc''}{\fictx}{\osucc}
\]
which yields

\begin{minipage}{7cm}
\begin{description}
        \item[(H0)] $\sinv_{\spc}(\ssto)$ and $\sinv_{\spc}(\smem)$
        \item[(H1)] $\cseeval{\pexp}{\ssto}{\spc}{\sval}{\spc'}$
        \item[(H2)] $\smem(\sloc) = \sval_m$
\end{description}
\end{minipage}
\begin{minipage}{7cm}
\begin{description}
        \item[(H3)] $\spc'' = (\sloc = \sval) \land \spc'$
        \item[(H4)] $\sat(\spc'') \quad$
\end{description}
\end{minipage}

\noindent Now, let $(\sto, \hp')=\sint(\ssto,\smem[\sval_l\mapsto\sval],\spc'')$, for some $\sint$, $\sto$ and $\hp'$. i.e.
\begin{description}
        \item[(H6a)] $\sintp(\ssto) = \sto$
        \item[(H6b)] $\sint(\smem[\sval_l\mapsto\cfreed]) = \hp'$
        \item[(H6c)] $\sint(\spc'') = \true$
\end{description}
(H1) and (H3) imply \textbf{(H7)} $\spc''\Rightarrow\spc'\Rightarrow\spc$ and through (H6c) we obtain \textbf{(H8)} $\sint(\spc'')=\sint(\spc')=\sint(\spc)=\true$.

We define \textbf{(H9)} $\hp=\sint(\smem)$.

Given (H0), (H1), (H8) amd (H6a), Lemma \ref{thm:ux-sound-eval} yields \textbf{(H10a)} $\esem{\pexp}{\sto}=\sint(\sexp)$. Defining $n=\sint(\sexp)$, (H0), (H2), (H3), (H8) and (H10a) imply \textbf{(H10b)} $n=\esem{\pexp}{\sto}=\sint(\sexp)=\sint(\sexp_l)$.
 With (H2) and (H9) we obtain \textbf{(H10c)} $\hp(n)\in\Val$.

(H6b), (H9) and (H10b) yield \textbf{(H11)} $\hp'=\hp[n\mapsto\cfreed]$.

Given (H10b), (H10c) and (H11), the concrete semantics yields
\[
(\sto,\hp),\pdealloc{\pexp} \baction{\fictx} (\sto,\hp')
\]
and (H6a), (H8) and (H9) imply $(\sto,\hp)=\sint(\ssto,\smem,\spc)$, concluding the proof.

\subparagraph{Seq.} Rule:
\[
\inferrule[\textsc{Seq}]
{\csesemtransabstract{\sst}{C_1}{\sst'}{\fictx}{\osucc} \quad \csesemtransabstract{\sst'}{C_2}{\sst''}{\fictx}{\result}}
{\csesemtransabstract{\sst}{C_1;C_2}{\sst''}{\fictx}{\result}}
\]
We assume
\[
\csesemtransabstract{\sst}{C_1;C_2}{\sst''}{\fictx}{\result}
\]
which yields
\begin{description}
        \item[(H1)] $\csesemtransabstract{\sst}{C_1}{\sst'}{\fictx}{\osucc}$
        \item[(H2)] $\csesemtransabstract{\sst'}{C_2}{\sst''}{\fictx}{\result}$
\end{description}
Let $\st''=\sint(\sst'')$ for some $\sint$.
The inductive hypothesis and (H2) imply that for $\st'=\sint(\sst')$ that we have \textbf{(H3)} $\st',\cmd_2\baction_{\gamma} \outcome: \st''$.

The inductive hypothesis and (H1) imply that for $\st=\sint(\sst)$ that we have \textbf{(H4)} $\st,\cmd_1\baction_{\gamma}\osucc: \st'$.

Given (H3) and (H4), the concrete semantics imply
\[
\st,\cmd_1;\cmd_2\baction_{\gamma} \outcome: \st''
\]
As $\st=\sint(\sst)$, the proof is concluded.

\subparagraph{Mutate-Err-Val-1.} Rule:
\[
        \inferrule[\textsc{Mutate-Err-Val-1}]
        {\cseeval{\pexp_1}{\ssto}{\spc}{\undefd}{\spc'} \quad \\
        \sverr = [\mathstr{\mathsf{ExprEval}}, \stringify{\pexp_1}]}
        {\csesemtrans{\ssto, \smem,  \spc}{\pmutate{\pexp_1}{\pexp_2}}{\ssto_{\oerr}, \smem, \spc'}{\fictx}{\oerr}}
\]
We assume
\[
\csesemtrans{\ssto, \smem, \spc}{\pmutate{\pexp_1}{\pexp_2}}{\ssto_\oerr, \smem,  \spc'}{\fictx}{\oerr}
\]
which yields

\begin{minipage}{7cm}
\begin{description}
        \item[(H0)] $\sinv_{\spc}(\ssto)$
        \item[(H1)] $\sat(\spc')$ (from Lemma~\ref{thm:eval-sat})
\end{description}
\end{minipage}
\begin{minipage}{7cm}
\begin{description}
        \item[(H2)] $\cseeval{\pexp_1}{\ssto}{\spc}{\undefd}{\spc'}$
        \item[(H3)] $\sverr = [\mathstr{\mathsf{ExprEval}}, \stringify{\pexp_1}]$
\end{description}
\end{minipage}

Now, let $(\sto',\hp)=\sint(\ssto_\oerr,\smem,\spc')$, i.e.
\begin{description}
        \item[(H4a)] $\sint(\ssto[\pvar{err} \mapsto \sverr])=\sto'$
        \item[(H4b)] $\sint(\smem)=\hp$
        \item[(H4c)] $\sint(\spc')=\true$
\end{description}

(H2) implies \textbf{(H5)} $\spc'\Rightarrow\spc$, which implies with (H4c) that \textbf{(H6)} $\sint(\spc')=\sint(\spc)=\true$.

Define \textbf{(H7)} $\sto=\sint(\ssto)$.

Given (H0), (H2), (H6) and (H7), Lemma~\ref{thm:ux-sound-eval} yields \textbf{(H8)} $\esem{\pexp}{\sto}=\undefd$.

As $\sverr$ has no symbolic variables, we have \textbf{(H9)} $\verr=\sint(\sverr)=\sverr$.

(H4a), (H7) and (H9) implies that \textbf{(H10)} $\sto'=\sto[\pvar{err}\mapsto\verr]$.

Given (H10), (H8), (H9) and (H3), the operational semantics implies
\[
        (\sto,\hp),\pmutate{\pexp_1}{\pexp_2} \baction_{\fictx} \oxerr: (\sto',\hp)
\]
(H7), (H4b) and (H6) imply that $ (\sto,\hp)=\sint(\ssto,\smem,\spc)$, concluding the proof.

%% file: sections/app-symbolic-soundness-func.tex

\section{Compositional Symbolic Execution: Soundness}\label{app:symbolic-soundness-func}

In this appendix we present the two soundness proofs of Thm.~\ref{thm:ux-ox-sound-func} for the consume-produce-based function call, predicate fold, and predicate unfold rules. We present a selection of proof cases, the other rules are either the same as in the proof of Thm.~\ref{thm:ux-ox-sound}, similar to the cases we present here, or simple.

\subsection{Further Definitions for Function Specifications}

\subparagraph{Satisfaction Relation for Assertions.} The logical expression evaluation, $\esem{\lexp}{\subst,\sto}$, extends program expression evaluation to interpret the logical variables using $\theta$. The satisfaction relation for assertions, denoted by $\subst, \st \models P$, is standard:
\[
\begin{array}{@{}l@{~}c@{~\ }l}
 \subst, (\sto, \hp)  \models \\
\begin{array}{@{}l@{~}c@{~\ }l@{\qquad\qquad}l@{~}c@{~\ }l}
 \quad {{\pc}} & \hspace{0.05cm}\Leftrightarrow & \esem{{\pc}}{\subst,\sto} = \true \wedge \hp = \emptyset &
  \AssFalse&\Leftrightarrow& \mbox{never}\\
 \quad P_1 \Rightarrow  P_2 &\hspace{0.05cm}\Leftrightarrow& \subst, (\sto, \hp)  \models P_1
                                         \Rightarrow  \subst, (\sto ,
                                          \hp) \models P_2 &
  P_1 \lor P_2 &\Leftrightarrow& \subst, (\sto, \hp)  \models P_1 \lor  \subst, (\sto, \hp) \models P_2 \\
  \quad\exists \lvar{x} \ldotp P &\hspace{0.05cm}\Leftrightarrow& \exists v \in \Val \ldotp \subst[\lvar{x} \mapsto v], (\sto, \hp) \models P &
  {\emp} &\Leftrightarrow& {\hp = \emptyset } \\
 \quad {\lexp_1 \mapsto \lexp_2} &\hspace{0.05cm}\Leftrightarrow& \hp = \{ \esem{\lexp_1}{\subst, \sto} \mapsto \esem{\lexp_2}{\subst, \sto}\} 
 & {\lexp \mapsto \cfreed} &\Leftrightarrow& \hp = \{ \esem{\lexp}{\subst, \sto} \mapsto \cfreed\} \\
\end{array} \\
\begin{array}{@{}l@{~}c@{~\ }l@{\qquad\qquad}l@{~}c@{~\ }l}
 \quad P_1 \lstar P_2 &\Leftrightarrow& \exists \hp_1, \hp_2 \ldotp \hp = \hp_1 \uplus \hp_2 \land \subst, (\sto, \hp_1) \models P_1 \wedge \subst, (\sto, \hp_2) \models P_2 \\
 \quad \Predd{\vec{\lexp_1}}{\vec{\lexp_2}} &\Leftrightarrow& \subst[\predin \mapsto \esem{\vec{\lexp_1}}{\subst, \sto}, \predout \mapsto \esem{\vec{\lexp_2}}{\subst, \sto}], (\sto, \hp) \models P \text{ for } \Predd{\predin}{\predout}~\{ P \} \in \preds
\end{array}
\end{array}
\]
Note that the meaning of Boolean assertions $\pi$ is defined using the empty heap; the alternative is for them to be satisfied in any heap, definable here as $\pi \lstar \AssTrue$. The meaning of the predicate assertion $\Predd{\vec{\lexp_1}}{\vec{\lexp_2}}$ is defined in terms of its unfolding. An assertion $P$ is {\it valid}, denoted $\models P$, iff $\forall \subst, \sto, \hp.~\subst, (\sto, \hp) \models P$.

\subparagraph{OX and UX Specifications.} 
OX and UX specifications
have, respectively, the form
\[
  \quad \quad \quadruple{P}{}{\Qok}{\Qerr}
  \islquadruple{P}{}{\Qok}{\Qerr} 
  \]
for precondition $P $, successful postcondition $\Qok$  and faulting
postcondition~$\Qerr$. To denote either type of specification, we write 
$\genquadruple{P}{\!}{\Qok}{\Qerr}$.
If execution only succeeds or only faults, we omit the
other post-condition instead of putting an unsatisfiable assertion.
Given implementation context $\gamma$, we define $\fictx \models
\genquadruple{P}{\cmd}{\Qok}{\Qerr}$, denoting that specification
$
\genquadruple{P}{\!}{\Qok}{\Qerr}$ of command $\cmd$ is
$\gamma$-valid, by
%
\[
  \begin{array}{l}
\fictx \models \quadruple{P}{\cmd}{\Qok}{\Qerr} \eqdef \forall \subst, \sto, \hp, \outcome, \sto', \hp'. \\
\qquad \subst, (\sto, \hp) \models P \land (\sto, \hp), \cmd \baction_{\fictx} \outcome: (\sto', \hp') \implies 
(\outcome \neq \oxm \land \subst, (\sto', \hp') \models Q_\outcome) \\[1mm]
\fictx \models \islquadruple{P}{\cmd}{\Qok}{\Qerr} \eqdef \forall \subst, \sto', \hp', \outcome.\\
\qquad \subst, (\sto', \hp') \models Q_\outcome \implies (\exsts{\sto,\hp}~\subst, (\sto, \hp) \models
    P~\land~(\sto, \hp), \cmd \baction_\fictx \outcome: (\sto', \hp'))
\end{array}
\]

\subparagraph{Internal and External Function Specifications.} Function specifications comprise  {\em external} specifications, which describe the function interface toward its callers,  and {\em internal} specifications, satisfied by the function body. An {external} specification has the form
\[
\genquadruple{\vec{\pvar{x}} = \vec x \lstar P}{}{\Qok}{\Qerr}
\]
where:
\begin{itemize}[leftmargin=*]
\item $\vec{\pvar{x}}$/$\vec x$ are distinct program/logical variables,
$\pv{P} = \emptyset$, and $P$ has no existentials; 
\item $\Qok/\Qerr$ are either unsatisfiable or are of the form $\exists \vec y.~Q$, with 
$Q$ having no existential quantification and $\pv{Q} = \{ \pvar{ret}/\pvar{err}\}$.
\end{itemize}
%
%
The constraints on the program variables in external specifications
are well-known from OX logics and follow the usual scoping of the
parameters and local variables of functions. No other program variables can be present in the two post-conditions due to variable scope being limited to the function body.




The relationship between external and internal UX specifications
is more complex than in OX reasoning. In particular, while the
internal UX pre-condition extends the
external by simply instantiating the locals to null,
the transition from internal to external post-condition must ensure
that: \etag{P1} no information is lost (for UX soundness); and~\etag{P2}~no locals leak into
the calling context (given variable scoping). To capture this, as is done in the work on ESL~\cite{esl}, we define an 
{\em internalisation} function, $\fext_{\fictx, \fid}^\macUX$, which takes an 
  implementation context~$\fictx$, a function
  $\pfunction{\fid}{\lst{\pvar x}}{\cmd;
  \preturn{\pexp}} \in  \gamma$, and a UX specification
and returns the set of all of its internal specifications:
\[
\begin{array}{l}
    \fext_{\fictx, \fid}^{\macUX} (\islquadruple{
  \mathit{\vec{\pvar x} = \vec x \lstar P}}{\!}{\Qok}{\Qerr}) \defeq \\
 \qquad
	\{~\islquadruple{\mathit{\vec{\pvar x} = \vec x \lstar P} \lstar \vec {\pvar
           z}\doteq\nil}{\!}{\Qok'}{\Qerr'}~{:}~\\\qquad\qquad\qquad~\models (\Qok'
                                              \Rightarrow \pexp \in
                                              \Val \lstar
                                            \AssTrue), \\\qquad\qquad\qquad~
        \models (\Qok \Rightarrow \exsts{\vec{p}} \Qok'\{\vec{p} / \vec{\pvar p}\}\lstar\pvar{ret}\doteq \pexp\{\vec{p} / \vec{\pvar p}\}),\\\qquad\qquad\qquad~\models (\Qerr \Rightarrow \exsts{\vec{p}} \Qerr'\{\vec{p} / \vec{\pvar p}\})~\}
\end{array}
\]
where 
$\pvvar z=\pv{\cmd}\backslash \{\vec{\pvar{x}} \} $, $\vec{\pvar p} = \{\vec{\pvar{x}} \} \uplus \{\pvvar z\}$, and the logical variables $\vec p$ are fresh w.r.t.~$\Qok$ and $\Qerr$.
In particular, the external post-condition must be equal to or stronger than the internal one (ensuring P1) in which the parameters and local variables have been replaced by fresh existentially quantified logical variables (ensuring P2).
OX specification internalisation, denoted by
$\fext_{\fictx, \fid}^\macOX$, is defined analogously, with the two latter
implications reversed: e.g.,
$\models (\Qok \Leftarrow \exsts{\vec{p}} \Qok'[\vec{p} / \vec{\pvar
    p}]\lstar\pvar{ret}\doteq \pexp\{\vec{p} / \vec{\pvar p}\}$).
Explicit OX internalisation is not strictly necessary in SL, however, as information about program variables can be forgotten in internal post-conditions using forward consequence.

The set of external specifications is denoted by $\especs$. A function specification context (also: specification context), $\fsctx \in \fids \rightharpoonup_{\mathit{fin}} \mathcal{P} (\especs)$, is a finite partial function from function identifiers to a finite set of external specifications.

\subparagraph{Function Environments.} A function environment, $(\fictx,\fsctx)$, is a pair comprising an implementation context $\fictx$ and a specification context $\fsctx$, where specification contexts map function identifiers to sets of their external specifications. A function environment $(\fictx,\fsctx)$ is {valid}, written $\models (\fictx, \fsctx)$, iff every function in $\fsctx$ has an implementation in $\fictx$ and for every external specification in $\fsctx$ there exists a corresponding internal specification valid under $\fictx$. Lastly, a function specification is valid in  specification context $\fsctx$, denoted by $\fsctx \models \genquadruple{P}{\fid(\vec{\pvar{x}})}{\Qok}{\Qerr}$ iff for all $\fictx$ such that $\models (\fictx,\fsctx)$ it holds that $\fictx \models \genquadruple{P}{\fid(\vec{\pvar{x}})}{\Qok}{\Qerr}$.

\subsection{Further Definitions for Extended Symbolic States}\label{appssec:extsstates}

From now on, instead of  denoting extended symbolic states as $(\ssto, \hat{H},\spc)$ we will expose the contents of $\hat{H}$ and write 
extended symbolic states in full, as $(\ssto, \smem, \sps, \spc)$.

Well-formedness of extended states is defined by:
\[
\sinv((\ssto, \smem, \sps, \spc)) \defeq \sinv((\ssto, \smem, \spc)) \land \svs{\sps} \subseteq \svs{\spc}
\]
In the main text, we use 
\[\sinvc((\smem_1, \sps_1)) \defeq \sinvc(\smem_1) \text{ and  }(\smem_1, \sps_1) \cup (\smem_2, \sps_2) \defeq (\smem_1 \uplus \smem_2, \sps_1 \cup \sps_2).\]
and we refer to the definition of well-formed symbolic heaps $\sinvc(\smem_1) $ in \S\ref{appssec:furtherdef}.

Well-formed of a symbolic substitution $\ssubst$, used in the consume and produce functions, with respect to $\spc$ is defined to be 
\[\sinv(\ssubst, \spc)\defeq \svs{\ssubst} \subseteq \svs{\spc} \land \spc \models \codom(\ssubst) \subseteq \Val.\]

As some symbolic predicates can be satisfied by multiple states, symbolic interpretations become a relation, $\smodels$, similar to the satisfaction relation for assertions:
\[
\sint, (\sto, \hp) \smodels (\ssto, \smem, \sps, \spc) \Leftrightarrow \exists \hp_1, \hp_2.~\hp = \hp_1 \uplus \hp_2 \land \sint((\ssto, \smem_1, \spc)) = (\sto, \hp_1) \wedge \sint, (\sto, \hp_2) \smodels \sps
\]
where $\sint, \st \smodels \sps$ denotes $\sint, \st \smodels
\scalerel*{\lstar}{\textstyle\sum} \{
\Predd{\vec{\sval}_1}{\vec{\sval}_2} ~|~
\Predd{\vec{\sval}_1}{\vec{\sval}_2} \in \sps \}$.

\begin{example}[Interpretation of symbolic states] Let $(\emptyset, (\{1\mapsto 5\}, \{\llist{\svar{x};[\svar{x}],[15]}\}),\svar{x}\in\Val )$ be a symbolic state, $\sint$ be an interpretation such that $\sint(\svar{x})=2$. In addition, consider the concrete state $(\emptyset, \{1\mapsto 5, 2\mapsto 15, 3\mapsto {\sf null}\})$. We can split the concrete heap in two disjoint parts  such that $\sint(\emptyset,\{1\mapsto 5\},\svar{x}\in\Val)=(\emptyset, \{1\mapsto 5\})$ and $\sint,(\emptyset, \{2\mapsto 15,3\mapsto {\sf null}\}\smodels \llist{\svar{x};[\svar{x}],[15]} $ hold. By definition the latter  is equivalent to  $\sint,(\emptyset, \{2\mapsto 15,3\mapsto {\sf null}\}\smodels P_{list} $, where $P_{list}$ is the predicate's body  as defined in \S~\ref{sub:funspec}.
\end{example}
\subsection{Additional Lemmas}

The following lemmas are convenient in our proofs and follows easily from the fact that our concrete language is compositional (Lem.~\ref{lem:ux-ox-frame}) and from the definition of symbolic state satisfaction (respectively):
\begin{lemma}[Specification Validity with Explicit Frame]\label{lem:valid-with-frame}
\[
\begin{array}{l}
\fictx \models \islquadruple{P}{\cmd}{\Qok}{\Qerr} \eqdef \\
\quad (\forall \subst, \sto', \hp', \hp_f, \outcome.~\subst, (\sto', \hp') \models Q_\outcome \land \hp_f~\sharp~\hp' \implies \\
    \quad \quad  (\exsts{\sto,\hp}~\subst, (\sto, \hp) \models
    P~\land~(\sto, \hp \uplus \hp_f), \cmd \baction_\fictx \outcome:
    (\sto', \hp' \uplus \hp_f)))\\
\fictx \models \quadruple{P}{\cmd}{\Qok}{\Qerr} \eqdef \\
\quad (\forall \subst, \sto, \hp, \hp_f, \outcome, \sto', \hp''.~\subst, (\sto, \hp) \models P \land (\sto, \hp \uplus \hp_f), \cmd \baction_{\fictx} \outcome: (\sto', \hp'') \implies \\
\qquad (\outcome \neq \oxm \land \exists \hp'.~\hp'' = \hp' \uplus \hp_f \land \subst, (\sto', \hp') \models Q_\outcome)) \\
\end{array}
\]
\end{lemma}

\begin{lemma}[State decomposability]
\label{lem:state-decompose}
\[
\begin{array}{r@{\ }l}
& (\ssto, \smem, \sps, \spc) = 
 (\ssto, \smem_a, \sps_a, \spc) \cdot (\ssto, \smem_b, \sps_b, \spc) \\
\land & \sint, \st \models (\ssto, \smem, \sps, \spc) \\
\land & \sint, \st_a \models (\ssto, \smem_a, \sps_a, \spc)\\
\implies & \exists \st_b.\ \sint,\st_b \models (\ssto, \smem_b, \sps_b, \spc) \land \st = \st_a \cdot \st_b
\end{array}
\]
\end{lemma}

\subsection{Underapproximate Soundness Proof}\label{app:symbolic-soundness-func-ux}

The formal definition of strictly exact assertions~\cite[p.~149]{Yangphd}, which are assertions that are satisfiable by at most one heap, is, if $\subst, (\sto, \hp) \models P$ and $\subst, (\sto, \hp') \models P$, then $\hp' = \hp$. We call a predicate strictly exact iff its body is strictly exact.

\subparagraph{UX soundness function call case.} Here, we present the proof of for $\oxok$-outcome for function calls, that is, when $\cmd = \pfuncall{\pvar{y}}{\procname}{\vec{\pexp}}$.

We have a symbolic step as follows:
\begin{mathpar}
\small
\infer[\mathsf{FCall}]
{\csesemtransabstractm{(\ssto,\smem,\sps,\spc)}{\passign{\pvar{y}}{\fid(\vec{\pexp})}}{(\ssto[\pvar{y} \mapsto \svar{r}],\smem',\sps',\spc''')}{\fsctx}{\macUX}{\osucc}}
{
\begin{array}{r@{\hspace{.5cm}}l}
\mbox{\hyp{1}} & \cseeval{\vec{\pexp}}{\ssto}{\spc}{\vec{\sval}}{\spc'}  \\
\mbox{\hyp{2}}&\islquadruple{\vec{\pvar{x}} = \vec{x} \lstar P}{\fid(\vec{\pvar{x}})}{\Qok}{\Qerr} \in \fsctx(f)|_{\macUX} \\
\mbox{\hyp{3}}&\hat{\theta} = \{  \vec{\sval}/\vec{x}\} \\
\mbox{\hyp{4}}&\mac(\macUX, P, \hat{\theta},(\ssto,\smem,\sps,\spc')) \rightsquigarrow (\hat{\theta}', (\ssto,\smem_f,\sps_f,\spc''))\\
\mbox{\hyp{5}}&\Qok = \exists \vec{y}.~\Qok'  \\
\mbox{\hyp{6}}&\vec{\hat{z}} \text{ fresh} \\
\mbox{\hyp{7}}&\ssubst''=\ssubst{'}\uplus \{ \vec{\hat{z}}/\vec{y}\}\\
\mbox{\hyp{8}}& r, \hat{r} \text{ fresh}   \\
\mbox{\hyp{9}}&\Qok'' = \Qok'\{r/\pvar{ret}\} \text{~and~} \ssubst''' = \ssubst''\uplus\{\hat{r}/r \}  \\
\mbox{\hyp{10}}&\textsf{produce}(\Qok'', \ssubst''', (\ssto,\smem_f, \sps_f,\spc'')) \rightsquigarrow  (\ssto,\smem',\sps', \spc''')
\end{array}
}
\end{mathpar}
where  \(\sst=(\ssto,\smem,\sps,\spc)\) and the final symbolic state is $\sst'=(\ssto[\pvar{y} \mapsto \svar{r}],\smem',\sps',\spc''')$ and, by definition of \produce, \hyp{10a} \(\spc'''=\spc''\wedge \spc_q\), for some $\spc_q$, such that $\sat(\spc''')$.
Our task is now to construct an analogous concrete step.


 Let $\sint$ and $\st'=(\sto', \cmem')$ be such that $\sint, \st'\smodels \sst'$. By definition of $\smodels$:
 
 \[
\begin{aligned}
 \qquad \sint, \st' \smodels \sst' \Leftrightarrow~ &\exists \hp_1', \hp_2'.~\hp = \hp_1' \uplus \hp_2' \land \sint(\sst'.\fieldsst{sto}) = \sto' \wedge \sint(\smem') = \hp_1' \wedge \\
 &\sint(\spc''') = \true~\wedge 
                                          \sint, (\sto', \hp_2') \smodels \sps'
\end{aligned}
\]


 From soundness of \mac and \produce (Prop.~\ref{prop:soundness}), we have:
\[
\begin{array}{l@{\hspace{1cm}}l}
\mbox{\hyp{4a}}\quad \smem= \smem_p \uplus \smem_f,  & \mbox{\hyp{4b}} \quad \sps = \sps_p\cup \sps_f\\
\mbox{\hyp{10b}}\quad \smem'= \smem_f \uplus \smem_{\Qok''} &\mbox{\hyp{10c}} \quad \sps' = \sps_f\cup \sps_{\Qok''}
\end{array}
\]


From the model $\st'$ of $\sst'$ and, by  the definition of $\smodels$, one has  that  
 \[\mbox{\hyp{10d}} \qquad \exists \hp_f,\hp_{q}.~\hp_1'=\hp_f'\uplus \hp_{q}'\wedge \sint(\smem_f)=\hp_f'\wedge  \sint(\smem_{\Qok''})=\hp_{q}'\]

The following hold from \hyp{10a} - \hyp{10d}:
\begin{description}

 \item[\hyp{10e}] \( \exists \hp_f'',\hp_{q}''.~\hp_2'=\hp_f''\uplus \hp_{q}''\wedge \sint,(\sto',\hp_f'')\smodels\sps_f\wedge   \sint,(\sto',\hp_q'')\smodels\sps_{\Qok''}\); 

\item[\hyp{10f}] $\sint, (\sto', \hp_f)\smodels (\ssto,\smem_f,\sps_f,\spc'')$, where $\hp_f=\hp_f'\uplus \hp_f''$;
\item[\hyp{10g}] Define $\sst_f=(\ssto,\smem_f,\sps_f,\spc'')$.
 \end{description}
 
 Similarly, it follows that there exists  a concrete state 
 $\st_q'=(\sto',\hp_q)$ such that   
 \begin{description}
  \item[\hyp{11}]~$\sint,\st_q'\smodels (\ssto,\smem_{\Qok''},\sps_{\Qok''},
  \spc''')$
  \item[\hyp{11a}]~$\hp_q=\hp_{q}'\uplus \hp_q''$.
  \item[\hyp{11b}]~$\sint(\ssubst'''),\st_q'\models \Qok''$, by soundness of \produce (Prop.~\ref{prop:soundness}). 
  Thus,
   $\sint(\ssubst'''),\st_q'\models \Qok'\{r/\pvar{ret}\}$, by definition.
   \item[\hyp{11c}]\(\sint(\ssubst'''), (\emptyset, \hp_q) \models \Qok'\{r/\pvar{ret}\}\),
 since there are no program variables in 
 \(\Qok'\{r/\pvar{ret}\}\). 
 \end{description}
   
Notice that  \hyp{10e} $\dish{\hp_f}{\hp_q}$, \hyp{10f}, \hyp{11c} and
 completeness of \produce (Prop.~\ref{prop:prod-compl}), implies that  
 there exists $\sst_{\Qok''}$ such that $\sst'=\sst_f\cdot \sst_{\Qok''}$ 
 and that   \hyp{11d}~$\sint, (\emptyset,\hp_q)\models \sst_{\Qok''}$. 
 Take $\sst_{\Qok''}=(\emptyset,\smem_{\Qok''},\sps_{\Qok''},\spc''')$.

By the hypothesis \(\models (\fictx, \fsctx) \), one has, for some $\cmd'$ and $\pexp'$, \hyp{12} $f(\vec{\pvar{x}}) \{\cmd'; \preturn{\pexp'}\} \in \fictx$ and a valid UX quadruple
\[
\mbox{\hyp{12a}}\quad \gamma \models \islquadruple{\vec{\pvar{x}} = \vec{x} \star P\lstar \vec{\pvar{z}}=\pvar{null} }{\cmd'}{\Qok^*}{\Qerr^*}
\]
such that 

\begin{description}
\item[\hyp{12b}] \( \models \Qok^* \implies \pvar{E}'\in \Val\)
\item[\hyp{12c}] \( \models \Qok \implies\exsts{\vec{p}} \Qok^*\{\vec{p} / \vec{\pvar p}\}\lstar\pvar{ret}\doteq \pexp'\{\vec{p} / \vec{\pvar p}\}~\text{and}\)
\item[\hyp{12d}] \( \models \Qerr \implies\exsts{\vec{p}} \Qerr^*\{\vec{p} / \vec{\pvar p}\} \)
\end{description}
where  $\pvvar z=\pv{\cmd'}\backslash\{\vec{\pvar{x}}\}$, $\vec{\pvar p} = \{\vec{\pvar{x}}, \vec{\pvar{z}}\}$, and the logical variables $\vec p$ are fresh with respect to $\Qok$ and $\Qerr$.

Now,  \hyp{5} \(\Qok=\exists \vec{y}.\Qok'\) and  \hyp{11c} \(\sint(\ssubst'''), (\emptyset, \hp_q) \models \Qok'\{r/\pvar{ret}\}\) entails $\sint(\ssubst'''), (\emptyset, \hp_q) \models~\Qok\{r/\pvar{ret}\}$. By property
 \hyp{12c} one has  \( \sint(\ssubst'''), (\emptyset,\hp_q) \models \exsts{\vec{p}} \Qok^*\{\vec{p} / \vec{\pvar p}\}\lstar r \doteq \pexp'\{\vec{p} / \vec{\pvar p}\}\). Thus, there exists $\vec{w}\in\Val$ such that 
 \[\mbox{\hyp{13}} \qquad\sint(\ssubst'''), (\emptyset[\vec{\pvar{p}}\mapsto \vec{w}], \hp_q) \models  \Qok^*\lstar r \doteq \pexp'.\]
Define
\begin{description}
\item[\hyp{13a}] $s_q = \emptyset [\vec{\pvar{p}}\mapsto \vec{w}]$;
\item[\hyp{13b}] $\theta_0=\sint(\ssubst''')$.
\end{description}
By hypothesis \hyp{9} \(\ssubst'''=\ssubst''\uplus\{r\mapsto \svar{r}\}\) we obtain \hyp{13d} $\esem{\pexp'}{\subst_0,\sto_q}=\theta_0(r)=\sint(\svar{r})$.

Validity of the UX quadruple \hyp{12a} \(\gamma \models \islquadruple{\vec{\pvar{x}} = \vec{x} \star P\lstar \vec{\pvar{z}}=\pvar{null} }{\cmd'}{\Qok^*}{\Qerr^*} \) together with Lem.~\ref{lem:valid-with-frame}, \hyp{13}, \hyp{10a}, and  \hyp{10e} $\dish{\hp_f}{\hp_{q}}$ imply that 
\[\mbox{\hyp{14}}\quad \exsts{\sto_p,\hp_p}~\subst_0, (\sto_p, \hp_p) \models \vec{\pvar{x}}=\vec{x}\lstar P \lstar \vec{\pvar{z}}=\pvar{null}~\land~(\sto_p, \hp_p \uplus \hp_f), \cmd' \baction_\fictx ok: (\sto_q, h_q \uplus \hp_f)
\]
Using \hyp{13b},  we obtain $\sint(\ssubst'''), (\sto_p, \hp_p) \models P$. Since $P$ does 
not have any program variable, we can conclude
$\sint(\ssubst'''), (\emptyset, \hp_p) \models P$. 
From the fact that  $r,\svar{r}$ are fresh, and  we can 
assume w.l.o.g. that  $\vec{y}$ do not occur in $P$;
 thus, we can conclude that 
  \hyp{15} $\sint(\ssubst'), (\emptyset, \hp_p) \models P$. 
 In particular, we have:
\begin{description}
\item[\hyp{15a}]\(\vint(\ssubst'), (\sto_p,\hp_p)\models \vec{\pvar{x}}=\vec{x}\lstar P\lstar \vec{\pvar{z}}=\pvar{null}\)
\item[\hyp{15b}] \((\sto_p, \hp_p \uplus \hp_f), \cmd' \baction_\fictx ok: (\sto_q, h_q \uplus \hp_f)\), 
\end{description}
where  $\ssto_p= \emptyset[\vec{\pvar{x}}\mapsto \vec{x}][\vec{\pvar{z}}\mapsto \pvar{null}]$ and $s_p=\sint(\ssto_p)=\emptyset[\vec{\pvar{x}}\mapsto \sint(\ssubst'(\vec{x}))][\vec{\pvar{z}}\mapsto \pvar{null}]$ (because any other program variables cannot affect execution), that is, \hyp{15c} $s_p=[\vec{\pvar{x}}\mapsto \vec{v}, \vec{\pvar{z}}\mapsto \pvar{null}]$, where $\sint(\ssubst'(\vec{x}))=\sint(\vec{\sval})=\vec{v}$.

Now, we define $\sst_P=(\emptyset,\smem_p,\sps_p,\spc'')$. 
From \hyp{4}, \hyp{15}, \hyp{10f} and \hyp{14} $\dish{\hp_f}{\hp_p}$, 
we have all the conditions for the completeness of UX \mac (Prop.~\ref{prop:mac-ux-comp})
which implies  $\sint, (\sto,\hp_f\uplus\hp_p)\smodels \sst \wedge \sint, (\emptyset,\hp_p)\smodels \sst_P $.
Since \(\sint(\spc')=\true\), UX soundness of evaluation (Thm.~\ref{thm:ux-sound-eval}), and \hyp{1}, the following holds: \hyp{16} \(\esem{\pexp}{s}=\vec{v}\), for $s=\sint(\ssto)$.
 
From \hyp{12}, \hyp{13d}, \hyp{15b}, \hyp{15c} and  \hyp{16} we obtain all components required for our concrete function call:
\[
   \infer{
     \sthreadp{ \sto }{ \hp_p\uplus \hp_f},
     \pfuncall{\pvar{y}}{\procname}{\vec{\pexp}} \baction_{\fictx}
     \sthreadp{\sto [\pvar{y} \storearrow v' ] }{ \hp_q\uplus \hp_f}
  }{
    \begin{array}{c}
    \pfunction{\procname}{\vec{\pvar{x}}}{\scmd'; \preturn{\pexp'}} \in \scontext
    \\
    \esem{\vec{\pexp}}{\sto} = \vec{v} 
     \quad \pv{\scmd'} \setminus
      \{\vec{\pvar{x}}\} = \{\vec{\pvar{z}} \}
    \\
 \sto_p  =  \emptyset [ \vec{\pvar{x}} \storearrow \vec{v}] [ \vec{\pvar{z}} \storearrow \nil]
     \\ (\sto_p, \hp_p\uplus \hp_f), \scmd'
      \baction_{\fictx} \sthreadp{ \sto_q }{ \hp_q\uplus \hp_f}  \quad  \esem{\pexp'}{\sto_q} =v' 
    \end{array}
  }
\]
and the results follows.

\subparagraph{Folding Predicates Case: $\cmd = \pfold{\pred}{\vec{\pexp}}$.} We want to prove the following:
\[\models (\fictx, \fsctx) \land
\csesemtransabstractm{\sst}{\pfold{\pred}{\vec{\pexp}}}{\sst'}{\fsctx}{\macUX}{\result} \land
\vint, \st' \smodels \sst' \implies
\exists \st.~\vint, \st \smodels \sst \land \st, \pfold{\pred}{\vec{\pexp}} \baction_{\fictx} \result : \st'\]

Since $\st, \pfold{\pred}{\vec{\pexp}} \baction_{\fictx} \osucc : \st$, we take $\st = \st'$. I.e., we must show $\sint, \st' \smodels \sst$.

The rule for folding predicates is: 

\begin{mathpar}
\mprset{flushleft}
\inferrule{
\mbox{ \hyp{1}}\\ \pred(\predin)(\predout)~\{\bigvee_i (\exists \vec{x}_i.~A_i)\} \in \preds \\\\
\mbox{ \hyp{2}}\\  \cseeval{\vec{\pexp}}{\ssto}{\spc}{\vec{\sval}}{\spc} \\\\
\mbox{ \hyp{3}}\\ \ssubst \defeq [\predin \mapsto \vec{\sval} ] \\\\
\mbox{ \hyp{4}}\\ \mac(m, A_i, \ssubst, \sst) \rightsquigarrow  (\ssubst', \sst_i) \\\\
\mbox{ \hyp{5}} \\  \hat{\mathcal{P}}' \defeq \{ \pred(\vec{\sval})(\ssubst'(\predout)) \} \cup \sst_i.\fieldsst{preds}}
 {\csesemtransabstractm{\sst}{\pfold{\pred}{\vec{\pexp}}}{\sst_i[\sstupdate{preds}{\sps'}]}{\fsctx}{m}{\osucc}}
\end{mathpar}
where $\sst=(\ssto,\smem,\sps,\spc)$ and \mbox{ \hyp{6}} $\sst'=\sst_i[\sstupdate{preds}{\sps'}]$.

Let $\sint$ and $\st'=(\sto,\hp')$ be such that $\sint,\st'\smodels \sst'$. Then, by definition of $\smodels$:
\[
\begin{aligned}
\mbox{\hyp{7}} \qquad \sint, \st' \smodels \sst' \Leftrightarrow~ &\exists \hp_1', \hp_2'.~\hp = \hp_1' \uplus \hp_2' \land \sint(\sst'.\fieldsst{sto}) = \sto \wedge \sint(\sst'.\fieldsst{heap}) = \hp_1' \wedge \sint(\sst'.\fieldsst{pc}) = \true~\wedge \\
                                          &\sint, (\emptyset, \hp_2') \smodels \scalerel*{\lstar}{\textstyle\sum} \{ \pred(\vec{\sval}_1)(\vec{\sval}_2) ~|~ \pred(\vec{\sval}_1)(\vec{\sval}_2) \in \sps' \}
\end{aligned}
\]
From  $\sint, (\emptyset, \hp_2')\smodels \scalerel*{\lstar}{\textstyle\sum} \{ \pred(\vec{\sval}_1)(\vec{\sval}_2) ~|~ \pred(\vec{\sval}_1)(\vec{\sval}_2) \in \sps' \}
$, it follows that 
\[
\begin{aligned}
 \exists \hp_{21}', \hp_{22}'. ~\hp_2'= \hp_{21}'\uplus \hp_{22}'~\wedge~ &\sint, (\emptyset, \hp_{21}') \smodels \pred(\vec{\sval})(\ssubst'(\predout))\\
 & \wedge \sint, (\emptyset, \hp_{22}') \smodels \scalerel*{\lstar}{\textstyle\sum} \{ \pred(\vec{\sval}_1)(\vec{\sval}_2) ~|~ \pred(\vec{\sval}_1)(\vec{\sval}_2) \in \sps' \}
 \end{aligned}
\]

Thus, \mbox{ \hyp{8}} $\sint,(\sto, \hp_1'\uplus \hp_{22}')\smodels \sst_i$.

Moreover, by definition, $\sint, (\emptyset, \hp_{21}')\smodels \pred(\vec{\sval})(\ssubst'(\predout)) \Leftrightarrow \subst, (\emptyset, \hp_{21}')\models \bigvee_j (\exists \vec{x}_j.~A_j)$ where $\subst = \{ \sint(\vec{\sval})/ \predin, \sint(\ssubst'(\predout))/\predout \}$.

By soundness of \mac (Prop.~\ref{prop:soundness}), we have $\smem_f$ and $\sps_f$ such that:
\[
\begin{array}{l@{\hspace{1cm}}l}
\mbox{\hyp{4a}}\quad \smem= \smem_{A_i} \uplus \smem_f,  & \mbox{\hyp{4b}} \quad \sps = \sps_{A_i}\cup \sps_f
\end{array}
\]
which implies that \hyp{9} $\sst_i = (\ssto,\smem_f,\sps_f,\spc')$, where $\spc'=\spc\wedge \spc_i$ (Prop.~\ref{prop:path_strength}), for some $\spc_i$, and \hyp{10} $\sst = (\ssto,\smem_{A_i},\sps_{A_i},\spc) \cdot (\ssto,\smem_f,\sps_f,\spc)$.



We have $\sint(\spc') = \true \Rightarrow \sint(\spc) = \true$, and hence, there exist $\cmem$ such that $\sint, (\sto, \cmem) \smodels \sst$. In consequence together with \hyp{10}, there exist $\cmem_{A_i}$ such that $\sint, (\sto, \cmem_{A_i}) \smodels (\ssto, \smem_{A_i}, \sps_{A_i}, \spc)$. It follows from soundness of \mac that $\sint(\ssubst'), (\emptyset, \cmem_{A_i}) \models A_i$. This in turn implies $\sint(\ssubst'), (\emptyset, \cmem_{A_i}) \models \exists \vec{x}_i.~A_i$, which in turn together with Prop.~\ref{prop:coverage} implies $\subst, (\emptyset, \cmem_{A_i}) \models \exists \vec{x}_i.~A_i$, which in turn implies $\subst, (\emptyset, \cmem_{A_i}) \models \bigvee_j \exists\vec{x}_j.~A_j$. Since for UX folding we require the predicate to be strictly exact, it follows that $\cmem_{A_i} = \cmem'_{21}$.

\hyp{8} and \hyp{9} implies $\sint, (\sto, \hp_1'\uplus \hp_{22}') \smodels (\ssto, \smem_f, \sps_f, \spc)$. \hyp{10} now gives $\sint, (\sto, \hp'_{21} \uplus (\hp_1' \uplus \hp_{22}')) \smodels \sst$.

\subparagraph{Unfolding Predicates Case: $\cmd = \punfold{\pred}{\vec{\pexp}}$.}
We want to prove the following:
\[
\models (\fictx, \fsctx) \land
\csesemtransabstractm{\sst}{\punfold{\pred}{\vec{\pexp}}}{\sst'}{\fsctx}{\macUX}{\result} \land
\vint, \st' \smodels \sst' \implies
\exists \st.~\vint, \st \smodels \sst \land \st, \punfold{\pred}{\vec{\pexp}} \baction_{\fictx} \result : \st'
\]

Since $\st, \punfold{\pred}{\vec{\pexp}} \baction_{\fictx} \osucc : \st$, we take $\st = \st'$. We must now show $\sint, \st' \smodels \sst$.

The symbolic rule for unfolding a predicate is:
\begin{mathpar}
\mprset{flushleft}
\inferrule{
\mbox{ \hyp{1}}\\ \pred(\predin)(\predout)~\{\bigvee_i (\exists \vec{x}_i.~A_i)\}\in \preds \\\\
 \mbox{ \hyp{2}}\\\cseeval{\vec{\pexp}}{\ssto}{\spc}{\vec{\sval}}{\spc} \\\\
 \mbox{ \hyp{3}}\\\conspred(\pred, \vec{\sval},\sst)\rightsquigarrow ( \spredout,\sst_p) \\\\
 \mbox{ \hyp{4}}\\\vec{\sym{z}} \text{ fresh} \\\\
\mbox{ \hyp{5}}\\  \ssubst \defeq [\predin \mapsto \vec{\sval}, \predout\mapsto \spredout, \vec{x}_i \mapsto \vec{\sym{z}}] \\\\
 \mbox{ \hyp{6}}\\ \produce(A_i,\ssubst,\sst_p) \rightsquigarrow \sst_p'}
 {\csesemtransabstractm{\sst}{\punfold{\pred}{\vec{\pexp}}}{\sst_p'}{\fsctx}{m}{\osucc}}
\end{mathpar}
where $\sst=(\ssto,\smem,\sps,\spc)$, $\sst_p=(\ssto,\smem_p,\sps_p,\spc_p)$, and \mbox{ \hyp{8}} $\sst' = \sst_p' = (\ssto, \smem'_p, \sps'_p, \spc'_p)$.

From soundness of \produce (Prop.~\ref{prop:soundness}), we obtain that there exist $\smem_{A_i}$ and $\sps_{A_i}$ such that $\sst'_p = (\ssto, \smem_p, \sps_p, \spc'_p) \cdot (\ssto, \smem_{A_i}, \sps_{A_i}, \spc'_p)$.

We have $\sint, \st' \smodels \sst'_p$. Say $\st'=(\sto, \hp')$. There exist $\hp_{A_i}$ and $\hp'_p$ such that $\hp' = \hp_{A_i} \uplus \hp_p$, $\sint, (\sto, \hp_p) \smodels (\ssto, \smem_p, \sps_p, \spc'_p)$, and $\sint, (\sto, \hp_{A_i}) \smodels (\ssto, \smem_{A_i}, \sps_{A_i}, \spc'_p)$. From this it follows that $\sint, (\sto, \hp_p) \smodels (\ssto, \smem_p, \sps_p, \spc_p)$.

From the definition of $\conspred$ it follows that there exist $\spredin$ such that $\sps = \{ \pred(\spredin)(\spredout) \} \cup \sps_p$ and $\spc_p = (\spc \wedge \vec{\sval} = \spredin)$. 

Note that we have from soundness of \produce that $\sint(\ssubst), (\emptyset, \hp_{A_i}) \models A_i$, and hence $\sint(\ssubst), (\emptyset, \hp_{A_i}) \models \pred(\predin)(\predout)$, and hence $\sint, (\emptyset, \hp_{A_i}) \smodels \pred(\ssubst(\predin))(\ssubst(\predout))$, and hence $\sint, (\emptyset, \hp_{A_i}) \smodels \pred(\vec{\sval})(\spredout)$, and hence $\sint, (\emptyset, \hp_{A_i}) \smodels \pred(\spredin)(\spredout)$.

Since $\hp_{A_i}$ models $\pred(\spredin)(\spredout)$ and $\hp'_p$ models $\sps_p$, the result now follows.

\subsection{Overapproximate Soundness Proof}\label{app:symbolic-soundness-func-ox}

We take, as in the previous section, successful function call as our illustrative example for the proof of Thm.~\ref{thm:ux-ox-sound-func}.

\subparagraph{OX soundness function call case.}

We have
\begin{description}
\item[\hyp{1}] $\models (\gamma,\Gamma)$
\item[\hyp{2}] $\csesemtranscollectm{\sst}{\cmd}{\hat\Sigma'}{\fsctx}{\macEX}$
\item[\hyp{3}] $\oxabort \not\in \hat\Sigma'$
\item[\hyp{4}] $\sint, \st \models \sst \land \st,\cmd\baction_{\fictx} \result : \st'$
\end{description}

We start out with the following step in the concrete semantics \hyp{5}:
\[
\infer{
     \sthreadp{ \sto }{ \hp },
     \pfuncall{\pvar{y}}{\procname}{\vec{\pexp}} \baction_{\fictx}
     \sthreadp{\sto [\pvar{y} \storearrow v' ] }{ \hp' }
  }{
    \begin{array}{c}
    \pfunction{\procname}{\vec{\pvar{x}}}{\scmd; \preturn{\pexp'}} \in \scontext
    \\
    \esem{\vec{\pexp}}{\sto} = \vec{v} 
     \quad \pv{\scmd} \setminus \{\vec{\pvar{x}}\} = \vec{\pvar{z}}
    \\
 \sto_p  =  \emptyset [ \vec{\pvar{x}} \storearrow \vec{v}] [ \vec{\pvar{z}} \storearrow \nil]
     \\ (\sto_p, \hp), \scmd
      \baction_{\fictx} \sthreadp{ \sto_q }{ \hp' }  \quad  \esem{\pexp'}{\sto_q} =v' 
    \end{array}
  }
\]
and must now construct an analogous step in the symbolic semantics.

That is, a step using the following rule:
\begin{mathpar}
\small
\infer[\mathsf{FCall}]
{\csesemtransm{\ssto, \smem, \spc}{\passign{\pvar{y}}{\fid(\vec{\pexp})}}{\ssto[\pvar{y} \mapsto \svar{r}], \smem', \spc'''}{\fsctx}{\macOX}{\osucc}}
{
\begin{array}{l}
\cseeval{\vec{\pexp}}{\ssto}{\spc}{\vec{\sval}}{\spc'}  \\
\quadruple{\vec{\pvar{x}} = \vec{x} \lstar P}{\fid(\vec{\pvar{x}})}{\Qok}{\Qerr} \in \fsctx(f)|_{\macOX} \\
\hat{\theta} = \{\vec{\sval}/\vec{x} \}\\
\mac(\macOX, P, \hat{\theta}, (\ssto,\smem, \sps, \spc')) \rightsquigarrow (\hat{\theta}', (\ssto,\smem_f, \sps_f, \spc''))\\
\Qok = \exists \vec{y}.~\Qok'  \\
\vec{\hat{z}} \text{ fresh} \\
\ssubst''=\ssubst{'}\uplus \{\vec{\hat{z}}/\vec{y}\}\\
r, \hat{r} \text{ fresh}   \\
\Qok'' = \Qok'\{r/\pvar{ret}\} \text{~and~} \ssubst''' = \ssubst''\uplus \{\hat{r}/r\}  \\
\produce(\Qok'', \ssubst''', (\ssto,\smem_f, \sps', \spc'')) \rightsquigarrow  (\ssto, \smem', \sps', \spc''')
\end{array}
}
\end{mathpar}

Let $\sigma=(s,h), \st'=(s[\pvar{y}\mapsto v'],\hp')$, and \(\sst=(\ssto,\smem,\sps, \spc)\) such that \hyp{6} $\sint, \st \models \sst$ (from \hyp{4}). Note that in particular we have $\sint(\spc)=\true$.

Say $f(\vec{\pvar{x}}) \{\cmd; \preturn{\pexp'}\} \in \fictx$. From \hyp{1}, the definition of OX valid environments, and the definition of OX specification internalisation one has 
\begin{itemize}
\item \hyp{1a} \(\gamma \models \quadruple{\vec{\pvar{x}} = \vec{x} \star P\lstar \vec{\pvar{z}}=\pvar{null} }{\cmd}{\Qok^*}{\Qerr^*} \)
\item \hyp{1b} \( \models \Qok^* \implies \pvar{E}'\in \Val\)
\item \hyp{1c} \( \models \Qok \Longleftarrow\exsts{\vec{p}} \Qok^*\{\vec{p} / \vec{\pvar p}\}\lstar\pvar{ret}\doteq \pexp'\{\vec{p} / \vec{\pvar p}\}~\text{and}\)
\item \hyp{1d} \( \models \Qerr \Longleftarrow\exsts{\vec{p}} \Qerr^*\{\vec{p} / \vec{\pvar p}\} \)
\end{itemize}
where  $\pvvar z=\pv{\cmd}\backslash\{\vec{\pvar{x}}\}$, $\vec{\pvar p} = \{\vec{\pvar{x, z}}\}$, and the logical variables $\vec p$ are fresh with respect to $\Qok$ and $\Qerr$.

Since \(\esem{\vec{\pexp}}{\sto} = \vec{v}\), \(\sint(\sinvc(\sst))=\true\)  and $\sint(\spc)=\true$,  OX soundness of evaluation (Thm.~\ref{thm:ox-sound-eval}) implies that  there exist  \( \spc'\) and \(\sym{v}\) such that \(~\cseeval{\pexp}{\ssto}{\spc}{\sym{v}}{\spc'}\), \( \sint(\spc') = \true\) and \( \sint(\sym{v}) = v\).
By {\em conditional branch completeness} (Prop.~\ref{prop:mac-branch-comp}) and \hyp{3} there exist  \(\ssubst',\smem_f, \sps_f, \spc'', \st_f\) such that 
\[ \mbox{\hyp{7}}\quad \mac(\macOX, \ssubst, P, (\ssto,\smem, \sps, \spc'))\rightsquigarrow (\ssubst', (\ssto,\smem_f, \sps_f, \spc'')) \gand \sint, \st_f \models (\ssto,\smem_f, \sps_f, \spc'') \]

Now, from {\em soundness of \mac} (Prop.~\ref{prop:soundness}) it follows that there exists $\smem_p, \sps_p$ such that $\smem = \smem_p \uplus \smem_f$ and $\sps = \sps_p \cup \sps_f$. Thus, from \hyp{6} and Lem.~\ref{lem:state-decompose}, we define $\st_p = (s, \hp_p)$ such that
\begin{description}
\item[\hyp{8}] $\sint, \st_p \models (\ssto, \smem_p, \sps_p, \spc'')$
\item[\hyp{9}] $\st = \st_p \cdot \st_f$;
\end{description}

Moreover, from soundness it also follows that \hyp{10}\ \(\sint(\ssubst'), \st_p \models P\). Now, from \hyp{4} and the definition of validity of an OX quadruple (specifically, Lem.~\ref{lem:valid-with-frame}), we obtain that there exists $\hp_q$ such that $\hp'=\hp_q\uplus \hp_f$, i.e., the heap after function body execution, for which the following holds \hyp{11a} \(\sint(\ssubst'), (\sto_q,\hp_q)\models \Qok^*\).
By abstracting away program variables $\vec{\pvar{p}}$ in \(\Qok^*\) with fresh logical variables $\vec{p}$, we obtain
\[ \mbox{\hyp{11b}} \ \sint(\ssubst'), (\sto,\hp_q)\models \exists \vec{p}.~\Qok^*\{\vec{p} / \vec{\pvar p}\}\]
where we can use the store $s$ since \(\exists \vec{p}. \Qok^*\{\vec{p} / \vec{\pvar p}\}\) does not contain program variables. Also,
\[ \mbox{\hyp{11c}} \ \sint(\ssubst'), (\sto[\pvar{ret}\mapsto v'],\hp_q)\models \exists \vec{p}.~\Qok^*\{\vec{p} / \vec{\pvar p}\}\lstar \pvar{ret}=\pexp'\{\vec{p} / \vec{\pvar p}\}\]
where \(\esem{\pexp'}{\sto_q} =v' \).

From \hyp{1c} and \hyp{11c} we can conclude that  the model that satisfies the postcondition of the internal spec \(\Qok^*\) satisfies the postcondition of the corresponding external spec  \(\Qok\): 
\[ \mbox{\hyp{12a}} \ \sint(\ssubst'), (\sto[\pvar{ret}\mapsto v'],\hp_q)\models \Qok\]

Since $\Qok$ is a postcondition, $\Qok=\exists\vec{y}.\Qok'$, which entails 
\[ \mbox{\hyp{12b}} \ \sint(\ssubst'), (\sto[\pvar{ret}\mapsto v'],\hp_q)\models \exists \vec{y}.~\Qok'\]

Since $\Qok''=\Qok'\{r/\pvar{ret}\}$, we also get
\[ \mbox{\hyp{12c}} \ \sint(\ssubst')[\vec{y}\mapsto \vec{z}, r\mapsto v'], (\sto,\hp_q)\models \Qok''\]

Let \(\ssubst'''=\ssubst'\uplus\{\hat{\vec{z}}/\vec{y}~~,  \svar{r}/r\}\) where $\hat{\vec{z}}, \hat{\vec{v}}$ and $\svar{z}$ are fresh. Take $\sint'$ that extends $\sint $ as  $\sint'=\sint\uplus\{\vec{w}/\hat{\vec{z}},  v'/\hat{r}\}$. Then, we have \hyp{12d} $\sint'(\ssubst'''), (s,h_q)\models \Qok''$.

From \hyp{7} (specifically, $\sint(\spc'')=\true$), we get $\sint'(\spc'')=\true$, since $\sint'$ extends $\sint$. Now, $\hp'=\hp_q \uplus \hp_f$ gives $\dish{\hp_q}{\sint'(\smem_f)}$.

By \emph{completeness of \produce} (Prop.~\ref{prop:prod-compl}) it follows that there exist $\smem_q, \sps_q$  and $\spc'''$ such that 
\begin{multline*}
\produce(\Qok'', \ssubst'', (\ssto,\smem_f, \sps_f, \spc''))\rightsquigarrow  (\ssto, \smem_f, \sps_f, \spc'') \cdot (\emptyset,\smem_q, \sps_q, \spc''')=\sst''
\gand \sint, (\emptyset,\hp_q) \models (\emptyset,\smem_q, \sps_q, \spc''').
\end{multline*}

Let $\smem' = \smem_q \uplus \smem_f$ and $\sps' = \sps_q \cup \sps_f$. Let $\sst'=(\ssto[\pvar{y}\mapsto \svar{r}], \smem', \sps', \spc^*)$. We now have constructed everything needed for a {\sf FCall} rule invocation ending in $\sst'$. Since $\sint, (\sto[\pvar{y}\mapsto \lvar{v}'], \cmem') \models \sst'$, the proof is concluded.

\subparagraph{Folding Predicates Case: $\cmd = \pfold{\pred}{\vec{\pexp}}$.}

\begin{description}
\item[\hyp{1}] $\models (\gamma,\Gamma)$
\item[\hyp{2}] $\csesemtranscollectm{\sst}{\cmd}{\hat\Sigma'}{\fsctx}{\macEX}$
\item[\hyp{3}] $\oxabort \not\in \hat\Sigma'$
\item[\hyp{4}] $\sint, \st \models \sst \land \st,\cmd\baction_{\fictx} \result : \st'$
\end{description}

\textbf{To prove}: \[
\exists \sst', \sint' \geq \sint.~(\result, \sst') \in \Sigma' \land \sint', \st' \models \sst'
\]

First of all, from the fact that $\texttt{fold}$ is a ghost command, we know
\[
\mhyp{5}\quad \st' = \st \land \result = \oxok
\]

Then, because symbolic execution does not abort \hyp{3}, by inversion on the fold rule, and by completeness of OX consume (Prop.~\ref{prop:mac-branch-comp}), and the fact that consume does not modify the store (Prop.~\ref{prop:wf}), we learn \hyp{6-11}, $\exists i, \st_f$:
\begin{align*}
 \mhyp{6}\quad & \pred(\predin)(\predout)~\{\bigvee_j (\exists \vec{x}_j.~P_j)\} \in \preds \\
 \mhyp{7}\quad & \cseeval{\vec{\pexp}}{\ssto}{\spc}{\vec{\sval}}{\spc} \quad \ssubst \defeq \{  \vec{\sval}/\predin \} \\
 \mhyp{8}\quad &  \mac(\macOX, P_i, \ssubst, (\ssto, \smem, \sps, \spc)) \rightsquigarrow  (\ssubst', (\ssto, \smem_f, \sps_f, \spc')) \\
 \mhyp{9}\quad & \sint, \st_f \models (\ssto, \smem_f, \sps_f, \spc') = \sst_f\\
 \mhyp{10}\quad & \sps' \defeq \{ \pred(\vec{\sval})(\ssubst'(\predout)) \} \cup \sps_f\\
 \mhyp{11}\quad & \begin{array}{rcl}
   \sst' = (\ssto, \smem_f, \sps', \spc') & = & (\ssto, \smem_f, \sps_f, \spc') \cdot (\ssto, \emptyset, \{ \pred(\vec{\sval})(\ssubst'(\predout)) \}, \spc')\\
   & = & \sst_f \cdot \sst_p
 \end{array}
\end{align*}

where $i$ is the index of the definition of the predicate that has been correctly consumed. By soundness of $\mac$ (Prop.~\ref{prop:soundness}), \hyp{6-8}, $\exists \smem_i, \sps_i$:
\begin{align*}
\mhyp{12}\quad & \smem = \smem_i \uplus \smem_f \\
\mhyp{13}\quad & \sps = \sps_i \cup \sps_f\\
\mhyp{14}\quad & \forall \sint, \st.\ \sint,\st \models (\ssto, \smem_i, \sps_i, \spc') \implies \sint(\ssubst'), \st \models P_i
\end{align*}

We define
\[
\sst_i = (\ssto, \smem_i, \sps_i, \spc')
\]

We have, by definition of compositionality of symbolic states that
\begin{align*}
\mhyp{15}\quad & \sst_i \cdot \sst_f = \sst \cdot (\emptyset, \emptyset, \emptyset, \spc')
\end{align*}

In addition, from \hyp{4}, the fact that $\spc' \implies \spc$ and the fact that $\sint(\spc') = \true$ (from \hyp{9}), we also have that
\begin{align*}
\mhyp{16}\quad & \st \models \sst_i \cdot \sst_f
\end{align*}

And therefore from \hyp{4}, \hyp{9}, \hyp{16} and Lem.~\ref{lem:state-decompose}: $\exists \st_i.$
\begin{align*}
\mhyp{17}\quad & \st = \st_i \cdot \st_f \\
\mhyp{18}\quad & \sint, \st_i \models \sst_i
\end{align*}

Then, by \hyp{14} and \hyp{18}, we have that:
\begin{align*}
\mhyp{19}\quad & \sint(\ssubst'), \st_i \models P_i
\end{align*}

Then, by definition of state satisfaction for a predicate, and definition of satisfaction of a disjunction, we can derive that 
\begin{align*}
\mhyp{20}\quad & \sint, \st_i \models \sst_p = (\ssto, \emptyset, \{ \pred(\vec{\sval})(\ssubst'(\predout)) \}, \spc')
\end{align*}

Finally, from \hyp{9}, \hyp{11}, \hyp{17}, \hyp{20} and the definition of satisfaction, we get that
\[
\sint, \st \models \sst'
\]

\subparagraph{Unfolding predicates Case: C = $\punfold{\pred}{\vec{\pexp}}$.}

Remember the Unfold rule:
\begin{mathpar}
\mprset{flushleft}
\inferrule[\textsc{Unfold}]{
 \pred(\predin)(\predout)~\{\bigvee_i (\exists \vec{x}_i.~P_i)\}\in \preds \\\\
 \cseeval{\vec{\pexp}}{\ssto}{\spc}{\vec{\sval}}{\spc} \quad 
 \conspred( \pred, \vec{\sval},\sst)\rightsquigarrow ( \spredout,\sst') \\\\
 \vec{\sym{z}} \text{ fresh} \quad \ssubst \defeq \{ \vec{\sval}/\predin,  \spredout/\predout,  \vec{\sym{z}}/\vec{x}_i\} \\\\
 \produce(P_i,\ssubst,\sst') \rightsquigarrow \sst''}
 {\csesemtransabstractm{\sst}{\punfold{\pred}{\vec{\pexp}}}{\sst''}{\fsctx}{\macOX}{\osucc}}
\end{mathpar}

\begin{description}
\item[\hyp{1}] $\models (\gamma,\Gamma)$
\item[\hyp{2}] $\csesemtranscollectm{\sst}{\cmd}{\hat\Sigma'}{\fsctx}{\macEX}$
\item[\hyp{3}] $\oxabort \not\in \hat\Sigma'$
\item[\hyp{4}] $\sint, \st \models \sst \land \st,\cmd\baction_{\fictx} \result : \st'$
\end{description}

\textbf{To prove}: \[
\exists \sst'', \sint' \geq \sint.~(\result, \sst') \in \Sigma' \land \sint', \st' \models \sst'
\]

First of all, from the fact that $\texttt{fold}$ is a ghost command, we know
\[
\mhyp{5}\quad \st' = \st \land \result = \oxok
\]

\begin{align*}
\mbox{\textbf{Let}}~&(\ssto, \smem, \sps, \spc) = \sst
\mbox{\textbf{Let}}~&(\sto, \hp) = \st
\end{align*}

Since execution did not abort \hyp{3}, given the two rules of $\conspred$, in OX mode, we know know
\begin{align*}
\mhyp{6}\quad & \sps = \{\Predd{\spredini}{\spredouti} \mid i \in I \}\uplus\spred' \quad \pred\notin \spred'\\
\mhyp{7}\quad & \exists i.\ \sat(\spc \land \vec{\sval} = \spredini)\\
\mhyp{8}\quad & \unsat(\spc \wedge \neg(\bigwedge_{i \in I}\vec{\sval} = \spredini))\\
\end{align*}

From \hyp{4}, \hyp{7}, \hyp{8}, we learn that, by choosing an adapted $i$, we have: $\exists \spredin, \spredout, \sps'.$
\begin{align*}
\mhyp{9}\quad & \sps = \{ \Predd{\spredin}{\spredout} \} \cup \sps'\\
\mhyp{10}\quad & \spc' = (\spc \land \vec{\sval} = \spredin)\\
\mhyp{11}\quad & \sint(\spc') = \true\\
\mhyp{12}\quad & \conspred( \pred, \vec{\sval}, \sst) \rightsquigarrow (\spredout, (\ssto, \smem, \sps', \spc'))
\end{align*}

From \hyp{11} and \hyp{4}, we know 
\begin{align*}
\mhyp{13}\quad & \sint, \st \models (\ssto, \smem, \sps, \spc')
\end{align*}

Moreover, by definition of composition on symbolic states, and satisfaction of symbolic states, we learn: $\exists \hp_f, \hp_p.$
\begin{align*}
\mhyp{14}\quad & (\ssto, \smem, \sps, \spc') = (\ssto, \smem, \sps', \spc') \cdot (\emptyset, \emptyset, \{ \Predd{\spredin}{\spredout} \}, \spc')\\
\mhyp{15}\quad & \sint, (\sto, \hp_f) \models (\ssto, \smem, \sps', \spc')\\
\mhyp{16}\quad & \sint, (\emptyset, \hp_p) \models (\emptyset, \emptyset, \{ \Predd{\spredin}{\spredout} \}, \spc') \\
\mhyp{17}\quad & \hp = \hp_f \uplus \hp_p \quad (\text{which implies } \dish{\hp_f}{\hp_p}) \
\end{align*}

\[
\begin{array}{ll@{\hspace{\tabcolsep}=\hspace{\tabcolsep}}ll}
\mbox{\textbf{Let}}~&\ssubst_{\text{init}} & \{ \vec{\sval}/\predin,  \spredout/\predout\} &  \\
\mbox{\textbf{Let}}~&\ssubst & \ssubst_{\text{init}}\uplus \{ \vec{\sym{z}}/\vec{x}_i\} & \vec{\sym{z}} \text{ fresh}
\end{array}
\]

From \hyp{11}, \hyp{16} and the definition of state satisfaction, we obtain:
\[
\begin{array}{lll}
 & \subst_{\text{init}}, (\sto, \hp_p) \models \bigvee\limits_i (\exists \vec{x}_i.~P_i) & \implies\\
 & \exists i, \subst_{\text{init}}, (\sto, \hp_p) \models \exists \vec{x}_i.~P_i & \implies\\
 & \exists i, \vec{v}.\ \subst_{\text{init}}\uplus \{ \vec{v}/\vec{x}_i\}, (\sto, \hp_p) \models P_i & \implies \\
\mhyp{18} & \exists i, \sint'.~ \sint \leq \sint' \land \sint'(\ssubst), (\sto, \hp_p) \models P_i & \text{ where } \sint' = \sint\uplus\{\vec{v}/\vec{\sym{z}} \}
 \end{array}
\]

Let us select such an $i$ and $\sint'$. From completeness of produce (Prop.~\ref{prop:prod-compl}), \hyp{15}, \hyp{18} and \hyp{17}, and the fact that produce does not modify the store (Prod.~\ref{prop:wf}), and the definition of symbolic state satisfaction: we have $\exists \smem_p, \sps_p, \spc''.$
\[\begin{array}{ll}
 \mhyp{19} & \produce(P, \ssubst, (\ssto, \smem, \sps', \spc')) \rightsquigarrow (\ssto, \smem \uplus \smem_p, \sps' \cup \sps_p, \spc'')\\
 \mhyp{20} & \sint', (\emptyset, \hp_p) \models (\emptyset, \smem_p, \sps_p, \spc'')\\
\mhyp{21} & \sint'(\spc'') = \true \\
\end{array}\]

From \hyp{15}, \hyp{21} and $\sint' \geq \sint$, we get:
\[\begin{array}{ll}
\mhyp{22} & \sint', (\sto, \hp_f) \models (\ssto, \smem, \sps', \spc'')
\end{array}\]

And therefore from \hyp{20} and \hyp{22}:
\[
\begin{array}{ll}
  \mhyp{23} & \sint', (\sto, \hp) \models (\ssto, \smem \uplus \smem_p, \sps' \cup \sps_p, \spc'')\\ & \\
  \mbox{\textbf{Let}} & \sst'' = (\ssto, \smem \uplus \smem_p, \sps' \cup \sps_p, \spc'')
\end{array}
\]

Putting it all together.

\[\begin{array}{l@{\hspace{\tabcolsep}\text{from }}l}
\sst, \punfold{\pred}{\vec{\pexp}} \baction_{\fsctx}^{\macOX} \oxok: (\ssto, \smem \uplus \smem_p, \sps' \cup \sps_p, \spc'') & \mhyp{9}, \mhyp{10}, \mhyp{11}, \mhyp{12} \text{ and } \mhyp{19}\\
\sint', \st' \models \sst'' & \mhyp{5} \text{ and } \mhyp{23}
\end{array}\]

\qed

%% file: sections/app-macproduce.tex
\section{Consume and Produce Implementation: Definitions and Correctness}\label{app:mac-produce-soundness}

This appendix provides the proofs of correctness for our $\mac$ and $\produce$ implementations. We start by providing the exhaustive definitions, algorithms, and rules of our implementation and then give the proofs in question.

\subsection{Rules for \mac}\label{app:rules-consume}

In this section we present the complete set of  rules for \mac: auxiliary rules in Fig.~\ref{fig:mac-additional}, the successful rules in Fig.~\ref{fig:mac-success}, and the error rules in Fig.~\ref{fig:mac-error}.

\begin{figure}[h]
\footnotesize
    \begin{mathpar}
    \mprset{flushleft}
        \inferrule{\smem=\smem_f \uplus \{\sval_1 \mapsto \sval_2\}\\\\ 
          \spc'=\spc\wedge (\sval = \sval_1) \\ \sat(\spc')}
        {\conscell(\sval,\sst)
        \rightsquigarrow ({\sval_2}, \sst[\sstupdate{heap}{\smem_f}, \sstupdate{pc}{\spc'}])}
\and
        \inferrule{\sat(\spc \wedge v\notin \dom(\smem))}
         {\conscell(\sval,\sst) \rightsquigarrow abort}
    \end{mathpar}
    \hdashrule{130mm}{1pt}{1pt} 
    \begin{mathpar}
%
\inferrule
{\spred = \{ \Predd{\spredin}{\spredout} \} \cup \spred' \\
\spc' \defeq \spc \wedge \vec{\sval}= \spredin \qquad \sat(\spc')}
{\conspred( \pred,\vec{\sval},\sst)\rightsquigarrow (\spredout, \sst[\sstupdate{preds}{\sps'}, \sstupdate{pc}{\spc'}])}
\and
%
\inferrule
{\spred= \{\Predd{\spredini}{\spredouti} \mid i \in I \}\uplus\spred' \\ \pred\notin \spred' \\
\sat(\spc \wedge \neg(\bigwedge_{i \in I}\vec{\sval}=  \spredini)) }
{\conspred( \pred,\vec{\sval},\sst)\rightsquigarrow abort}
\end{mathpar}
\caption{Rules for $\conscell$ and \(\conspred\)}\label{fig:mac-additional}
\end{figure}

\begin{figure}[h]
\begin{mathpar}
\footnotesize
\mprset{flushleft}
\inferrule*[left=(Inductive)]{
 \macMP(m, [(P,outs)], \ssubst, \sst)\rightsquigarrow (\ssubst', \sst')\\\\
 \macMP(m, mp, \ssubst', \sst')\rightsquigarrow {O}
}
{\macMP(m, (P,outs)::mp, \ssubst,  \sst)\rightsquigarrow {O}}
\and 
\inferrule*[left=(Pure)]{
 P \text{ is pure } \\ outs=[(\lvar{x}_i, \lexp_i)|_{i=1}^n]\\\\
 \ssubst'=\ssubst \uplus \{(\ssubst (\lexp_i)/\lvar{x}_i )|_{i=1}^n\} \\
 \consPure(m,\spc,\ssubst'(P)) = \spc' }
 {\macMP( m, [(P, outs)], \ssubst ,\sst)\rightsquigarrow ( \ssubst', \sst[\sstupdate{pc}{\spc'}])) }
 \and
 \inferrule*[left=(CellNeg)]
 {\funcAlg{conscell}{\ssubst(\lexp_1),\sst } =(\cfreed, \sst')}
 {\macMP(m, (\lexp_1\mapsto \cfreed, []), \ssubst, \sst ) 
       \rightsquigarrow  ( \ssubst', \sst')}
  \and
 \inferrule*[left=(CellPos)]
 {\consPure(m, \spc, \ssubst(\lexp_a)\in \Val) =\spc' \\
\conscell( \ssubst(\lexp_a), \sst[\sstupdate{pc}{\spc'}]) \rightsquigarrow  (
       \sval,\sst')\\\\
       \ssubst_{subst}=\{\sval/\mathsf{O}\} \\
       outs=[(\lvar{x}_i, \lexp_i)|_{i=1}^n]  \\
     ((\ssubst\uplus \ssubst_{subst})(E_i)=\hat v_i)|_{i=1}^n
   \\
      \ssubst'=\ssubst \uplus \{( \hat v_i/x_i)|_{i=1}^n\} \\
        \consPure(m, (\sst').\fieldsst{pc}, \ssubst' (\lexp_v)=\sval)= \spc'''}
 {\macMP(m, [(\lexp_a\mapsto \lexp_v, outs)],  \ssubst,  \sst) 
       \rightsquigarrow  (\ssubst', \sst'[ \sstupdate{pc}{\spc'''}])}
\and
\inferrule*[left=(ConsPred)]
{ \consPure(m, \spc,\ssubst(\vec{\lexp}_{ins}) \subseteq \Val) = \spc'\\
   \conspred( \pred,\ssubst(\vec{\lexp}_{ins}), \sst[\sstupdate{pc}{\spc'}])\rightsquigarrow (\vec{\sval}_{outs},\sst') \\
    outs = [(x_i,\lexp_i)|_{i=1}^n]\\
    \vec{\sval}_{outs}=[\sval_{o_1},\dots, \sval_{o_k}]\\
    \ssubst_{subst}=\{(\sval_{o_i}/{\sf O}_i)|_{i=1}^k\} \\
    ((\ssubst\uplus \ssubst_{subst})(E_i)=\hat v_i)|_{i=1}^n \\
    \ssubst'=\ssubst \uplus \{(\hat v_i/x_i)|_{i=1}^n\}  \\  \consPure(m, (\sst').\fieldsst{pc}, \ssubst' (\vec{\lexp}_{outs})=\vec{\sval}_{outs})= \spc''' 
}
{\macMP(m,[(\pred(\vec{\lexp}_{ins})(\vec{\lexp}_{outs}),outs)], \ssubst,\sst)\rightsquigarrow (\ssubst',\sst'[\sstupdate{pc}{\spc'''}])}
\end{mathpar}

\caption{Rules \macMP~-- successful cases}\label{fig:mac-success}
\end{figure}

\begin{figure}[!t]
\footnotesize   
\begin{mathpar}
\mprset{flushleft}
\inferrule*[left=(Pure-fail)]{P \text{ is pure }\\
 outs=[(\lvar{x}_i, \lexp_i)|_{i=1}^n]\\\\
 \ssubst'=\ssubst \uplus \{(\lvar{x}_i \mapsto \ssubst (\lexp_i))|_{i=1}^n\} \\
   \consPure(\macEX,\spc, \ssubst'(P))= abort }
   {\macMP(\macEX, [(P, outs)], \ssubst, \sst) \rightsquigarrow abort([\mathstr{\sf consPure}, \ssubst'(P), \spc])}
\and
\inferrule*[left=(Points2Error1)]
{\conscell( \ssubst(\lexp_a),\sst) =abort}
{\macMP(m, (\lexp_a\mapsto \lexp_v, outs),  \ssubst,\sst) \rightsquigarrow  abort([\mathstr{\sf MissingCell}, \ssubst(\lexp_a), \spc \land \ssubst(\lexp_a) \not\in \dom(\smem)])}
\and
\inferrule*[left=(Points2Error2)]
{\conscell(\ssubst(\lexp_a),\sst) =abort}
{\macMP(m, (\lexp_a\mapsto \cfreed, outs),  \ssubst,\sst) \rightsquigarrow  abort([\mathstr{\sf MissingNegCell}, \ssubst(\lexp_a), \spc \land \ssubst(\lexp_a) \not\in \dom(\smem)])}
\and
\mprset{flushleft}
\inferrule*[left=(CellError2)]
{\conscell(\ssubst(\lexp_a),\sst ) = (
        \sval, \sst')\\\\
\ssubst_{subst}=\{\mathsf{O}\mapsto \sval\}  \\
        outs=[(\lvar{x}_i, \lexp_i)|_{i=1}^n] \\\\
          ((\ssubst\uplus \ssubst_{subst})(E_i)=\sym{v}_i)|_{i=1}^n\\
        \ssubst'=\ssubst \uplus \{(\lvar{x}_i \mapsto \ssubst \lexp_i)|_{i=1}^n\} \\\\
         \consPure(\macEX,(\sst').\fieldsst{pc},\ssubst' (\lexp_v)=\sval))= abort
        }
 {\macMP(\macEX, (\lexp_a\mapsto \lexp_v, outs), \ssubst,\sst) \rightsquigarrow abort([\mathstr{\sf consPure}, \spc])}
\and
\inferrule*[left=(CellError3)]{\conscell(\ssubst(\lexp_a),\sst) = (
       \cfreed,\sst')}
{\macMP(\macEX, (\lexp_a\mapsto \lexp_v, outs), \ssubst, \sst) \rightsquigarrow abort([\mathstr{\sf consError}, \spc])}
\and 
\inferrule*[left=(CellError4)]{\conscell(\ssubst(\lexp_a),\sst) = (
       \sval,\sst')}
       {\macMP(\macEX, (\lexp_a \mapsto \cfreed, outs), \ssubst, \sst) \rightsquigarrow abort([\mathstr{\sf consError}, \spc])}
\and
\mprset{flushleft}
\inferrule*[left=(PropagateError)]{
 \macMP(m, [(P,outs)], \ssubst, \sst)\rightsquigarrow abort(\sval)}
{\macMP(m, (P,outs)::mp, \ssubst,  \sst)\rightsquigarrow abort(\sval)}
\and
\inferrule*[left=(PredError1)]{  \conspred( \pred,\ssubst(\vec{\lexp}_{ins}),\sst)\rightsquigarrow abort}
       {\macMP(m,[(\pred(\vec{\lexp}_{ins})(\vec{\lexp}_{outs}),outs)], \ssubst,\sst)\rightsquigarrow abort([\mathstr{\sf Pred}, \ssubst(\vec{\lexp}_{ins}), \spc] )}
       \and
  \inferrule*[left=(PredError2)]{ \consPure(\macEX, \spc,\ssubst(\vec{\lexp}_{ins}) \subseteq \Val) = abort}
       {\macMP(\macOX,[(\pred(\vec{\lexp}_{ins})(\vec{\lexp}_{outs}),outs)], \ssubst,\sst)\rightsquigarrow abort([\mathstr{\sf consError}, \spc] )}
       \end{mathpar}

\begin{prooftree}
    \AxiomC{\(
    \begin{array}{l@{\hspace{.5cm}}l}
    \consPure(m, \spc,\ssubst(\vec{\lexp}_{ins}) \subseteq \Val) = \spc'&
    \conspred( \pred,\ssubst(\vec{\lexp}_{ins}),]\sst)\rightsquigarrow ( \vec{\sval}_{outs},\sst') \\
    outs = [(x_i,\lexp_i)|_{i=1}^n]&
    \vec{\sval}_{outs}=[\sval_{o_1},\dots, \sval_{o_k}]\\
    \ssubst_{subst}=\{({\sf O}_i\mapsto \sval_{o_i})|_{i=1}^k\} & 
    ((\ssubst\uplus \ssubst_{subst})(E_i)=\hat v_i)|_{i=1}^n \\
    \ssubst'=\ssubst \uplus \{(x_i\mapsto \hat v_i)|_{i=1}^n\}&   \consPure(m, (\sst').\fieldsst{pc}, \ssubst' (\vec{\lexp}_{outs})=\vec{\sval}_{outs})= abort 
   \end{array}
    \)
    }
    \UnaryInfC{\(\macMP(\macEX,[(\Pred{p}{\vec{\lexp}_{ins}}{\vec{\lexp}_{outs}},outs)], \ssubst,\sst)\rightsquigarrow abort([\mathstr{\sf consError}, \spc])\)}
\end{prooftree}

\caption{Rules \macMP~-- error rules, where $\sst = (\ssto, \smem, \spred, \spc)$ }\label{fig:mac-error}
\end{figure}

\begin{remark}
Rules [CellPos] and [Pure] can also be applied when $outs=[]$. In such cases, $\ssubst=\ssubst'$, i.e., the initial substitution already covers the assertion and does not need to be extended.
\end{remark}

\subsection{Correctness of \mac}\label{sec:macMP-correct}

\setcounter{theorem}{0}
Suppose 
\[
\begin{array}{r@{~\rightsquigarrow~}l}
\mac(m, P, \ssubst, (\ssto,\smem, \sps, \spc)) & (\ssubst', (\ssto', \smem_f, \sps_f, \spc')) 
\end{array}
\]

\begin{theorem}[Property~\ref{prop:wf}]
$\sinv((\ssto, \smem_f,\sps_f, \spc'))$  and $\sinv(\ssubst',\sps')$ hold.
\end{theorem}

\begin{proof}
We assume the following:
\begin{description}
\item[(A1a)] $\sinv((\ssto, \smem, \sps, \spc))\defeq \sinv((\ssto,\smem,\spc)) \wedge \svs{\sps}\subseteq \svs{\spc}$, where $\sinv((\ssto,\smem,\spc))\defeq(\svs{\ssto} \cup \svs{\smem}
  \subseteq \svs{\spc}) \land \sat(\spc) \land \spc \models (\sinvc(\ssto) \land \sinvc(\smem))$
\item[(A1b)] $  \sinv(\ssubst,\spc)\defeq \svs{\ssubst} \subseteq \svs{\spc} \wedge \spc \models \sinvc(\ssubst)$
\item[(A1c)] $\lv{P}\subseteq  \dom(\ssubst)$
\end{description}

The proof now follows by induction on $P$:
\begin{description}
\item[(\(P\) is pure)] 
In this case,   $\plan(P,\dom(\ssubst))= (P,outs)$.
 
 The interesting case is for $outs(P)\neq []$.

\begin{mathpar}
\inferrule*[left=(Pure)]{
 P \text{ is pure } \\ outs=[(\lvar{x}_i, \lexp_i)|_{i=1}^n]\\\\
 \ssubst'=\ssubst \uplus \{(\lvar{x}_i \mapsto \ssubst (\lexp_i))|_{i=1}^n\} \\
 \consPure(m,\spc,\ssubst'(P)) = \spc' }
 {\macMP( m, [(P, outs)], \ssubst ,\sst)\rightsquigarrow ( \ssubst', \sst[\sstupdate{pc}{\spc'}])) }
\end{mathpar}
Thus, $\smem_f=\smem$  and $\sps_f=\sps$. Depending on the mode we have:
\begin{itemize}
\item if $m=\macUX$ then $\spc'=\spc\wedge \ssubst'(P)$ and $\sat(\spc')$.
\item if $m=\macEX$ then $\spc=\spc'$ and $\spc \models \ssubst'(P)$
\end{itemize}
Notice that 

\hyp{1}  the resulting state is well-formed. i.e., $\sinv((\ssto, \smem, \sps, \spc'))$. It follows directly from the initial assumption  $\sinv((\ssto, \smem, \sps, \spc))$ and the shape of $\spc'$.

%

\hyp{2} \(\svs{\ssubst'}= \svs{\ssubst}\cup \svs{\ssubst(\lexp_i)}\subseteq \svs{\spc} \subseteq \svs{\spc'}\),  from (A1b).

\hyp{3} $\spc' \models \sinvc(\ssubst')$.

In fact, for each $\sint$ s.t. $\sint(\spc')=\true$, it follows that $\sint(\spc)=\true$, which implies, by (A1c) that   $\spc' \models \sinvc(\ssubst)$ holds. Since $\ssubst'=\ssubst \uplus \{(x_i\mapsto \ssubst(\lexp_i))|_{i=1}^n\}$, it follows that $\ssubst(\lexp_i)\in \Val$. Thus, \(\sinvc( \ssubst')\) holds.

The result follows from \hyp{1}, \hyp{2} and \hyp{3}.

\item[(\(P\) is spatial)] \(P=\lexp_a\mapsto \lexp_v\)  (The case \(P=\lexp_a\mapsto \cfreed\)  is similar.)

In this case,   $\plan(P,\dom(\ssubst))= (P,outs)$. 

  The interesting case is for $outs(P)\neq []$.

\begin{mathpar}
 \inferrule*[left=(CellPos)]
 {\consPure(m, \spc, \ssubst(\lexp_a)\in \Val) =\spc' \\
\conscell(\ssubst(\lexp_a), \sst[\sstupdate{pc}{\spc'}]) \rightsquigarrow  (
       \sval,\sst')\\
       \ssubst_{subst}=\{\mathsf{O}\mapsto \sval\} \\
       outs=[(\lvar{x}_i, \lexp_i)|_{i=1}^n]  \\
     ((\ssubst\uplus \ssubst_{subst})(E_i)=\hat v_i)|_{i=1}^n
   \\
      \ssubst'=\ssubst \uplus \{(x_i\mapsto \hat v_i)|_{i=1}^n\} \\
        \consPure(m, (\sst').\fieldsst{pc}, \ssubst' (\lexp_v)=\sval)= \spc''}
 {\macMP(m, [(\lexp_a\mapsto \lexp_v, outs)],  \ssubst,  \sst) 
       \rightsquigarrow  (\ssubst', \sst'[ \sstupdate{pc}{\spc''}])}
\end{mathpar}

If we analyse each mode separately:
\begin{itemize}
\item if $m=\macUX$ then \conscell gives  $\smem=\smem_f\cdot \{\sym{v}_a\mapsto \sym{v}\}$,  $\spc'=\spc\wedge \ssubst'(\lexp_a)=\sym{v}_a$ and $\sat(\spc')$. Besides,  \(\consPure(\macUX,(\sst').\fieldsst{pc}, \ssubst' (\lexp_v)=\sym{v})= \spc''\) gives 
 $\spc''=\spc'\wedge\ssubst' (\lexp_v)=\sym{v}$ and $\sat(\spc'')$.
\item if $m=\macEX$ then \conscell gives  $\smem=\smem_f\cdot \{\sym{v}_a\mapsto \sym{v}\}$,  $\spc'=\spc$ and $\spc \models \ssubst'(\lexp_a)=\sym{v}$. Besides, \( \consPure(\macEX,(\sst').\fieldsst{pc}, \ssubst' (\lexp_v)=\sym{v})=\spc''\) gives $\spc''=\spc$ and $\spc \models \ssubst' (\lexp_v)=\sym{v}$.
\end{itemize}
\hyp{1}  the resulting state is well-formed, i.e., \(\sinv((\ssto, \smem_f, \sps,\spc''))\). This follows directly from the construction above.

%
%

\hyp{2} \(\svs{\ssubst'}\subseteq  \svs{\spc''}\). 

 In fact, by definition, \(\svs{\ssubst'}= \svs{\ssubst}\cup \bigcup_i \svs{\sym{v}_i}\). From the construction above $\ssubst'(\lexp_v)=\sym{v}$ and $\svs{\sym{v}}\subseteq \svs{\smem}$. For \macUX:  first,  \(\svs{\sym{v}} \subseteq \svs{\spc'}\subseteq \svs{\spc''}\);  second, from (A1b), $\svs{\ssubst}\subseteq \svs{\spc}$. For \macEX: first $\spc=\spc''$; second, \(\svs{\ssubst'}= \svs{\ssubst}\cup \svs{\sym{v}}\subseteq \svs{\ssubst}\cup \svs{\smem} \subseteq \svs{\spc}\).

\hyp{3} $\spc'' \models \sinvc(\ssubst')$.

For \macUX: In fact, for each $\sint$ s.t. $\sint(\spc'')=\true$, it follows that $\sint(\spc)=\true$, which implies, by (A1c) that $\spc \models \sinvc(\ssubst)$ holds. Since $\sat(\spc'')$,  $\spc''=\spc'\wedge\ssubst' (\lexp_v)=\sym{v}$ and \(\sinvc(\smem)\) we get $\sym{v} \in \Val$.  Therefore, \(\sinvc( \ssubst')\) holds. For \macEX: the analysis is similar.

The result follows from \hyp{1}, \hyp{2} and \hyp{3}.
\item[(Predicates)] $P=\pred(\vec{\lexp_1})(\vec{\lexp_2})$ this case is analogous to the previous.
\item[ (Inductive Step)] \(P=P_1\lstar P_2\)

This case follows by the induction hypothesis.
\end{description}

\end{proof}

\begin{theorem}[Property~\ref{prop:path_strength}]
$\spc' \implies \spc$ holds.
\end{theorem}

\begin{proof}
   The analysis depends on the \(mode\) used.
   \begin{itemize}
    \item  \(mode=\macUX\)
    
    By inspection of the rules for $\mac$ it is easy to see that
      $\spc'=\spc \wedge \spc''$,
     for some $\spc''$ such that $\sat(\spc')$, which trivially gives $\spc' \implies \spc$.
     \item \(mode=\macEX\)
     
     This case is trivial since when the path condition is modified in \(\macEX\) it is of the form  $\spc'=\spc\wedge \spc''$, for some $\spc''$.
     \end{itemize}
\end{proof}

%
\begin{theorem}[Property~\ref{prop:coverage}] $ \ssubst' \geq \ssubst \gand \dom(\ssubst') \supseteq \lv{P}$.

\end{theorem}

\begin{proof}
By hypothesis, $\dom(\ssubst)=\{\vec{\lvar{x}}\}\subseteq \lv{\vec{\pvar{x}} = \vec{x} \lstar P}$. 
        By inspecting the rules, it is easy to check that $\ssubst'$  extends $\ssubst$. 
        
        It remains to  prove that 
        \(\lv{P}\subseteq \dom(\ssubst')\). Suppose that there exists $y\in\lv{P}$ such that 
        $y\notin dom(\ssubst')$. Then, clearly $y\notin \vec{\lvar{x}}$ and there exists a 
        simple assertion $P_i$ in the composition of $P$ such that $y\in \lv{P_i}$. 
        There are two cases to consider:
        \begin{enumerate}
            \item $P_i$ is a pure assertion. 
            
               The interesting case is for equality $E_1=E_2$ and \(y\in \lv{E_2}\). The case for \(y\in \lv{E_1}\)
               can be verified analogously.
               Since $\plan(\knowb, P)=mp$, there exists some ordering of the simple assertions in
               $P$ such that  at least one of the $E_i$'s is known and 
               $\getinsouts(E_1=E_2)=(\lv{E_1,E_2}\setminus \lv{outs},outs)$. 
               Therefore, $ins=\lv{E_1,E_2}\setminus \lv{outs}\subseteq \knowb$
               and \((E_1=E_2, outs)\) is in $mp$. 
               \begin{itemize}
                   \item If $outs=[]$ then $y\in ins=\lv{E_1,E_2} \subseteq \knowb$, 
                   which implies that $y$ was added in $\knowb$ in a previous point of the search. 
                   Then, there exists a simple assertion $P_j$ that occurs before $E_1=E_2$ such that $(P_j, [(y, \lexp_j)])$ is in $mp$. In that case, since 
                   $\mac$ is successful, $[y\mapsto E_j]$ was added to the substitution 
                   \(\ssubst\) after one application of rule [Pure], i.e, $y\in \dom(\ssubst')$. Contradiction.
                \item If $outs\neq []$ and $y\notin ins=\lv{E_1,E_2} \subseteq \knowb$, then 
                  \((y, \lexp')\) is in $outs$ for some $E'$. With an application
                  of [Pure] in \(\mac\), which gives a contradiction as above.
               \end{itemize}
            \item \(P_i\) is a spatial assertion.
            
              The interesting case is for $P_i$ of the form $\lexp_a\mapsto E_v$ and $y\in \lv{\lexp_a}$. The case for $y\in \lv{\lexp_v}$ 
              can be  verified analogously.
            
              In the computation  of the $\mps$  $mp$ for $P$ (by hypothesis it exists) we compute $(ins,outs)=\getinsouts(\knowb',\lexp_a\mapsto \lexp_v)$
              where $\knowb_i=\knowb\cup V$, and $V$ is a possibly empty set of logical variables 
              that were added in a previous point in the computation of $mp$. 
        
              Thus, there must exist an ordering of the simple assertions  of $P$ such that $\lv{\lexp_a}\in\knowb_i$, 
              in which case  one has  $(ins,outs)=(\lv{\lexp_a,\lexp_v}\setminus {\rlvar}(outs),outs)$ and 
              $outs=\learn(\knowb_i,\mathsf{O},\lexp_v)$. 
              \begin{itemize}
                 \item If $outs=[]$ and it is the case that $\lv{\lexp_v}\subseteq \knowb'$, then
                 $y\in \lv{E_2} \subseteq \knowb_i$. This implies that there exists a simple 
                 assertion $P_{j}$ and $\knowb_{j}$ (for some $i>j$) such that $(P_j,outs_j)$ is in $mp$ and $(y, \lexp_j)$ is in 
                 $outs_j$. 
                 \item Otherwise, $[]\neq outs=[(y, \lexp_v')]$
                    for some $\lexp_v'$ obtained after manipulating $\lexp_v$. 
            \end{itemize}
          In both cases, since \mac
          is successful, $\ssubst$ will be extended with $[y\mapsto \_]$, where $\_$ denotes
          either $E_v'$ or $E_j$. Contradiction.
     \end{enumerate}
 \end{proof}
\begin{lemma}\label{lem:esem_comm} 
Let $P$ be a pure assertion. Then,
\(\esem{P}{\sint(\ssubst)}= \sint(\ssubst(P))\).
\end{lemma}

\begin{proof}
By induction on the structure of the pure assertion $P$.
\end{proof}

\begin{theorem}[Property~\ref{prop:soundness}] Soundness of \mac holds.
\hfill
\[
\exists \smem_P, \sps_P.~\smem = \smem_P \uplus \smem_f \land \sps = \sps_P \cup \sps_f \land (\forall \vint, \st.~\sint, \st \smodels (\ssto, \smem_p, \sps_p, \spc') \implies \sint(\sym\theta'), \st \models P).
\]
%
\end{theorem}
\begin{proof}
Suppose \(\mac(m, P, \ssubst, (\ssto,\smem,\sps, \spc)) \rightsquigarrow (\ssubst', (\ssto, \smem_f, \sps_f, \spc')) \). 
The proof is by induction on the structure of $P$.
\begin{description}
\item[($P$ is pure)] and not an equality.

In this case $outs=[]$ and the rule [Pure] from Fig.~\ref{fig:mac-success} is instantiated as:

\begin{mathpar}
    \inferrule*[left=(Pure)]{P \text{ is pure } \qquad  \consPure(m,\spc,\ssubst(P))= \spc'}
    {\macMP( m, (P, []), \ssubst,\sst)\rightsquigarrow ( \ssubst, \sst[\sstupdate{pc}{\spc')}]}
\end{mathpar}
Then, $\smem_f=\smem$ and $\sps_f=\sps$ which gives $\smem_P=\emptyset$
 and $\sps_P=\emptyset$.

The rest of the analysis relies on the mode used:
\begin{itemize}
\item if \(m=\macUX\) then $\spc'=\spc \wedge \ssubst(P)$ and \(\sat(\spc')\).

Let $\sint$ and $\st=(\sto, \hp)$ such that $\sint,\st\smodels (\ssto,\emptyset,\emptyset,\spc')$. Then, $\sint(\ssto)=\sto$, $\sint(\smem)=\hp=\emptyset$ and   $\sint(\spc')=\sint(\spc \wedge \ssubst(P))=\true$. By definition, the following holds:
\begin{align*}
\sint(\ssubst), (\sto, \emptyset) \models P & \iff \sint(\ssubst), (\sto, \emptyset)\models P \\
& \iff \esem{P}{\sint(\ssubst)}{}=\true \wedge \emp\\
& \iff \sint(\ssubst)(P)=\true \wedge \emp ~\text{ (Lemma~\ref{lem:esem_comm})}
\end{align*}
which trivially follows from the hypothesis that \(\sint(\spc \wedge \ssubst(P))=\true\).
\item if \(m=\macEX\) then $\spc'=\spc$ and \(\spc \models \ssubst(P)\). 

Let $\sint$ and $\st=(\sto,\hp)$ be  such that $\sint,\st\smodels (\ssto,\emptyset,\emptyset,\spc)$.  From \(\spc \models \ssubst(P)\)  in \macEX, we have $\sint(\ssubst(P))=\true$ and the result follows similarly to the case above.
\end{itemize}
\item[($P$ is an equality)] $P=(\lexp_1= \lexp_2)$

The interesting case is when $outs\neq []$ (otherwise, it is similar to the case above).

\begin{mathpar}
\inferrule*[left=(Pure)]{
 P \text{ is pure } \\ outs=[(\lvar{x}_i, \lexp_i)|_{i=1}^n]\\\\
 \ssubst'=\ssubst \uplus \{(\lvar{x}_i \mapsto \ssubst (\lexp_i))|_{i=1}^n\} \\
 \consPure(m,\spc,\ssubst'(P)) = \spc' }
 {\macMP( m, [(P, outs)], \ssubst ,\sst)\rightsquigarrow ( \ssubst', \sst[\sstupdate{pc}{\spc'}])) }
\end{mathpar}

Then, \(\smem_f=\smem\) and $\smem_P=\emptyset$. 

Let $\sint$ be such that and $\st=(\sto,\hp)$ be  such that $\sint,\st\smodels (\ssto,\emptyset,\emptyset,\spc)$.
The rest of the analysis relies on the mode used and is similar to the previous case.

\item[($P$ is an spatial assertion)]  $P=\lexp_a\mapsto \lexp_v$.

Then,  the following rule  from Fig.~\ref{fig:mac-success} was applied 

\begin{mathpar}
 \inferrule*[left=(CellPos)]
 {\consPure(m, \spc, \ssubst(\lexp_a)\in \Val) =\spc' \quad
\conscell( \ssubst(\lexp_a), \sst[\sstupdate{pc}{\spc'}]) \rightsquigarrow  (
       \sval,\sst')\\\\
       \ssubst_{subst}=\{\mathsf{O}\mapsto \sval\} \\
       outs=[(\lvar{x}_i, \lexp_i)|_{i=1}^n]  \\
     ((\ssubst\uplus \ssubst_{subst})(E_i)=\hat v_i)|_{i=1}^n
   \\
      \ssubst'=\ssubst \uplus \{(x_i\mapsto \hat v_i)|_{i=1}^n\} \\
        \consPure(m, (\sst').\fieldsst{pc}, \ssubst' (\lexp_v)=\sval)= \spc'''}
 {\macMP(m, [(\lexp_a\mapsto \lexp_v, outs)],  \ssubst,  \sst) 
       \rightsquigarrow  (\ssubst', \sst'[ \sstupdate{pc}{\spc'''}])}
  \end{mathpar}

where 

\begin{mathpar}
        \inferrule{\smem=\smem_f\cdot \{\sval_a \mapsto \sval\}\\ 
          \spc''=\spc\wedge (\ssubst(\lexp_a) = \sval_a) \\ \sat(\spc'')}
        {\conscell(\ssubst(\lexp_a),\sst[\sstupdate{pc}{\spc'}])
        \rightsquigarrow ({\sval}, \sst[\sstupdate{heap}{\smem_f}, \sstupdate{pc}{\spc''}])}
\end{mathpar}
Take  $\smem_P = \{\sval_a \mapsto \sval\}$, $\sps_P=\emptyset$ and $\sst_P=(\ssto, \smem_P,\emptyset,\spc'')$.
\begin{itemize}
\item if $m=\macUX$ then $\spc''=\spc'\wedge  ( \ssubst'( \lexp_a)=\sval_a)
 \wedge  ( \ssubst'( \lexp_v)=\sval)$ and $\sat(\spc'')$.

Let \(\sint\) and \(\st=(\sto,\hp)\)  be such that  $\sint, (\sto,\hp)\smodels \sst_P$.
Then, \(\sint(\ssto)=\sto\), \(\sint(\smem_P)=\hp\), and \(\sint(\spc'')=\true\), which implies
\hyp{1} \( \sint(\ssubst'(\lexp_a)=\sval_a)=\true \) and
  \(\sint(\ssubst' (\lexp_v)=\sval)=\true\). 
 By definition,

\begin{align*}
\sint(\sym\theta'), (\sto, \hp) \models \lexp_a\mapsto \lexp_v 
&\iff \sint(\sym\theta'), (\sto, \hp) \models \lexp_a\mapsto \lexp_v \\
&\iff \sint(\smem_p)=\{\esem{\lexp_a}{\sint(\ssubst')}\mapsto \esem{\lexp_v}{\sint(\ssubst')}\}, \text{ by \hyp{1}}
\end{align*}
and the result follows.
\item if \(mode=\macEX\) then the proof is analogous.
\end{itemize}
\item[(\(P\) is a predicate assertion)] \(P=\Pred{p}{\expredin}{\expredout}\)

Then, the following rule from Fig.~\ref{fig:mac-success} was applied

\begin{mathpar}
  \inferrule*[left=(ConsPred)]
{ \consPure(m, \spc,\ssubst(\expredin) \subseteq \Val) = \spc'\\ 
   \conspred( \pred,\ssubst(\expredin), \sst[\sstupdate{pc}{\spc'}])\rightsquigarrow (\vec{\sval}_{outs},\sst') \\\\
    outs = [(x_i,\lexp_i)|_{i=1}^n]\\
    \vec{\sval}_{outs}=[\sval_{o_1},\dots, \sval_{o_k}]\\
    \ssubst_{subst}=\{({\sf O}_i\mapsto \sval_{o_i})|_{i=1}^k\} \\
    ((\ssubst\uplus \ssubst_{subst})(E_i)=\hat v_i)|_{i=1}^n \\
    \ssubst'=\ssubst \uplus \{(x_i\mapsto \hat v_i)|_{i=1}^n\}  \\ 
     \consPure(m, (\sst').\fieldsst{pc}, \ssubst' (\expredout)=\vec{\sval}_{outs})= \spc''' 
}
{\macMP(m,[(\Pred{p}{\expredin}{\expredout},outs)], \ssubst,\sst)\rightsquigarrow (\ssubst',\sst'[\sstupdate{pc}{\spc'''}])}
\end{mathpar}
where 

\begin{mathpar}
  \inferrule
{\spred = \{ \Predd{\spredin}{\spredout} \} \cup \spred_f \\
\spc'' \defeq \spc' \wedge \ssubst(\expredin)= \spredin \\ \sat(\spc'')}
{\conspred(\pred,\ssubst(\expredin),\sst[\sstupdate{pc}{\spc'}])
\rightsquigarrow (\spredout, \sst[\sstupdate{preds}{\sps_f}, \sstupdate{pc}{\spc''}])}
\end{mathpar}
Then  $\smem_f=\smem$, since this case does not affect
 the symbolic heap, which gives $\smem_P=\emptyset$.  Also, $\sps=\sps_P\cup \sps_f$ where 
 $\sps_P=\{\Predd{\spredin}{\spredout} \}$.

 The rest of the proof depends on the mode. 
 \begin{itemize}
 \item \(m=\macUX\)
 
 By the construction
 above $\spc'''=\spc' \wedge \ssubst(\expredin)=\spredin \wedge \ssubst'(\expredout)=\spredout$.

 Let $\sint, \st=(\sto,\hp)$ be such that
  $\sint,\st \smodels (\emptyset,\emptyset,\sps_P,\spc''')$. Then, by 
  definition 
\[
\begin{aligned}
  \sint,\st \smodels (\emptyset,\emptyset,\sps_P,\spc''') & \iff 
  \sto=\emptyset \wedge \sint(\spc''')=\true \wedge \sint,(\emptyset, h)\models \sps_P\\
    & \iff 
    \sto=\emptyset \wedge \sint(\spc''')=\true \wedge \sint,(\emptyset, h)\models \Predd{\spredin}{\spredout}\\
\end{aligned}
\]
Now,  from the assumption that  $\Predd{\predin}{\predout}~\{ P' \} \in \preds $ 
and the definition of $\models$:
\[
  \begin{aligned}
    \sint(\ssubst'), \st \models \Pred{p}{\expredin}{\expredout} & \iff 
    \sint(\ssubst')[\predin \mapsto \esem{\expredin}{\sint(\ssubst'), \emptyset},
     \predout \mapsto \esem{\expredout}{\sint(\ssubst'), \emptyset}], (\emptyset, \hp) \models P' \\
      & \iff
      \sint(\ssubst')[\predin \mapsto \sint(\ssubst')(\expredin),
      \predout \mapsto \sint(\ssubst')(\expredout)], (\emptyset, \hp) \models P' \\ 
      & \iff
      \sint(\ssubst')[\predin \mapsto \spredin,
      \predout \mapsto \spredout], (\emptyset, \hp) \models P'
\end{aligned}
\]
and the result follows from $ \sint,(\emptyset, h)\models \Predd{\spredin}{\spredout}$
above.
 \item \(m=\macEX\)
 
 This case is similar.
\end{itemize}

\item[(\(P\) is a star-conjunction)] $P=P_1\lstar P_2\lstar \ldots \lstar P_n$ (for some $n$).

This case follows easily by induction hypothesis.

\end{description}
\end{proof}

\begin{theorem}[Property~\ref{prop:mac-branch-comp}]Completeness of \mac in \macEX mode:
\[
  \begin{array}{l}
  \oxabort \not\in \mac(\macEX, P, \ssubst, \sst) \land
  \vint, \st \smodels \sst \implies \\
  \qquad \exists \ssubst', \st_f, \sst_f.~\mac(\macEX, P, \ssubst, \sst) \rightsquigarrow (\ssubst', \sst_f) \land
  \vint, \st_f \smodels \sst_f
  \end{array}
  \]

%
%
\end{theorem}
\begin{proof}
Let $\sst=(\ssto,\smem, \sps,\spc)$ be a symbolic state  and \(\sint\) and 
$\st=(\sto, \hp)$ be such that $\sint, \st \smodels \sst$.
Since $abort$ is not an outcome for \mac, none of the rules in Fig.~\ref{fig:mac-error} were applied and \(\plan(\ssubst,P)\)
 does not fail.  The proof is by induction on the structure of $P$. 

\begin{description}
\item[(\(P\) is pure)] Then $\plan(\ssubst,P)=(P,outs)$.

The interesting case is for $outs\neq []$. The rule applied is:

\begin{mathpar}
\inferrule*[left=(Pure)]{
 P \text{ is pure } \\ outs=[(\lvar{x}_i, \lexp_i)|_{i=1}^n]\\\\
 \ssubst'=\ssubst \uplus \{(\lvar{x}_i \mapsto \ssubst (\lexp_i))|_{i=1}^n\} \\
 \consPure(\macEX,\spc,\ssubst'(P)) = \spc' }
 {\macMP( \macEX, [(P, outs)], \ssubst ,\sst)\rightsquigarrow ( \ssubst', \sst[\sstupdate{pc}{\spc'}]) }
\end{mathpar}

Then, $\sst_f=\sst[\sstupdate{pc}{\spc'}]=(\ssto, \smem,\sps, \spc')$.
Since  \(m=\macEX\) then   $\spc=\spc'$ and the result follows trivially.

\item[(\(P\) is spatial)] $P=\lexp_a\mapsto \lexp_v$

The interesting case happens for $\lexp_v\neq \cfreed.$

Then,  the rule [CellPos]  from Fig.~\ref{fig:mac-success} was applied 

\begin{mathpar}
\inferrule*[left=(CellPos)]
{\consPure(\macEX, \spc, \ssubst(\lexp_a)\in \Nat) =\spc' \\
\conscell(\ssubst(\lexp_a), \sst[\sstupdate{pc}{\spc'}]) \rightsquigarrow  (
      \sval,\sst')\\\\
      \ssubst_{subst}=\{\mathsf{O}\mapsto \sval\} \\
      outs=[(\lvar{x}_i, \lexp_i)|_{i=1}^n]  \\
    ((\ssubst\uplus \ssubst_{subst})(E_i)=\hat v_i)|_{i=1}^n
  \\
     \ssubst'=\ssubst \uplus \{(x_i\mapsto \hat v_i)|_{i=1}^n\} \\
       \consPure(\macEX, (\sst').\fieldsst{pc}, \ssubst' (\lexp_v)=\sval)= \spc'''}
{\macMP(\macEX, [(\lexp_a\mapsto \lexp_v, outs)],  \ssubst,  \sst) 
      \rightsquigarrow  (\ssubst', \sst'[ \sstupdate{pc}{\spc'''}])}
\end{mathpar}

where 
\begin{mathpar}
  \mprset{flushleft}
      \inferrule{\smem=\smem_f \uplus \{\sval_1 \mapsto \sval\}\\
        \spc''=\spc'\wedge (\ssubst(\lexp_a) = \sval_1) \\ \sat(\spc'')}
      {\conscell( \ssubst(\lexp_a),\sst[\sstupdate{pc}{\spc'}])
      \rightsquigarrow ({\sval}, \underbrace{\sst[\sstupdate{heap}{\smem_f}, \sstupdate{pc}{\spc''}]}_{\sst'})}
  \end{mathpar}

 Note that $\spc'=\spc$  and $ \neg(\sat(\spc\wedge \neg(\ssubst'(\lexp_a)\in \Nat))$
  by definition of $\consPure$.
  From $\conscell$, we have  $\spc''=\spc\wedge \ssubst(\lexp_a)=\sval_1$. From the last $\consPure$ call,  $\neg(\sat(\spc'' \wedge \neg(\ssubst'(\lexp_v)=\sval))$ and $\spc'''=\spc''$.
  By soundness of \mac, we have \hyp{1}~$\smem=\smem_f\uplus \smem_P$, 
  where $\smem_P=\{\sval_1\mapsto \sval\}$.

By definition  of $\smodels$:

\[
  \begin{aligned} 
 \mbox{\hyp{2}}~
 \sint, (\sto, \hp)\smodels \sst & \iff \exists \hp_1,\hp_2.~\hp=\hp_1\uplus \hp_2 \wedge \sint(\ssto)=\sto
  \wedge \sint(\smem)=\hp_1\wedge \sint(\spc)=\true \\
  & \qquad \wedge \sint, (\sto, \hp_2)\smodels \sps\\
  \end{aligned}
\]
Take $\sst_f=(\ssto, \smem_f, \sps,\spc'')$. 
We will show that  $\sint,(\sto, \hp_f\uplus \hp_2)\smodels \sst_f$. In fact, from \hyp{1} and \hyp{2}, it follows that there exist $\hp_P,\hp_f$ such that 
$\hp_1=\hp_f\uplus \hp_P$ and \hyp{3}~$\sint(\smem_f)=\hp_f\wedge 
  \sint(\smem_P)=\hp_P$. Notice that there exists at least one branch (in \conscell) such that $\sint(\spc'')=\sint(\spc\wedge \ssubst'(\lexp_a)=\sval_1)=\true$, otherwise we would have $\ssubst'(\lexp_a)\notin \dom(\smem)$ and \conscell, therefore \mac, would $abort$, and this contradicts our initial hypothesis. Therefore,

\[
\begin{aligned}
 \exists \hp_f,\hp_2.~\dish{\hp_2}{\hp_f} \wedge \sint(\smem_f)=\hp_f \wedge \sint(\spc'')=\true\wedge \sint, (\sto, \hp_2)\smodels \sps
  \end{aligned}
\]
 and the result follows for 
$\st_f=(\sto,\hp_f\uplus \hp_2)$, i.e., $\sint,\st_f\smodels \sst_f$.

\item[\bf (\(P\) is a predicate assertion)] $P=\Predd{\expredin}{\expredout}$

This case follows similarly to the previous case.
\end{description}
\end{proof}

\begin{theorem}[Property~\ref{prop:mac-ux-comp}]Completeness of UX \mac.
\[
\begin{array}{l}
\mac(\macUX, A, \ssubst, \sst)\rightsquigarrow (\ssubst', \sst') 
 \gand 
 \sint(\ssubst'), (\emptyset ,h_A)\models A \wedge \dish{\cmem'}{\hp_A} \gand \sint, (\sto, \cmem') \smodels \sst' \implies \\ 
\qquad
 \sint, (\emptyset, \hp_A) \smodels \sst_A  \gand \sint, (\sto, \hp' \uplus \hp_A) \smodels \sst 
\end{array}
\]
\end{theorem}

\begin{proof} By induction on the structure of $A$.

Let  $\sst=(\ssto,\smem,\sps,\spc)$. We have the following: 

\begin{enumerate}
\item $\smem=\smem_A\uplus \smem_f$ and  $\sps=\sps_A\uplus \sps_f$ (soundness of UX \mac)
\item $\sst'.\fieldsst{pc}=\spc\wedge \spc'$, for some $\spc'$ (definition of UX \mac)
\end{enumerate}

The remaining of the proof is done after the planning phase, considering the \mac of the \mps~obtained from $A$, which we will denote as $\mps(A)$. 

Thus, we analyse the rules for  $\macMP(\macUX, [\mps(A)], \ssubst, \sst)$.
\begin{description}
\item[ ($A$ is pure)]
\item[ ($A$ is an spatial assertion)] $A=\lexp_a\mapsto \lexp_v$.

Then, the rule applied is:
\begin{mathpar}
\infer{
\macMP(\macUX, [(\lexp_a\mapsto \lexp_v, outs)],  \ssubst,  \sst) 
       \rightsquigarrow  (\ssubst', \sst_f[ \sstupdate{pc}{\spc'''}])
}
{
\begin{array}{lr}
\consPure(\macUX, \spc, \ssubst(\lexp_a)\in \Val) =\spc'& \text{check for evaluation error} \\
\sst[\sstupdate{pc}{\spc'}].\conscell(\ssubst(\lexp_f)) \rightsquigarrow  (
       \sval,\sst_a) &  \text{consume cell}\\
        \ssubst_{subst}=\{\mathsf{O}\mapsto \sval\} &   \text{ substitution with cell contents}\\
       outs=[(\lvar{x}_i, \lexp_i)|_{i=1}^n]  &  \text{collect outs}\\
     ((\ssubst\uplus \ssubst_{subst})(E_i)=\hat v_i)|_{i=1}^n
     &  \text{instantiate outs}\\
      \ssubst'=\ssubst \uplus \{(x_i\mapsto \hat v_i)|_{i=1}^n\} &   \text{extend substitution with outs}\ \\
        \consPure(\macUX, (\sst_f).\fieldsst{pc}, \ssubst' (\lexp_v)=\sval)= \spc''' &   \text{consume  cell contents}
\end{array}
}
\end{mathpar}
and, in this case,  $\sst'=\sst_f[ \sstupdate{pc}{\spc'''}]$.

The definition of $\consPure$ implies 
\begin{enumerate}
\setcounter{enumi}{2}
\item $\exists\sval_a,\sval,\smem_f.~\smem=\smem_f\uplus \{\sval_a\mapsto \sval\}$
\item $\sst'.\fieldsst{pc}=\spc'''=\spc\wedge (\ssubst(\lexp_a)\in\Val)\wedge (\ssubst(\lexp_a)=\sval_a)\wedge (\ssubst'(\lexp_v)=\sval)$
\item $\sat(\spc''')$
\item well-formedness of $\sst$ implies that $\spc\models \sval_a\notin dom(\smem_f)$
\end{enumerate}
Thus, $\sst'=(\ssto,\smem_f,\sps,\spc''')$, $\sps_A=\emptyset$ 
and we can
define $\sst_A=(\emptyset,\{\sval_a\mapsto \sval\},\emptyset,\spc''')$.

Let $\sint,\hp_A$ and $\hp'$ be such that 
\begin{enumerate}
\setcounter{enumi}{6}
\item $ \sint(\ssubst'), (\emptyset ,h_A)\models \lexp_a\mapsto \lexp_v$;
\item $ \dish{\cmem'}{\hp_A} $
\item $ \sint, (\sto, \cmem') \smodels \sst' $
\item By definition of $\smodels$:

\[
\begin{aligned}
\sint, (\sto, \cmem') \smodels \sst' \Leftrightarrow \exists \hp_1',\hp_2'.~&\hp'=\hp'_1\uplus\hp'_2\wedge \sint(\smem_f)=\hp_1' \wedge \sint(\ssto)=\sto\\
&\wedge \sint(\spc''')=\true \wedge \sint, (\sto,\hp_2')\smodels \sps
\end{aligned}
\]
\item By definition of $\models$:

\[
\begin{aligned}
 \sint(\ssubst'), (\emptyset ,h_A)\models \lexp_a\mapsto \lexp_v &\Leftrightarrow \hp_A=\{\sint(\ssubst'(\lexp_a))\mapsto \sint(\ssubst'(\lexp_v))\}\\
 & \Leftrightarrow \hp_A=\{\sint(\sval_a)\mapsto \sint(\sval_v)\}
\end{aligned}
\]
\end{enumerate}

From $\sint(\spc''')=\true$, it follows that 
$ \sint, (\emptyset, \hp_A) \smodels \sst_A $, by definition.

From  (8) $\dish{\hp'}{\hp_A}$, in combination with (4), (10) and (11)  we have 
\[
\begin{aligned}
   \sint, (\sto, \hp' \uplus \hp_A) \smodels \sst &\iff 
   \sint, (\sto, \hp' \uplus \hp_A) \smodels (\ssto, 
   \smem_f\uplus \smem_A,\sps, \spc) \\
   & \iff  \exists \hp_1',\hp_2'.~\hp'=\hp'_1\uplus\hp'_2\wedge
    \sint(\smem_f)=\hp_1' \wedge \sint(\smem_A)=\hp_A\wedge \sint(\ssto)=\sto\\
   & \qquad \wedge \sint(\spc)=\true \wedge \sint, (\sto,\hp_2')\smodels \sps
\end{aligned}
\]
and the result follows.
\end{description}
\end{proof}

\begin{example}[Consume Computation of Ex. \ref{ex:fcall-ux}]\label{ex:fcall-ux-complete}
Suppose the symbolic execution is in a state $\sst=(\ssto,\smemb,\spc)$  with the initial substitution is  $\ssubst=\{50/c,1/x\}$ and $P=x\mapsto v$ that we want to consume. Let  $\smemb=(\{\svar{x}\mapsto \svar{c},\svar{y}\mapsto 1, 3\mapsto 5\},\emptyset)$ be the symbolic resource, and $\spc=\svar{c}\geq 42\land\svar{x}\neq \svar{y}\land \svar{x},\svar{y}\in \Nat$ be  symbolic path condition  (the symbolic store $\ssto$ is irrelevant to this computation and left opaque). 
We will use our consume implementation:

We use the ${\sf CellPos}$,  in Fig.~\ref{fig:mac-success}, that requires to check:
\begin{itemize}
\item $\consPure(\macUX,\spc,\ssubst(x)\in\Val)=\spc \land 1\in \Val$.
\item conscell branches: 
\begin{description}
\item[(branch 1)] $\conscell(1,\sst[ pc:=\spc\land 1\in \Val])\rightsquigarrow (\svar{c},\sst')$, where $\sst'=\sst[heap:= \smem_f,pc:= \spc'']$, with $\smem=\smem_f\uplus \{\svar{x}\mapsto \svar{c}\}$, $\spc''=\spc'\land (\svar{x}=1)$ and $\sat(\spc'')$.
\item[(branch 2)] $\conscell(1,\sst[ pc:=\spc\land 1\in \Val])\rightsquigarrow (1,\sst',pc:= \spc''])$, where $\sst'=\sst[heap:= \smem_f',pc:= \spc'']$ with $\smem=\smem_f'\uplus \{\svar{y}\mapsto 1\}$, $\spc''=\spc'\land (\svar{y}=1)$ and $\sat(\spc'')$.
\end{description}
\item $\ssubst_{subst}$ for each branch:
\begin{description}
\item[(branch 1)] $\ssubst_{subst}=\{\svar{c}/{\sf O}\}$ and {\bf (branch 2)}~  $\ssubst_{subst}=\{1/{\sf O}\}$
\end{description}
\item $outs=[(v,{\sf O}]$ 
\item $\ssubst'$ for each branch:
\begin{description}
\item[(branch 1)]  $\ssubst'=\ssubst\uplus \{\svar{c}/v\}$ and {\bf (branch 2)}  $\ssubst'=\ssubst\uplus \{1/v\}$

\end{description}
\item $\consPure(\macUX,(\sst').pc,\ssubst'(v)=\svar{v})$ for each branch:
\begin{description}
\item[(branch 1)]  $\consPure(\macUX,(\sst').pc,\svar{c}=\svar{c})=\spc'\land (\svar{x}=1) $

\item[(branch 2)]  $\consPure(\macUX,(\sst').pc,1=1)=\spc'\land (\svar{y}=1)$
\end{description}
\end{itemize}
Thus, 
\[
\begin{aligned}
\macMP(\macUX, [(x\mapsto v, outs)],  \ssubst,  \sst) 
      & \rightsquigarrow  (\ssubst', \sst'[ \sstupdate{pc}{\spc'\land \svar{x}=1}]) \qquad {\bf (branch 1)}\\
        & \rightsquigarrow  (\ssubst', \sst'[ \sstupdate{pc}{\spc'\land \svar{y}=1}])\qquad  {\bf (branch 2)}\\
  \end{aligned}
       \]
The computation of produce is in Ex.~\ref{app:ex-produce}.

\end{example} 

\begin{example}[Ilustrating Properties]\label{ex:illustrate_props} Consider the  consume step in Example~\ref{ex:fcall-ux}:
\[\mac(\macUX, x\mapsto v, \ssubst, \sst)\rightsquigarrow (\ssubst \cup \{ \svar{c}/v\}, (\ssto, (\{\svar{y}\mapsto 1, 3\mapsto 5\},\emptyset),\spc\land \svar{x}=1))\]
 The substitution is
 $\ssubst'=\ssubst \cup \{\svar{c}/v\}$, the assertion is  $P=x\mapsto v$ and the consumed resource is $\sst_P=(\emptyset, (\{\svar{x}\mapsto \svar{c}\},\emptyset) ,\spc\land \svar{x}=1)$.
For Property~\ref{prop:soundness}, consider  the decomposition of resource $\smemb$ into the parts $\smemb_P=(\{\svar{x}\mapsto \svar{c}\},\emptyset)$ and $\smemb_f=(\{\svar{y}\mapsto 1, 3\mapsto 5\},\emptyset)$. Take the interpretation $\sint=\{1/\svar{x} , 3/\svar{y}, 42/\svar{c}\}$.
By definition of  interpretation of symbolic states (\S\ref{sub:funspec}), $\sint,(\emptyset, \{1\mapsto 42\})\models \sst_P$ holds  if and only if $\sint((\emptyset, \{\svar{x}\mapsto \svar{c}\}))=(\emptyset, \{1\mapsto 42\})$ and $\sint( \spc\wedge \svar{x}=1)=\true$, and  to see that the latter holds, just apply the $\sint$ to the symbolic variables and verify the equality.  Thus, $\sint,(\emptyset, \{1\mapsto 42\})$ is a model of $\sst_P$. Also, note that $\sint(\ssubst')=\{1/x, 5/c, 42/v\}$ and  $\sint(\ssubst'),(\emptyset, \{1\mapsto 42\})\models P$ holds. Therefore,  $\sint(\ssubst'),(\emptyset, \{1\mapsto 42\})$ is a model of $P$.

 For Property~\ref{prop:mac-ux-comp}, consider an interpretation $\sint$ such that $\sint(\spc \wedge \svar{x}=1)=\true$ and $\sint(\svar{c})= 50$. In addition, consider the concrete state $(\emptyset, \{1\mapsto 50\})$.  Then $\sint(\ssubst')=\{1/x,50/c,50/v\}$, and by definition of satisfaction (App.~\ref{appssec:furtherdef}),  the judgement  $\sint(\ssubst'),(\emptyset, \{1\mapsto 50\})\models x\mapsto v$ holds. Thus, the premises of Property~\ref{prop:mac-ux-comp} are satisfied.  The following holds $\sint,  (\emptyset, \{1\mapsto 50\})\models (\emptyset, (\{\svar{x}\mapsto \svar{c}\},\emptyset) ,\spc \land \svar{x}=1)$ iff $\sint((\emptyset, \{\svar{x}\mapsto \svar{c}\}, \spc\wedge \svar{x}=1))=  (\emptyset, \{1\mapsto 50\}) $ as $ \sint(\spc \land \svar{x}=1)=\true$. 
\end{example}

\subsection{Rules for \produce}\label{app:prod_rules}

Fig.~\ref{fig:produce} contains the rules for \produce.

\begin{figure}[!h]
\footnotesize

\begin{mathpar}
\inferrule
{P ~\text{pure}\\ \spc' \defeq \spc \wedge  \ssubst(P) \quad \sat(\spc') }
{\funcAlg{produce}{P,\ssubst, \sst}\rightsquigarrow \sst[\sstupdate{pc}{\spc'}]}
\and
\inferrule
{\spc' \defeq \spc \wedge \ssubst(\lexp_a)\in \Nat\wedge \ssubst(\lexp_v)\in \Val \\\\
\prodcell(\ssubst(\lexp_a), \ssubst(\lexp_v),\sst[\sstupdate{pc}{\spc'}]) =\sst'}
{\produce(\lexp_a\mapsto \lexp_v, \ssubst, \sst) \rightsquigarrow \sst'}
\and
\inferrule
{\spc' \defeq \spc \wedge \ssubst(\lexp_a)\in \Nat \\\\
\prodcell(\ssubst(\lexp_a), \cfreed, \sst[\sstupdate{pc}{\spc'}]) =\sst'}
{\produce(\lexp_a\mapsto \cfreed, \ssubst, \sst) \rightsquigarrow \sst'}
\and
\inferrule[]{\produce(P_1, \ssubst,\sst)\rightsquigarrow \sst'\\\\
\produce(P_2, \ssubst, \sst')\rightsquigarrow \sst''}{\produce(P_1\lstar P_2, \ssubst, \sst)\rightsquigarrow \sst'}
\and
\inferrule
{\sps' \defeq \Predd{\ssubst(\vec{\lexp}_{ins})}{\subst(\vec{\lexp}_{outs})}\cup \sps \\\\
\spc' \defeq \spc\wedge \ssubst(\vec{\lexp}_{ins})\subseteq \Val \wedge \ssubst(\vec{\lexp}_{outs})\subseteq \Val \\
\sat(\spc')}
{\produce(\Predd{\vec{\lexp}_{ins}}{\vec{\lexp}_{outs}}, \ssubst,\sst)\rightsquigarrow \sst[\sstupdate{pred}{\sps'}, \sstupdate{pc}{\spc'}]}
\and
\inferrule[]{\produce(P_1, \ssubst,\sst)\rightsquigarrow \sst'}{\produce(P_1\lor P_2, \ssubst,\sst)\rightsquigarrow \sst'}
\and
\inferrule[]
{\produce(P_2, \ssubst,\sst)\rightsquigarrow \sst'}{\produce(P_1\lor P_2, \ssubst,\sst)\rightsquigarrow \sst'}
\end{mathpar}	
\caption{\produce rules}\label{fig:produce}
\end{figure}

\begin{example}[Produce computation of Ex.~\ref{ex:fcall-ux}]\label{app:ex-produce}
In this example, we give the complete computation of the produce step used in Ex.~\ref{ex:fcall-ux}. Consider the symbolic state $\sst_f=(\ssto,(\{\svar{y}\mapsto 1, 3\mapsto 5\},\emptyset ), \spc\land \svar{x}=1)$, with the notations  $\smemb_f$ for the symbolic heap and  $\spc_f$ for the symbolic path condition.  The substitution is  $\ssubst''=\ssubst'\cup \{1/x, 50/c, \svar{c}/v, \svar{r}/r\}$, with  $r,\svar{r}$ are fresh  variables. The post-condition to be produced is $\Qok''=x \mapsto c \lstar c \ge 42 \lstar r = v$. First, we apply the rule to produce cell assertions, which adds to the path condition $\spc'=\spc_f\land \ssubst'(x)\in \Nat \land \ssubst'(c)\in \Nat$ and check that it is $\sat$. Second, we apply $\prodcell$ rule which checks that $1\notin \dom(\smemb_f)$, updates the path condition to $\spc''=\spc'\land 1\notin \dom(\smemb_f)\land 50\in \Val$ and extends the symbolic heap with $\ssubst'(x)\mapsto \ssubst'(c)$ obtaining a new symbolic state $(\ssto,\smemb_f\cup (\{1\mapsto 50\},\emptyset),\spc'')$. Finally, we have to produce the pure assertions $c\geq 42 \lstar r=v$ and these only add to the symbolic path condition $\spc'''=\spc''\land \ssubst(c)\geq 42\land \ssubst'(r)=\ssubst'(v)$ which is the same as $\spc'''=\spc''\land 50\geq 42\land \svar{r}=\svar{c}$ 

The produce step for {\bf (branch 2)} is computed analogously.


\end{example}

\subsection{Correctness of \produce}\label{sec:produce-correct}
Suppose
\[
 \funcAlg{produce}{Q, \ssubst, (\ssto,\smem,\sps,\spc)}\rightsquigarrow
   (\ssto,\smem',\sps', \spc')
   \]
 and assume 
\begin{description}
\item[(A1)]  $\sinv((\ssto, \smem, \sps, \spc),\ssubst)$;
\item[(A3)] $\ssubst$ covers $P$.
\end{description}

%

\begin{theorem}[Property~\ref{prop:wf} for \({\tt produce}\)]
\(\sinv((\ssto, \smem', \sps', \spc'))
 \)
\end{theorem}

\begin{proof} By induction on the structure of $Q$ and also assuming 
 \(\sinv((\ssto,\smem, \sps, \spc), \ssubst)\) which we recall below:
\begin{align*}
\sinv((\ssto, \smem, \sps, \spc),\ssubst)&\defeq \sinv((\ssto, \smem, \sps, \spc))\wedge \svs{\ssubst}\subseteq \svs{\spc} \gand \spc \models \sinvc(\ssubst)\\
 \sinv((\ssto, \smem, \sps, \spc)) & \defeq (\svs{\ssto}\cup \svs{\sps} \cup \svs{\smem}\subseteq \svs{\spc})\wedge \sat(\spc)\wedge \spc \models (\sinvc(\ssto) \wedge \sinvc(\smem))
\end{align*} 
There are some cases to consider:
\begin{description}
  \item[(\(Q\) is pure)] \hfill
  
  In this case, the following rule from Fig.~\ref{fig:produce}  was applied
  
  \begin{prooftree}
\AxiomC{\(P\text{ is pure} \quad \spc'=(\spc\wedge \ssubst P) \quad \sat(\spc')\)}
\UnaryInfC{\(\produce(P,\ssubst, \sst)\rightsquigarrow \sst[\sstupdate{pc}{\spc'}]\)}
  \end{prooftree}
  

Notice that
\begin{description}
\item[\hyp{1}] \(\svs{\ssto}\cup \svs{\smem}\cup\svs{\sps}\subseteq \svs{\spc'}\).

A consequence of  the \(\sinv((\ssto,\smem, \sps, \spc))\) hypothesis and the fact that $\svs{\spc}\subseteq\svs{\spc'}$.
\item[\hyp{2}]\( \sat(\spc')\).

Trivially. 

\item[\hyp{3}]\(\spc' \models (\sinvc(\ssto)\wedge \sinvc(\smem))\)
\end{description}
  Now, let $\sint$ be such that \( \vint(\spc')=\true\). Then, $\vint(\spc)=\true$ and 
   
  $\vint(\spc')=\vint(\spc\gand Q\ssubst)=\true$. This trivially gives $\vint(\spc)=\true$,
$\sinv(\ssto,\smem,\sps,\spc)$ implies $\sinvc(\ssto)\wedge \sinvc(\smem)$.

From \hyp{1}, \hyp{2} and \hyp{3} one has $\sinv(\ssto,\smem,\sps,\spc')$, and the result follows.
  \item[(\(Q\) is spatial assertion)]\(Q=\lexp_1\mapsto \lexp_2\)  (The case \(Q=E_1\mapsto \cfreed \) is analogous
)\hfill
  
    In this case, the following rule from Fig.~\ref{fig:produce}  was applied

\begin{mathpar}
\inferrule
{\spc' \defeq \spc \wedge \ssubst(\lexp_a)\in \Nat\wedge \ssubst(\lexp_v)\in \Val \\\\
\prodcell(\ssubst(\lexp_a), \ssubst(\lexp_v),\sst[\sstupdate{pc}{\spc'}]) =\sst'}
{\produce(\lexp_a\mapsto \lexp_v, \ssubst, \sst) \rightsquigarrow \sst'}  
\end{mathpar}

  Let $\sst = (\ssto, \smem, \sps, \spc)$. The definition of  \(\prodcell \) gives:
\begin{description}
\item[\hyp{0}] $\sst' = (\ssto, \smem', \sps, \spc'')$
\item[\hyp{1}] $\smem'=\{\ssubst(\lexp_1)\mapsto \ssubst(\lexp_2))\} \uplus \smem$;
\item[\hyp{2}] $\spc''=\spc'\wedge (\ssubst(\lexp_1)\notin \dom(\ssubst))$; and 
\item[\hyp{3}] \(\sat(\spc'')\).
\end{description}

Notice that
\begin{description}
\item[\hyp{4}] \(\svs{\ssto}\cup \svs{\smem'} \cup \svs{\sps}\subseteq \svs{\spc''}\).

From \hyp{1} \(\svs{\smem'}=\svs{\smem}\cup \svs{\{\ssubst(\lexp_1)\mapsto \ssubst(\lexp_2)\}}\subseteq \svs{\spc''}\).
The \(\sinv(\sst,\ssubst)\) hypothesis implies that $\svs{\ssubst}\cup\svs{\smem}\cup\svs{\sps}\subseteq \svs{\spc}$. In combination with the fact that  $\svs{\spc}\subseteq\svs{\spc''}$, the result follows.
\item[\hyp{5}]\( \sat(\spc'')\) 

Trivially, from the definition of \hyp{3}.
\item[\hyp{6}]\(\spc'' \models (\sinvc(\ssto)\wedge \sinvc(\smem'))\)
\end{description}
  Now, let $\sint$ be such that \( \vint(\spc'')=\true\) i.e., such that  
  \[ \vint(\spc\wedge  \ssubst(\lexp_1)\in \Nat\wedge \ssubst(\lexp_2)\in \Val\wedge (\ssubst(\lexp_1)\notin \dom(\ssubst)))=\true.\] 
  Then, $\vint(\spc)=\true$ and \(\sinv(\sst)\) gives $\spc\models (\sinvc(\ssto)\wedge \sinvc(\smem))$, which entails \(\sinvc(\ssto)\wedge \sinvc(\smem)\). However, to prove $\sinvc(\smem')=\sinvc(\smem\uplus\{\ssubst(\lexp_1)\mapsto \ssubst(\lexp_2)\})$ we need to verify 
  \[\dom(\smem')\subseteq \Nat \wedge \codom(\smem')\subseteq \Val \wedge (\forall v_i,v_j\in \dom(\smem').i\neq j\implies v_i\neq v_j)\]
   which clearly follows from $\sint(\spc'')=\true$.
   
From \hyp{4}, \hyp{5} and \hyp{6} one has $\sinv(\sst')$, and the result follows.

\item[(Q is a predicate assertion)] $Q = \Predd{\vec{\lexp}_{ins}}{\vec{\lexp}_{outs}}$
\begin{mathpar}
\inferrule
{\sps' \defeq \Predd{\ssubst(\vec{\lexp}_{ins})}{\subst(\vec{\lexp}_{outs})}\cup \sps \\\\
\spc' \defeq \spc\wedge \ssubst(\vec{\lexp}_{ins})\subseteq \Val \wedge \ssubst(\vec{\lexp}_{outs})\subseteq \Val \\
\sat(\spc')}
{\produce(\Predd{\vec{\lexp}_{ins}}{\vec{\lexp}_{outs}}, \ssubst,\sst)\rightsquigarrow \sst[\sstupdate{pred}{\sps'}, \sstupdate{pc}{\spc'}]}  
\end{mathpar}

Let $\sst = (\ssto, \smem, \sps, \spc)$ and $\sst' = (\ssto, \smem, \sps', \spc')$. 

Notice that
\begin{description}
\item[\hyp{1}] $\svs{\ssto} \cup \svs{\smem} \cup \svs{\sps} \subseteq \svs{\spc'}$

From the rule, $\svs{\sps'} = \svs{\sps} \cup \svs{\ssubst(\vec{\lexp}_{ins})} \cup \svs{\ssubst(\vec{\lexp}_{outs})} \subseteq \svs{\spc'}$. In addition, the $\sinv(\sst, \ssubst)$ hypothesis implies that $\svs{\ssubst}\cup\svs{\smem}\cup \svs{\ssto} \subseteq \svs{spc}$. Results hence follows from $\svs{\spc}\subseteq\svs{\spc'}$.

\item[\hyp{2}] $\sat(\spc')$

Follows trivially from the rule

\item[\hyp{3}] $\spc' \models (\sinvc(\ssto) \land \sinvc(\smem))$

Follows trivially from the fact that $\spc'$ is stronger than $\spc$ and $\spc \models (\sinvc(\ssto) \land \sinvc(\smem))$ from the hypothesis $\sinv(\sst)$.

From \hyp{1-3}, we immediately have the result.
\end{description}

\item[(Inductive Step)] $Q=Q_1\lstar \ldots \lstar Q_n$, where each $Q_i$ is a simple assertion.
   
 To simplify the reasoning, consider the case $n=2$. The general case is similar. 
 
     In this case, the following rule from Fig.~\ref{fig:produce}  was applied

\begin{mathpar}
\inferrule{\produce( Q_1, \ssubst, \sst)\rightsquigarrow \sst'\\\\
\produce(Q_2, \ssubst, \sst')\rightsquigarrow \sst''}{
\produce(Q_1\lstar Q_2, \ssubst, \sst)\rightsquigarrow \sst''}
\end{mathpar}

From \(\produce( Q_1, \ssubst,\sst)\rightsquigarrow \sst'\) and $\sinv(\sst,\ssubst)$ we obtain from inductive hypothesis $\sinv(\sst')$. Then, by applying the inductive hypothesis again, we also get that $\sinv(\sst', ssubst)$
\end{description}
\end{proof}

\begin{theorem}[Property~\ref{prop:soundness}]Soundness for \produce. 
\[
\begin{array}{l}
\produce(Q, \ssubst, (\ssto, \smem, \sps, \spc)) \rightsquigarrow(\ssto,\smem', \sps', \spc')  \implies \\
\qquad \exists \smem_q, \sps_q.~ \smem'=\smem_q\uplus \smem \wedge \sps' = \sps \cup \sps_q \land (\forall\vint, \st.~ \sint, \st \models (\ssto, \smem_q, \sps_q, \spc') \implies \sint(\ssubst), \st \models Q)
\end{array}
\]
\end{theorem}
\begin{proof}
  First, we recall the assumption of well-formedness  \(\sinv((\ssto, \smem, \sps, \spc), \ssubst)\) and the fact that  Prop.~\ref{prop:wf} implies 
    \(\sinv(\ssto, \smem', \sps', \spc')\), and consequently $\sat(\spc')$.
   
    The proof follows by induction on the structure of \(Q\).


  \begin{description}
    \item[(\(Q\) is pure)] \hfill
    
%
    In this case,
     \(\funcAlg{produce}{Q,\ssubst, (\ssto, \smem, \sps, \spc)}\rightsquigarrow (\ssto, \smem, \sps, \spc')\)
    where \(\spc'=\spc\wedge \ssubst(Q)\) and \(\sat(\spc')\).  Also,    \({\tt produce}\) does not change the symbolic heap or predicates, i.e., \(\smem=\smem' \land \sps = \sps'\) and  we can take 
    \(\smem_q=\sps_q=\emptyset.\) 

   Let $\sint, \st$ such that $\sint, \st \models (\ssto, \emptyset, \emptyset, \spc')$. then $\st = (\sint(\ssto), \emptyset)$ and $\sint(\spc') = \true$. Therefore $\sint(\ssubst(Q)) = \true$, by definition of $\spc$, meaning $\sint(\ssubst), (\sint(\ssto), \emptyset) \models Q$, proving what the property.
    \item[(\(Q\) is spatial)] \(Q= \lexp_1\mapsto \lexp_2\) (The case \(Q= \lexp \mapsto \cfreed\) is analogous)

        In this case, the following rule from Fig.~\ref{fig:produce}  was applied

\begin{mathpar}
\inferrule
{\spc' \defeq \spc \wedge \ssubst(\lexp_a)\in \Nat\wedge \ssubst(\lexp_v)\in \Val \\\\
\prodcell(\ssubst(\lexp_a), \ssubst(\lexp_v),\sst[\sstupdate{pc}{\spc'}]) =\sst'}
{\produce(\lexp_a\mapsto \lexp_v, \ssubst, \sst) \rightsquigarrow \sst'}  
\end{mathpar}

 Let $(\ssto, \smem, \sps, \spc) = \sst$ and  $(\ssto, \smem', \sps, \spc'') = \sst'$
 
   The definition of  \(\prodcell \)
gives \hyp{1} $\smem'=\{\ssubst(\lexp_1)\mapsto \ssubst(\lexp_2))\}\uplus \smem$; \hyp{2} $\spc''=\spc'\wedge (\ssubst(\lexp_1)\notin \dom(\ssubst))$; and \hyp{3} \(\sat(\spc'')\). Take $\smem_q=\{\ssubst(\lexp_1)\mapsto \ssubst(\lexp_2))\}$.

   Let   $\sint$  such that $\sint(\spc')=\true$. Then, using well-formedness, $\sint, (\sint(\ssto), \sint(\smem_q)) \models (\ssto,\smem_q,\emptyset, \spc'')$ is a well-defined concrete state. By definition,

   \begin{align*}
    \sint(\ssubst), (\sint(\ssto), \sint(\smem_q)) \models E_1\mapsto E_2 \iff  \sint(\smem_q) =\{\esem{\lexp_1}{\sint(\ssubst)}\mapsto \esem{\lexp_2}{\sint(\ssubst)}\}
    \end{align*}
     which trivially holds.
     
\item[(Q is a predicate assertion)] $Q=\Predd{\vec{\lexp}_{ins}}{\vec{\lexp}_{outs}}$.

\begin{mathpar}
\inferrule
{\sps' \defeq \Predd{\ssubst(\vec{\lexp}_{ins})}{\subst(\vec{\lexp}_{outs})}\cup \sps \\\\
\spc' \defeq \spc\wedge \ssubst(\vec{\lexp}_{ins})\subseteq \Val \wedge \ssubst(\vec{\lexp}_{outs})\subseteq \Val \\
\sat(\spc')}
{\produce(\Predd{\vec{\lexp}_{ins}}{\vec{\lexp}_{outs}}, \ssubst,\sst)\rightsquigarrow \sst[\sstupdate{pred}{\sps'}, \sstupdate{pc}{\spc'}]}
\end{mathpar}

Let $(\ssto, \smem, \sps, \spc) = \sst$ and $(\ssto, \smem, \sps', \spc') = \sst'$.

We naturally pick $\sps_q = \{ \Predd{\ssubst(\vec{\lexp}_{ins})}{\ssubst(\vec{\lexp}_{outs}) } \}$ and $\smem_q = \emptyset$.
     
Let $\sint, \st$ such that $\sint, \st \models (\ssto, \emptyset, \sps_q, \spc')$, then by definition of satisfaction, we have that $\st = (\sto, \hp)$ where $\sto = \sint(\ssto)$ and $\sint, (\emptyset, \hp) \models \sps_q$. Therefore $\sint, \st \models \Predd{\ssubst(\vec{\lexp}_{ins})}{\ssubst(\vec{\lexp}_{outs})}$, giving $\sint(\ssubst), \st \models \Predd{\vec{\lexp}_{ins}}{\vec{\lexp}_{outs}}$, proving the required property.
     
\item[(Inductive Step)] \(Q=Q_1 \lstar Q_2\)
    
    This case follows by induction hypothesis.
  \end{description}
%
%
\end{proof}


\begin{theorem}[Property~\ref{prop:prod-compl}]Completeness of \produce.
\[
\begin{array}{l}
\sint, (\sto, \hp) \models \sst \land \sint(\ssubst), (\emptyset, \hp_p) \models P \land \dish{\hp}{\hp_p}\\
\qquad \implies \exists \sst_p. \produce(P, \ssubst, \sst) \rightsquigarrow \sst \cdot \sst_p \land \sint, (\emptyset, \hp_p) \models \sst_p)
\end{array}
\]
\end{theorem}
\begin{proof}
The proof is by induction on the structure of \(P\).
\begin{description}
\item[(\(P\) is pure)] \hfill

Let  $\sst$, $\sint$ and $\ssubst$ be such that 
\begin{enumerate}
\item $\sst=(\ssto,\smem,\sps,\spc)$;
\item $\sint, (\sto, \hp) \models \sst$ (and hence, $ \sint(\spc)=\true$);
\item  $\dish{\hp}{\hp_p}$;
\item $\sint(\ssubst),(\emptyset, \hp_p)\models P$.
\end{enumerate}

By definition of satisfaction,  
\begin{align*}
\sint(\ssubst), (\emptyset ,h_p)\models P &\iff \esem{P}{\sint(\ssubst)}{}=\true \wedge \emp\\
& \iff \sint(\ssubst(P))=\true \wedge \hp_p=\emptyset.
\end{align*}

 Thus,  $\sint(\spc\wedge \ssubst(P))=\true$ and we  can apply the rule for produce from Fig.~\ref{fig:produce}:

    \begin{prooftree}
    \AxiomC{$\spc'=(\spc\wedge \ssubst(P)) \quad \sat(\spc')$}
    \UnaryInfC{$\funcAlg{produce}{P,\ssubst, \sst}\rightsquigarrow \sst[\sstupdate{pc}{\spc'}]$}
  \end{prooftree}
Notice that  $(\ssto,\smem, \sps,\spc')=(\ssto,\smem,\sps, \spc) \cdot (\emptyset, \emptyset, \emptyset, \ssubst(P))$ and the result follows trivially.
\item[(\(P\) is a spatial assertion)] \(P=\lexp_1\mapsto \lexp_2\)

Let  $\sst$, $\sint$ and $\ssubst$ be such that 
\begin{enumerate}
\item $\sst=(\ssto,\smem,\sps,\spc)$;
\item  $ \dish{\hp}{\hp_p}$;
\item $\sint(\ssubst),(\emptyset, \hp_p)\models \lexp_1\mapsto \lexp_2$.
\end{enumerate}

By definition of satisfaction,  
\begin{align*}
\sint(\ssubst), (\emptyset,h_p)\models \lexp_1\mapsto \lexp_2 &\iff \hp_p=\{\esem{\lexp_1}{\sint(\ssubst)}{} \mapsto \esem{\lexp_2}{\sint(\ssubst)}{}\}
\end{align*}
By definition of heaps, it follows that  $ \sint(\ssubst(\lexp_1))\in \Nat $ and $\sint(\ssubst(\lexp_2))\in \Val$.

Thus,  $\sint(\spc\wedge \ssubst(\lexp_1)\in \Nat \wedge \ssubst(\lexp_2))\in \Val)=\true$ and we  can apply the rule for produce fromFig.~\ref{fig:produce}:
\begin{mathpar}
\inferrule
{\spc' \defeq \spc \wedge \ssubst(\lexp_a)\in \Nat\wedge \ssubst(\lexp_v)\in \Val \\\\
\prodcell(\ssubst(\lexp_a), \ssubst(\lexp_v),\sst[\sstupdate{pc}{\spc'}]) =\sst'}
{\produce(\lexp_a\mapsto \lexp_v, \ssubst, \sst) \rightsquigarrow \sst'}  
\end{mathpar}
prodcell goes through successfuly by the hypothesis that $ \dish{\hp}{\hp_p} $. Take $\smem_p=\{\ssubst(\lexp_1)\mapsto \ssubst(\lexp_2)\}$, $\sst' = (\ssto, \smem', \spc'')$.

Notice  that \((\ssto,\smem',\spc'')=(\ssto, \smem,\spc)\cdot (\emptyset,\smem_p,\spc'')\), where \(\spc'=\ssubst(\lexp_1)\in \Nat\wedge \ssubst(\lexp_2)\in \Val\wedge  \ssubst (\lexp_1) \notin \dom(\smem)\). Also, 
\(\sint, (\sto, \hp_p) \models ((\emptyset,\smem_p,\emptyset, \spc''))\) and the result follows.

\item[(P is a predicate assertion)] $P = \Predd{\vec{\lexp}_{ins}}{\vec{\lexp}_{outs}}$

This case is the most trivial and does not even require the disjointness of heaps in the pre-condition.
$\sat(\spc \land \ssubst(\vec{\lexp}_{ins}) \subseteq \vals \land \ssubst(\vec{\lexp}_{ins}) \subseteq \vals)$ is obtained from the $\sint(\ssubst), (\emptyset, \hp_p) \models P$, and $\sst_p$ is trivially constructed as $\{ \Predd{\ssubst(\vec{\lexp}_{ins})}{\ssubst(\vec{\lexp}_{outs})} \}$ which satisfies the property by definition of the satisfaction.

\item[(Inductive Step)] \(P=P_1\lstar P_2\) and $P=P_1\vee P_2$

Follow immediately from the induction hypothesis.
\end{description}
\end{proof}

%% file: sections/app-biabduction.tex

\section{Bi-abduction: Semantics and Correctness}%
\label{app:bi-abduction-correct}

This appendix provides a soundness proof of our bi-abduction work.

\subsection{Semantics}%
\label{app:biabd-semantics}

For readability, we work with modified rules rather than the catch-fix-continue approach introduced in the main text. That is, bi-abduction can also be introduced by modifying the existing missing-resource rules of the engine (including the rules for UX \mac) -- this requires more invasive surgery but is easier to follow for the proofs.

To illustrate, e.g., consider the rule \textsc{Lookup-Err-Missing} from the symbolic semantics of the core engine. The bi-abductive version of the rule now instead of faulting extends the heap with the missing cell with a fresh symbolic variable as its value and records this cell in the anti-heap:
\[
\inferrule[\textsc{Biab-Lookup}]
 {\cseeval{\pexp}{\ssto}{\spc}{\sval_1}{\spc'} \quad
 \spc'' \defeq \spc' \land \sval_1 \in \nats \land \sval_1 \not\in \domain(\smem) \quad
 \sat(\spc'') \quad \sval_2 \text{ fresh}}
 {\biab{(\ssto, \smem, \spc)}{\pderef{\pvar{x}}{\pexp}}{ ((\ssto[\pvar{x} \mapsto \sval_2], \smem[\sval_1 \mapsto \sval_2], \spc'' \land \sval_2 \in \vals), \{ \sval_1 \mapsto \sval_2 \}) }{\fsctx}{\osucc}}
\]

Moreover, to simplify the presentation, we consider an engine without support for predicates; fixing missing-resource errors arising from missing predicates are analogous to fixing missing-resource errors arising from missing heap resources.

All in all, in this setting, the bi-abduction semantics uses a new judgement of the following form:
\[
	\biab{\sst}{\scmd}{(\sst', \smem)}{\fsctx}{\outcome}
\]
The judgement is similar to the judgements presented in App.~\ref{app:symbolic}, but additionally returns an anti-heap.

The rules for $\macMP$ and $\conscell$ are also extended to work appropriately in $\biexmode$ mode, with the following judgements:
\[
\begin{array}{c}
	\sst.\conscell(\biexmode, \sval) = ((\sval', \sst'), \smem)\\
	\macMP(\biexmode, (P_i, outs_i)|_{i=1}^n, \ssubst, \sst) \rightsquigarrow ((\ssubst', \sst'), \smem)
\end{array}
\]

Below, we provide the list of rules that significantly change compared to the semantics presented in App.~\ref{app:symbolic}. Rules we do not provide are generally the same as previously, but now return an additional empty anti-heap $\emptyset$.
{\footnotesize
\begin{mathpar}
\inferrule[\textsc{Biab-Lookup}]
 {\cseeval{\pexp}{\ssto}{\spc}{\sval_1}{\spc'} \quad\quad \sval_2 \text{ fresh}\\\\
 \spc'' \defeq \spc' \land \sval_1 \in \nats \land \sval_1 \not\in \domain(\smem) \quad
 \sat(\spc'') }
 {\biab{(\ssto, \smem, \spc)}{\pderef{\pvar{x}}{\pexp}}{ ((\ssto[\pvar{x} \mapsto \sval_2], \smem[\sval_1 \mapsto \sval_2], \spc'' \land \sval_2 \in \vals), \{ \sval_1 \mapsto \sval_2 \}) }{\fsctx}{\osucc}}
\and
\inferrule[\textsc{Biab-Mutate}]
 {\cseeval{\pexp_1}{\ssto}{\spc}{\sval_1}{\spc'} \quad
 \cseeval{\pexp_2}{\ssto}{\spc'}{\sval_2}{\spc''} \quad \sval_3 \text{ fresh} \\\\
 \spc''' \defeq \spc'' \land \sval_1 \in \nats \land \sval_1 \not\in \domain(\smem)\wedge \sval_3 \in \Val\quad
 \sat(\spc''') }
 {\biab{(\ssto, \smem, \spc)}{\pmutate{\pexp_1}{\pexp_2}}{ ((\ssto, \smem[\sval_1 \mapsto \sval_2], \spc'''), \{ \sval_1 \mapsto \sval_3 \}) }{\fsctx}{\osucc}}
\and
\inferrule[\textsc{Biab-Seq}]
{\biab{\sst}{\scmd_1}{(\sst', \smem_1)}{\fsctx}{\osucc} \quad \biab{\sst'}{\scmd_2}{(\sst'', \smem_2)}{\fsctx}{\result} \quad \svs{\domain(\smem_2)} \cap (\svs{\sst'} \setminus \svs{\sst}) = \emptyset}
{\biab{\sst}{\scmd_1; \scmd_2}{(\sst'', \smem_1 \uplus \smem_2)}{\fsctx}{\result}}
\end{mathpar}
\begin{mathpar}
\inferrule[\textsc{Biab-ConsCell}]
{\spc' \defeq \spc \wedge \sval \in \nats \land \sval \notin \dom(\smem)  \quad \sat(\spc')\\\\
\sval' \text{  fresh}\quad
\spc'' = \spc' \land \sval' \in \vals}
{(\ssto, \smem, \spc).\conscell(\biexmode, \sval) = ((\sym{v}', (\ssto, \smem, \spc'')), \{ \sval \mapsto \sval' \})}
\and
\inferrule[\textsc{Biab-MacMP-Seq}]{\macMP(\biexmode,(P,outs), \ssubst, \sst)\rightsquigarrow ((\ssubst', \sst'), \smem_1)\\\\
\macMP(\biexmode, ps,\ssubst', \sst')\rightsquigarrow ((\ssubst'', \sst''), \smem_2)}
{\macMP(\biexmode, (P,outs)::ps,\ssubst, \sst)\rightsquigarrow ((\ssubst'', \sst''), \smem_1 \uplus \smem_2)}
\end{mathpar}
}
\subsection{Soundness Proofs}

To prove Thm.~\ref{thm:bi-abduction}, we need the following lemma:

\begin{lemma}[Symbolic variables of codomains of anti-heaps]\label{thm:antiheap-codom-sv}
If $\biab{\sst}{\scmd}{(\sst', \smem)}{\fsctx}{\outcome}$, then $\svs{\codomain(\smem)} \subseteq (\svs{\sst'} \setminus \svs{\sst})$; which moreover implies $\svs{\codomain(\smem)} \cap \svs{\sst} = \emptyset$.
\end{lemma}

%
We also need the following frame result:
\begin{theorem}\label{thm:ux-symbolic-frame}
\[
\begin{array}{l}
\csesemtransabstractm{\sst}{\cmd}{\sst'}{\fsctx}{m}{\result} \land ( \svs{\smem_f} \cap (\svs{\sst'} \setminus \svs\sst) = \emptyset) \implies \csesemtransabstractm{\sst \cdot \smem_f}{\cmd}{\sst' \cdot \smem_f}{\fsctx}{m}{\result}
\end{array}
\]
\end{theorem}

We now present the main proof:

\begin{proof}[Proof of Thm.~\ref{thm:bi-abduction}]
The proof is by induction. We show some  illustrative case for fixes.
\begin{description}
\item[(rule {\sc Biab-Lookup})]  
From the application of the rule, we have: 

\begin{mathpar}
\inferrule[\textsc{Biab-Lookup}]
 {\cseeval{\pexp}{\ssto}{\spc}{\sval_1}{\spc'} \quad
 \spc'' = \spc' \land \sval_1 \in \nats \land \sval_1 \not\in \domain(\smem) \quad
 \sat(\spc'') \quad \sval_2 \text{ fresh}}
 {\biab{(\ssto, \smem, \spc)}{\pderef{\pvar{x}}{\pexp}}{ ((\ssto[\pvar{x} \mapsto \sval_2], \smem[\sval_1 \mapsto \sval_2], \spc'' \land \sval_2 \in \vals), \{ \sval_1 \mapsto \sval_2 \}) }{\fsctx}{\osucc}}
\end{mathpar}
Define  $\spc'' = \sval_1 \in \nats \land \sval_1 \not\in \domain(\smem) \land \spc'$, $\sst' = (\ssto[\pvar{x} \mapsto \sval_2], \smem[\sval_1 \mapsto \sval_2], \spc'')$ and  $\smem = \{ \sval_1 \mapsto \sval_2 \}$.

Notice that now {\sc Lookup} with $\sst\cdot \smem=(\ssto, \smem[\sval_1 \mapsto \sval_2], \spc'')$:
\begin{mathpar}
\inferrule[\textsc{Lookup}]
 {\cseeval{\pexp}{\ssto}{\spc}{\sval_1}{\spc'} \quad \smem(\sval_1) = \sval_2 \quad \sat(\spc'')} 
 {\csesemtrans{\ssto, \smem[\sval_1 \mapsto \sval_2], \spc''}{\pderef{\pvar{x}}{\pexp}}{ \ssto[\pvar{x} \mapsto \sval_2], \smem[\sval_1 \mapsto \sval_2], \spc''}{\fsctx}{\osucc}}
\end{mathpar}
and the result follows.


\item[(rule {\sc Biab-Mutate})] 
From the application of the rule we have:
\begin{mathpar}
\inferrule[\textsc{Biab-Mutate}]
 {\cseeval{\pexp_1}{\ssto}{\spc}{\sval_1}{\spc'} \quad
 \cseeval{\pexp_2}{\ssto}{\spc'}{\sval_2}{\spc''} \quad \sval_3 \text{ fresh} \\\\
 \spc''' = \spc'' \land \sval_1 \in \nats \land \sval_1 \not\in \domain(\smem) \wedge \sval_3\in\Val\qquad
 \sat(\spc''') }
 {\biab{(\ssto, \smem, \spc)}{\pmutate{\pexp_1}{\pexp_2}}{ ((\ssto, \smem[\sval_1 \mapsto \sval_2], \spc'''), \{ \sval_1 \mapsto \sval_3 \}) }{\fsctx}{\osucc}}
 \end{mathpar}
 
 Now we can apply the rule {\sc Mutate} with $\sst\cdot [\sval_1\mapsto \sval_3]=(\ssto,\smem[\sval_1\mapsto \sval_3],\spc^* )$ where \(\spc^*=\spc\wedge \sval_1 \in \nats \land \sval_1 \not\in \domain(\smem) \wedge \sval_3\in\Val\)
 \begin{mathpar}
 \inferrule[\textsc{Mutate}]
 {\cseeval{\pexp_1}{\ssto}{\spc^*}{\sval_1}{\spc'}\quad \cseeval{\pexp_2}{\ssto}{\spc'}{\sval_2}{\spc''}  \quad \smem(\sval_1) = \sval_3 \quad  \\\\ \sat(\spc'') \qquad \smem' = \smem[\sval_1 \mapsto \sval_2]  }
 {\csesemtrans{\ssto, \smem [\sval_1\mapsto \sval_3],  \spc^*}{\pmutate{\pexp_1}{\pexp_2}}{\ssto, \smem',\spc''}{\fsctx}{\osucc}}
 \end{mathpar}
 and $\spc'''=\spc''\wedge  \sval_1 \in \nats \land \sval_1 \not\in \domain(\smem) \wedge \sval_3\in\Val$, which gives the result.

\item[(rule {\sc Biab-Seq})] 
From the application of the rule, we have:

\begin{mathpar}
\inferrule[\textsc{Biab-Seq}]
{\biab{\sst}{\scmd_1}{(\sst', \smem_1)}{\fsctx}{\osucc} \qquad \biab{\sst'}{\scmd_2}{(\sst'', \smem_2)}{\fsctx}{\result}\\\\
 \svs{\domain(\smem_2)} \cap (\svs{\sst'} \setminus \svs{\sst}) = \emptyset}
{\biab{\sst}{\scmd_1; \scmd_2}{(\sst'', \smem_1 \uplus \smem_2)}{\fsctx}{\result}}
\end{mathpar}

By the induction hypothesis,  $\csesemtransabstractm{\sst \cdot \smem_1}{\scmd_1}{\sst'}{\fsctx}{\macUX}{\oxok}$ and $\csesemtransabstractm{\sst' \cdot \smem_2}{\scmd_2}{\sst''}{\fsctx}{\macUX}{\outcome}$ where $\smem_1$ and $\smem_2$ are the antiframes for $\scmd_1$ and $\scmd_2$, respectively. 

From $\csesemtransabstractm{\sst \cdot \smem_1}{\scmd_1}{\sst'}{\fsctx}{\macUX}{\oxok}$, Lem.~\ref{thm:antiheap-codom-sv}, and Thm.~\ref{thm:ux-symbolic-frame}, we get $\csesemtransabstractm{\sst \cdot \smem_1 \cdot \smem_2}{\scmd_1}{\sst' \cdot \smem_2}{\fsctx}{\macUX}{\oxok}$. 

Now we can apply rule {\sc Seq}:

\begin{mathpar}
\inferrule[\textsc{Seq}]
{\csesemtransabstractm{\sst \cdot \smem_1\cdot \smem_2}{\scmd_1}{\sst'\cdot \smem_2}{\fsctx}{\macUX}{\oxok} \\\\ \csesemtransabstractm{\sst' \cdot \smem_2}{\scmd_2}{\sst''}{\fsctx}{\macUX}{\outcome}}
{\csesemtransabstractm{\sst \cdot \smem_1\cdot \smem_2}{\scmd_1}{\sst''}{\fsctx}{\macUX}{\result}}
\end{mathpar}
and the result follows.
\end{description}
\end{proof}

%% file: sections/app-applications-proofs.tex

\section{Applications: Correctness}\label{app:applications-proofs}

The following lemma is convenient for our proofs, it, like Lem.~\ref{lem:valid-with-frame}, follows easily from the fact that the concrete language is compositional, this time with respect to stores:
\begin{lemma}[Specification Validity with Restrictions]\label{lem:frameless-validity}
\[
\begin{array}{l}
\fictx \models \quadruple{P}{\cmd}{\Qok}{\Qerr} \iff \\
\quad (\forall \subst, \sto, \hp, \outcome, \sto', \hp'.~\dom(\subst)\subseteq\lv{P} \land \pv{\sto} \subseteq \pv{\scmd} \land \\
\qquad \subst, (\sto, \hp) \models P \land (\sto, \hp), \cmd \baction_{\fictx} \outcome: (\sto', \hp') \implies \\
\qquad (\outcome \neq \oxm \land \subst, (\sto', \hp') \models Q_\outcome)) \\
\fictx \models \islquadruple{P}{\cmd}{\Qok}{\Qerr} \iff \\
\quad (\forall \subst, \sto', \hp', \outcome.~\dom(\subst)\subseteq\lv{Q_\outcome} \land \pv{\sto'} \subseteq \pv{\scmd} \land \subst, (\sto', \hp') \models Q_\outcome \implies \\
    \quad \quad (\exsts{\sto,\hp}~\subst, (\sto, \hp) \models P~\land~(\sto, \hp), \cmd \baction_\fictx \outcome: (\sto', \hp')))
\end{array}
\]
\end{lemma}

\subsection{Correctness of \texttt{verifyOX}}

\begin{theorem}[Correctness of \texttt{verifyOX}]
Given a function $\pfunction{f}{\vec{\pvar{x}}}{\scmd; \preturn{\pexp}}$, a function specification $s = \quadruple{\vec{\pvar{x}} = \vec{x} \lstar P}{f(\vec{\pvar{x}})}{\Qok}{\Qerr}$, and an OX valid environment $\models (\fictx, \fsctx)$, where $f \not\in \domain(\fictx)$: let $\fsctx' = \fsctx[f \mapsto s]$, if $\texttt{verifyOX}(\fsctx, f, s)$, then $\models (\fictx', \fsctx')$, where $\fictx' = \fictx[f \mapsto (\vec{\pvar{x}}, \scmd, \pexp)]$.
\end{theorem}

\begin{proof}
    %
    We assume some $\fictx$, $\fsctx$, $\fictx'$, $\fsctx'$, $f$, $\pvvar x$, $\cmdf$, $\pexp$, $t$, $P$, $\Qok$, $\Qerr$ such that 
    \begin{description}
        \item[(A1)] $t= \quadruple{\pvvar x\doteq\vec x\lstar P}{f(\pvvar x)}{\Qok}{\Qerr}$ is an external function specification
        \item[(A2)] $f\notin\dom(\fictx)$
        \item[(A3)] $\models(\fictx,\fsctx)$
        \item[(A4)] $\fictx'=\fictx[f\mapsto (\pvvar x, \cmdf, \pexp)]$
        \item[(A5)] $\fsctx'=\fsctx[f\mapsto t]$
    \end{description}
and we aim to show that $\models(\fictx',\fsctx')$, i.e.
\[
\begin{array}{l}
 \dom(\fsctx')\subseteq\dom(\fictx')~\land \\
 (\forall \fid, \pvvar x, \cmd, \pexp, \sspec.~f(\vec{\pvar{x}})\{\cmd; \preturn{\pexp}\} \in \fictx' \land \sspec \in \fsctx'(\fid) \implies \exists \sspec' \in \fext_{\fictx', \fid}(\sspec).~\fictx'\models \cmd:\sspec')
\end{array}
\]
(A3-5) imply $\dom(\fsctx')\subseteq\dom(\fictx')$ and for all functions that are already in $\fictx$, the statement follows straightforwardly from (A3). The only non-trivial case is the newly added function $f$ and its specification $t$. (A1) and the definition of Internalisation imply that any internal specification of $t$ has the precondition $\pvvar x\doteq\vec x\lstar P \lstar \pvvar z\doteq\nil$, where $\pvar x$ are the arguments of $f$ as defined in (A4) and $\pvvar z$ are the remaining program variables of $\cmd$.

We therefore assume some $\subst, \sto, \hp, \sto'$ and $\hp'$ such that 
    \begin{description}
        \item[(H1)] $\subst, \sto, \hp \models \pvvar x\doteq\vec x \lstar P \lstar \pvvar z\doteq\nil$, where 
        \begin{description}
            \item[(H1a)] $\pv{P}= \emptyset$
            \item[(H1b)] $\lv{P}=\vec x$
            \item[(H1c)] $\pvvar z= \pv{\cmd}\backslash \{\pvvar x\}$
            \item[(H1d)] $\pv{s}\subseteq\pv{\cmd}$
            \item[(H1e)] $\dom(\subst)\subseteq\pv{\pvvar x\doteq\vec x\lstar P}$. 
        \end{description}
        \item[(H2)] $(\sto,\hp), \cmd \Downarrow_{\fictx'} \outcome: (\sto', \hp')$
    \end{description}
To avoid some clutter later, we also define \textbf{(H3a)}~$\st=(\sto,\hp)$ and \textbf{(H3b)}~$\st'=(\sto',\hp')$.
%
%

The execution of $\texttt{verifyOX}(\fsctx, f, s)$ implies the following definitions and properties:
\begin{description}
    \item[(V1)] Definine $\ssubst$: $\ssubst(\vec x)=\vec{\hat x}$ with $\vec{\hat x}$ being distinct symbolic variables, else undefined.
    \item[(V2)] Define $\ssto$: $\ssto(\pvvar x)=\vec{\hat x}$, $\ssto(\pvvar z)=\nil$, else undefined.
    \item[(V3)] Define $\smem=\emptyset$.
    \item[(V4)] Define $\spc=\true$. 
    \item[(V5)] Define $\sst=(\ssto, \smem, \spc)$
    %
    %
\end{description}
%
%
Next, we define an interpretation $\sint$: \textbf{(H3a)} $\sint(\vec{\hat x})=\sto(\pvvar x)$. 
(V2), (H1) and (H1d) imply \textbf{(H3b)} $\sint(\ssto)=\sto$. (H1), (H1e) and (H3) furthermore imply \textbf{(H3c)} $\sint(\ssubst)=\subst$.  
(H1) and (H1a) implies \textbf{(H4)} $\subst, (\emptyset,\hp) \models P$.  (H3c) implies \textbf{(H5)} $\sint(\ssubst), (\emptyset,\hp) \models P$. \\
\textit{Property~\ref{prop:prod-compl}} of \texttt{produce} implies the more specific statement
\begin{multline*}
    \frall{\sint, \ssubst, \ssto, \hp, P} \sint(\ssubst), (\emptyset, \hp) \models P \implies \\
    \exsts{\smem',\spc'} \produce(\ssubst, (\ssto, \smem, \spc))\rightsquigarrow (\ssubst'(\ssto, \smem', \spc')) ~\land~ \sint(\smem')=\hp ~\land~ \sint(\spc')=\true
\end{multline*}
This implies with (H5) (and because \texttt{verifyOX} explores all paths):
\begin{description}
    \item[(H6)] $\sint(\smem')=\hp$
    \item[(H7)] $\sint(\spc')=\true$   
\end{description}
Defining $\sst'=(\ssto,\hp',\spc')$, (H3b), (H6) and (H7) yield \textbf{(H8)}~$\sint(\sst')=\st$. \\
As \texttt{verifyOX} covers all paths and it's successful execution implies that step 3 never yields the outcomes $\omiss$ or {\it abort}, Thm.~\ref{thm:ux-ox-sound-func} together with (A3), (H2), (H3b) and (H8) yields the existence of some $\sint'\geq\sint$ and $\sst''$ such that 
\[
\textbf{(H9)}~ \sst', \cmd, \pvar{ret}:=\pexp \Downarrow_{\fsctx}^{\macOX} \outcome: \sst'' ~\land~ \sint'(\sst''\backslash\{\pvar{ret}\})=(\sto',\hp')  
\]
(H9) implies \textbf{(H10)}~$\sint'(\sst''\mathtt{.pc})=\true$. \\\\
\textbf{Case: $\outcome=\osucc$.} \\
Let $r$ be a fresh logical variables and define \textbf{(V9)} $\ssubst'=\ssubst[r\mapsto \sst''\mathtt{.sto}(\pvar{ret})]$. 
Step 5 implies that \textbf{(V14)} $abort\notin\mac(\macEX, Q_\osucc[r/\pvar{ret}], \ssubst', \sst'')$
Step 3 and \textit{Property~\ref{prop:mac-branch-comp}} imply the existence of some $\ssubst''\geq\ssubst'$ and $\sst'''$ such that
\begin{description}
    \item[(V10)] $\mac(\macEX, Q_\osucc[r/\pvar{ret}], \ssubst', \sst'') \rightsquigarrow (\ssubst'', \sst''')$
    \item[(V11)] $\sst'''\mathtt{.sto} = \sst''\mathtt{.sto}$
    \item[(V12)] $\sst'''\mathtt{.hp} = \emptyset$ 
    \item[(V13)] $\sint'(\sst'''\mathtt{.pc}) = \true$
\end{description}
(V12) implies vacuously that the new heap $\sst'''\mathtt{.hp}$ is covered by $\sint'$, and without loss of generality, we assume that the new substitution $\ssubst''$ is covered as well, i.e. \textbf{(H11)}~$ \sst'''\mathtt{.hp}, \ssubst'' \subseteq \sint'$. \\
(V10), (V12), (V13) and \textit{Property~\ref{prop:soundness}} imply
\[
\textbf{(H12a)}~\sint'(\ssubst''), \sint'(\sst''\mathtt{.sto}, \sst''\mathtt{.hp}) \models Q_\osucc[r/\pvar{ret}]
\]
(H12a) and (H9) imply
imply
\[
\textbf{(H12b)}~\sint'(\ssubst''), (\sto'[\pvar{ret}\mapsto\esem{\pexp}{\sto'}], \hp') \models Q_\osucc[r/\pvar{ret}]
\]
and as $Q_\osucc[r/\pvar{ret}]$ does not hold the program variable $\pvar{ret}$, (H12b) yields 
\[
\textbf{(H12c)}~\sint'(\ssubst''), (\sto', \hp') \models Q_\osucc[r/\pvar{ret}]
\]
(V9), (V10), $\ssubst''\geq\ssubst'$ and \textit{Property~\ref{prop:wf}} imply that
\[
    ~\textbf{(H13a)}~
(\sint'(\ssubst''))(r)
= (\sint'(\ssubst'))(r)
= (\sint'(\ssubst'(r)))
= \sint'(\sst''\mathtt{.sto}(\pvar{ret}))
\]
(H9), the fact that $\pvar{ret}\notin\pv{\pexp}$ imply 
\[
    \textbf{(H13b)}~ \sst''\mathtt{.sto}(\pvar{ret})=\esem{\pexp}{\sint'(\sst''\mathtt{.sto}\backslash\{\pvar{ret}\})}=\esem{\pexp}{\sto'}
\]
(H12c) and (H13b) imply \textbf{(H14)}~$\sint'(\ssubst''), (\sto',\hp')\models Q_\osucc[r/\pvar{ret}] \lstar r\doteq \pexp$.\\
(H3c), (V9), $\ssubst''\geq\ssubst'$ and $\sint'\geq\sint$ implies that  $\sint'(\ssubst'')\rest{\dom(\subst)} = \subst$ and (A1) and (H12c) therefore yields
\[
\subst[r \mapsto \esem{\pexp}{\sto'}], (\sto',\hp')\models \Qok[r/\pvar{ret}]
\]
which implies 
\[
\textbf{(H15)}~\subst, (\sto',\hp')\models \Qok[\pexp/\pvar{ret}].
\]
We define $\Qok'\eqdef Q_\osucc[\pexp/\pvar{ret}]$ and show that it is an OX-internalisation of $\Qok$:
Let $\pvvar p=\pv{\Qok'}$, $\vec p$ some fresh logical variables (w.r.t. $\Qok'$), we show for arbitrary $\bar\subst, \bar\sto$ and $\bar\hp$: 
\begin{align*}
    &\bar\subst, (\bar\sto, \bar\hp)\models \exsts{\vec p} \Qok'[\vec p / \pvvar p]\lstar \pvar{ret}\doteq\pexp[\vec p / \pvvar p]  \\
    \Leftrightarrow~ &\bar\subst, (\bar\sto, \bar\hp)\models\exsts{\vec p}(\Qok[\pexp/\pvar{ret}])[\vec p / \pvvar p] \lstar \pvar{ret}\doteq \pexp[\vec p / \pvvar p] \\
    \Leftrightarrow~ &\bar\subst, (\bar\sto, \bar\hp)\models\exsts{\vec p}\Qok[(\pexp[\vec p / \pvvar p])/\pvar{ret}]\lstar \pvar{ret}\doteq \pexp[\vec p / \pvvar p] \\
    \Leftrightarrow~ &\exsts{\vec v\subset\Val}~\bar\subst[\vec p\mapsto\vec v], (\bar\sto, \bar\hp)\models \Qok[(\pexp[\vec p / \pvvar p])/\pvar{ret}]\lstar \pvar{ret}\doteq \pexp[\vec p / \pvvar p]\\
    \Leftrightarrow~ &\exsts{\vec v\subset\Val}~\bar\subst[\vec p\mapsto\vec v], (\bar\sto, \bar\hp)\models \Qok\lstar \pvar{ret}\doteq \pexp[\vec p / \pvvar p] \\
    \Rightarrow~ &\exsts{\vec v\subset\Val}~\bar\subst[\vec p\mapsto\vec v], (\bar\sto, \bar\hp)\models \Qok \\
    \Rightarrow~ &\bar\subst, (\bar\sto, \bar\hp)\models \Qok
\end{align*}
where the last implications holds because $\vec p$ are fresh w.r.t $\Qok$ (since they are fresh w.r.t. $\Qok'$).\\

\textbf{Case: $\outcome=\oerr$.} This case works analoguous to the $\oxok$ case. \\

Therefore, 
\[
\quadruple{\vec{\pvar{x}} = \vec{x} \lstar P \lstar \pvvar z \doteq\nil}{}{\Qok'}{\Qerr'} \in \fext_{\fictx', f}^{OX}(t).
\]
This yields the desired $\models(\fictx', \fsctx')$.

\end{proof}

\subsection{Correctness of \synthesise}

The $\toasrt$ function is correct in the following sense:
\begin{theorem}[Correctness of $\toasrt$]\label{thm:toAsrt}
\[
\toasrt(\sst) = P \implies \forall \vint, \sto, \cmem.~
(\vint(\sst) = \st \Rightarrow \sintpe{\vint}(\ssubstid), \st \models P) \land (\sintpe{\vint}(\ssubstid), (\sto, \cmem) \models P \Rightarrow \vint(\sst) = (\sto|_{\sst.\mathsf{sto}}, h))
\]
where: the \emph{identity substitution} $\ssubstid$, turning logical variables into symbolic variables, is injective and defined as $\ssubstid(x) = \hat{x}$, with $\codom(\ssubstid) = \dom(\vint)$; $\sst.\mathsf{sto}$ denotes the store of $\sst$; and $\sto|_{\ssto}$ denotes the store~$\sto$ with domain restricted to the domain of $\ssto$.
\end{theorem}

\begin{proof}[Proof of Thm.~\ref{thm:synthesise-correct}]
    $\synthesise(f, \pvvar x=\vec x \lstar P)$ with $f(\pvvar x)\{\cmd;\preturn{E}\}$ is guaranteed to end successfuly with a set of triples $\left\{(\outcome_i, \sst_i, \smem_i)\right\}$, which in the worst case scenario will simply be empty. The following assumptions form the hypothesis of our theorem:
    \begin{description}
        \item[(S1)] $\sst=(\ssto,\smem,\spc)$
        \begin{description}
            \item[(S1a)] $\ssubst(x)=\hat x$ for all $x\in\lv{\pvvar x\doteq\vec x \lstar P}$, else undefined
            \item[(S1b)] $\ssto(\pvvar x)=\vec{\hat x}$ and $\ssto(\pvvar z)=\nil$ where $\pvvar z = \pv{\cmd}\setminus\{\pvvar x\}$, else undefined
            \item[(S1c)] $\smem=\emptyset$
            \item[(S1d)] $\spc=\true$ 
        \end{description}
        \item[(S2)] $\produce(P,\ssubst,\sst)\rightsquigarrow\sst'$
        \item[(S3)] $\sst',\cmd\baction_{\fsctx}^{\biexmode} \hat\Sigma'$
        \item[(S4)] $\synthesise(\fsctx, \fid, P) = \left\{\isltriplex{\vec{\pvar{x}} = \vec{\lvar{x}} \lstar P_i}{\fid(\vec{\pvar{x}})}{\outcome}{\exists \vec{\lvar{y}}.~Q_i} \right\}$ such that:
        \begin{description}
        		\item[(S4a)] $(\outcome_i, (\ssto_i', \smem_i', \spc_i'), \smem_i')) \in \Sigma'$
        		\item[(S4b)] $P_i = P \lstar \toasrt((\emptyset, \smem_i, \true))$
			\item[(S4c)] $Q_i = \toasrt(\ssto_i'|_{\{\pvar{x}\}}, \smem_i', \spc_i')$ with $\pvar{x} = \pvar{ret}$ if $\outcome_i = \osucc$, $\pvar{x} = \pvar{err}$ otherwise.
			\item[(S4d)] $\vec{\lvar{y}} = \lv{Q_i} \setminus \lv{P_i}$
        \end{description}
    \end{description}

We start by showing that $\exists \vec{\lvar{y}}.~Q_i$ is a UX internalisation of $Q_i' = \toasrt(\ssto_i', \smem_i', \spc_i')$, i.e, that they satisfy the following internalisations conditions:
\begin{description}
        \item[(IC1)] $\pexp \in \vals \Leftarrow Q_i' \star \AssTrue$ if $\outcome_i = \osucc$
        \item[(IC2)] $Q_i \Leftarrow \exists \vec{\lvar{p}}. Q_i'[\vec{\lvar{p}}/\vec{\pvar{p}}] \star \pvar{ret} = \pexp{}[\vec{\lvar{p}}/\vec{\pvar{p}}]$ if $\outcome_i = \osucc$
        \item[(IC3)] $\Qerr \Leftarrow \exists \vec{\lvar{p}}. \Qerr'[\vec{\lvar{p}}/\vec{\pvar{p}}]$ if $\outcome_i = \oerr$
\end{description}

    \textit{Internalisation conditions 1-3.}\\
These properties, in the case $\outcome_i = \osucc$, naturally follow from Thm.~\ref{thm:toAsrt} as well as (S4), and the fact that $\pvar{ret} := \pexp$ is executed in each branch before the call to $\texttt{toAsrt}$.

We want to show that

\[
\isltriplex{P \star \vec{\pvar{x}} = \vec{\lvar{x}} \star \vec{\pvar{z}} = \nil}{C}{\outcome_i}{Q'_i}{}
\]

\textit{Case } $\outcome_i = \osucc$\\
We assume some $\subst$, $\sto'$ and $\hp'$ such that \textbf{(H3)}~$\subst,(\sto',\hp')\models Q_i'$. (S4) implies that there is some \textbf{(H4a)}~$(\outcome, \sst'_i, \smem_i')\in\hat\Sigma'$ such that \textbf{(H4b)}~$\subst,(\sto',\hp')\models\texttt{toAsrt}(\sst_i'|_{\{\pvar{ret}\}})$. Define $\sint$ as $\sint(x)=\subst(x)$, so by definition of $\ssubstid$, \textbf{(H5)}~$\sint(\ssubstid)=\subst$.

(S4), (H5) and Thm.~\ref{thm:toAsrt} imply that \textbf{(H6)}~$\sint(\sst'')=(\sto',\hp')$. Thm.~\ref{thm:ux-ox-sound-func}, Thm.~\ref{thm:bi-abduction}, (H4a) and (S3) imply \textbf{(H7)}~$\sint(\sst'),\cmd\baction_{\fictx}\oxok: (\sto', \hp')$.

Soundness of \texttt{produce}, (S1) and (S2) implies \textbf{(H8)}~$\subst,\sint(\sst')\models P$. 

(S1), (S2) and Prop.~\ref{prop:wf} of \texttt{produce} implies that $\sint(\sst')\texttt{.sto}=\sto$ and (S1b) then implies $\subst,\sint(\st')\models \pvar x\doteq\vec x\lstar P \lstar \pvvar z\doteq\nil$. This concludes the proof of the $\oxok$ case. \\

\textit{Err case.}\\
This case works fully analogously to the $\oxok$ case.

\end{proof}

%% file: sections/app-gillian.tex
\section{Gillian Evaluation}%
\label{app:gillian}


As explained in \S\ref{sec:evaluation}, we ran our freshly implemented UX bi-abduction analysis on the Collections-C library~\cite{collections}. In this appendix, we explain what fixes were implemented, provide examples of generated specifications, and detail limitations of the current implementation of Gillian-C, the instantiation of Gillian to the C programming language.

\subparagraph*{Implemented Fixes.}

To obtain our preliminary results, we only implement a subset of all possible fixes for missing errors in the Gillian-C memory model. In particular, we implemented the fixes corresponding to those presented in this paper: we add cells in memory when they are missing during a load, a store, or a $\conscell$.

While Gillian-C does feature a variety of other memory operations and consumers, each with their corresponding fixes, our current implementation has proven sufficient in generating a compelling number of interesting specifications.

\subparagraph*{Examples of Generated Specifications.}

We now show examples of specifications generated by our bi-abduction analysis, on a function of the \texttt{array} library. In Collections-C, an array structure contains a size and a capacity both of type \lstinline{size_t}. Then, it contains an expansion factor, which specifies how to increase the capacity of the array when it is full and we need to add a new element. Finally, it contains a pointer to an array of pointers to \lstinline[language=C]{void} -- an artefact of the lack of polymorphism in C.

The structure definition in the original implementation of the library also contained 3 additional function pointers for users of the library to provide custom implementations of \texttt{malloc}, \texttt{calloc} and \texttt{free}. We removed these pointers from the structure and replaced any calls to these functions to the \texttt{stdlib} versions. We performed similar changes to any other structure containing allocator function pointers.

We provide examples for 

\begin{minipage}{0.40\textwidth}
\begin{lstlisting}[language=C]
typedef struct array_s {
    size_t size;
    size_t capacity;
    float exp_factor;
    void **buffer;
    // Modified to remove
    // allocator pointers.
} Array;
\end{lstlisting}
\end{minipage}\hfill%
\begin{minipage}{0.57\textwidth}
\begin{lstlisting}[language=C]
cc_enum array_get_at(Array *ar, size_t index,
                   void **out)
{
    if (index >= ar->size)
        return CC_ERR_OUT_OF_BOUND;
    *out = ar->buffer[index];
    return CC_OK;
}
\end{lstlisting}
\end{minipage}\hfill%

Gillian is able to generate 9 specifications for this function, 4 of which correspond to successful executions, and 5 to erroneous executions. We detail 2 successful ones and 1 erroneous one, provided in Fig.~\ref{fig:gillian-app-specs}.

\newcommand{\facc}[1]{\texttt{.#1}}
\newcommand{\arraygetattcall}{\mathtt{
{\color{blue} cc\_enum}\ array\_get\_at({\color{blue} Array} *\! ar,\ {\color{blue} size\_t}\ index,\ {\color{blue} void} *\!*out)
}}
\newcommand{\NULL}{\mathtt{NULL}}

\begin{figure}[!ht]
\begin{tabular}{|l|c|}
\hline
Spec 1: Successful access &
$\bigisltripleok
  {{\color{gray} \pvar{ar} = a \lstar \pvar{index} = i \lstar \pvar{out} = o \lstar}\\
   a\facc{size} \mapsto s \lstar \\
   a\facc{buffer} \mapsto b \lstar (b + i) \mapsto v \lstar\\
   o \mapsto {\color{red} o'} \lstar {\color{violet} o' = \NULL}}
  {\arraygetattcall}
  {a\facc{size} \mapsto s \lstar {\color{red} i < s} \lstar \\
   a\facc{buffer} \mapsto b \lstar (b + i) \mapsto v \lstar\\
   o \mapsto {\color{red} v} \lstar {\color{violet} o' = \NULL} \lstar\\
   {\color{red} \pvar{ret} = \pvar{CC\_OK}}}$ \\ \hline 
Spec 2: Graceful out-of-bounds &
$\bigisltripleok
  {{\color{gray} \pvar{ar} = a \lstar \pvar{index} = i \lstar \pvar{out} = o \lstar}\\
   {\color{violet} o = \NULL} \lstar \\
   a\facc{size} \mapsto s}
  {\arraygetattcall}
  {{\color{violet} o = \NULL} \lstar \\
   a\facc{size} \mapsto s \lstar {\color{red} i \geq s} \lstar\\
   {\color{red}\pvar{ret} = \pvar{CC\_ERR\_OUT\_OF\_BOUNDS}}}$ \\ \hline
Spec 3: NULL dereference &
$\bigisltripleerr
  {{\color{gray} \pvar{ar} = a \lstar \pvar{index} = i \lstar \pvar{out} = o \lstar}\\
   a = \NULL \lstar o = \NULL}
  {\arraygetattcall}
  {a = \NULL \lstar o = \NULL \lstar \\
   {\color{red}\pvar{ret} = \mathtt{"segmentation\ fault"}}}$ \\ \hline
\end{tabular}
\caption{Examples of specifications generated by Gillian-C's bi-abduction}
\label{fig:gillian-app-specs}
\end{figure}

In all three specifications, we highlight in {\color{red} red} the fragments which are modified between the pre and the post condition, and we {\color{gray} grey out} the mandatory program variable assignments.

\begin{description}

\item[\textbf{Spec 1}] corresponds to a successful execution of that function: the size check passes, the accessed cell and the $\pvar{out}$ pointer are properly allocated. The content pointed to by the $\pvar{out}$ pointer is overridden and the function returns $\pvar{CC\_OK}$. Note the minimal footprint, only the \texttt{size} and \texttt{buffer} fields of the structure are required and therefore only those fields are bi-abduced.

\item[\textbf{Spec 2}] corresponds to a case where the \texttt{size} field of the array structure indicates that the index is out of bounds. The library gracefully handles this by returning an error code.

\item[\textbf{Spec 3}] corresponds to the case where the pointer given as input is a $\NULL$ pointer, which triggers a $\NULL$ dereference, which Gillian-C detects as an error. This spec is perhaps the most important, as it is the one which, when propagated through the codebase by function calls, would allow a front-end that can filter true bugs to detect an issue.

\end{description}

\subparagraph*{On Specification Duplication.}

Above, we only presented 2 out of the 4 successful specifications generated by Gillian-C. In addition, we highlighted in {\color{violet} purple} extraneous conditions bi-abduced by the Gillian-C engine, in the form $\_ = \NULL$. This is due to a current limitation resulting of a bad cocktail mixing the way Gillian-C encodes values in its symbolic heap and information lost in compilation from C to GIL\footnote{Gillian-C uses the \texttt{CSharpMinor} intermediate language of the CompCert formally verified compiler which has lost most of the typing information available at the source level}: bi-abducing the ``shape'' of the value (NULL, or not NULL) becomes necessary to preserve well-formedness of the memory.

This phenomenon leads to an explosion in the number of generated specifications, especially when many different memory cells containing pointers are accessed in a row. In particular, Collections-C exposes several iterator structures composed of many pointers, such as \lstinline[language=C]{slist_zip_iter}, \lstinline[language=C]{hashtable_iter} and \lstinline[language=C]{hashet_iter}. For each of these structure, the library also exposes an initialiser function which will assign each field one by one, leading to a path explosion.

Thankfully, this is purely a limitation of the current implementation of Gillian-C, and not of our theoretical framework for bi-abduction, nor of Gillian's implementation of bi-abduction. It could be overcome if we were to write our own instantiation of Gillian for C, using another compiler and improving the encoding of value in memory. Such work is already in progress, but outside the scope of this paper.